\documentclass[letter,11pt]{article}

\usepackage{hyperref}
\usepackage{xspace}
\usepackage{times}
\usepackage[margin=1in,footskip=0.3in]{geometry}
\usepackage{verbatim}
\usepackage{amssymb,amsmath,amsthm,epsfig}
\usepackage{enumitem}
\usepackage{graphicx,wrapfig,lipsum}
\usepackage[font=small]{caption}
\usepackage{subcaption}
\usepackage[english]{babel}
\usepackage[utf8]{inputenc}
\usepackage{algorithm}
\usepackage[noend]{algpseudocode}
\usepackage{url}
\usepackage{color}
\usepackage[all]{xy}
\usepackage[noadjust]{cite}
\usepackage{rotating}
\usepackage{ifpdf}
\usepackage{tcolorbox}
\usepackage{lipsum}
\usepackage{thm-restate}
\usepackage{multicol}

\newtheorem{theorem}{Theorem}
\newtheorem{lemma}{Lemma}
\newtheorem{remark}{Remark}

\newtheorem*{lemma*}{Lemma}
\newtheorem{corollary}{Corollary}
\newtheorem{problem}{Open Problem}

\newtheorem{claim}{Claim}

\newtheorem*{theorem*}{Theorem}
\newtheorem*{corollary*}{Corollary}

\newcommand{\opt}{opt}
\newcommand{\OPT}{OPT}

\newcommand{\cost}{\small{\mathsf{cost}}}
\newcommand{\dist}{dist}

\newcommand{\wO}{\widetilde{\text{OPT}}}
\newcommand{\wo}{\widetilde{\text{opt}}}
\newcommand{\calP}{\mathcal{P}}
\newcommand{\tF}{\tilde{F}}

\newcommand{\N}{\mathbb{N}}
\newcommand{\argmin}{\mathrm{argmin}}
\newcommand{\R}{\mathbb{R}}
\newcommand{\RP}{\mathbb{R}^{\geq 0}}

\newcommand{\poly}{\mathrm{poly}}

\usepackage{titlesec}
\usepackage{sidecap}

\newcommand{\E}{\mathbb{E}}

\newcommand{\eps}{\varepsilon}

\definecolor{Darkblue}{rgb}{0,0,0.4}
\definecolor{Brown}{cmyk}{0,0.61,1.,0.60}
\definecolor{Purple}{cmyk}{0.45,0.86,0,0}
\definecolor{brickred}{rgb}{0.8, 0.25, 0.33}

\newcommand{\open}{\texttt{open}}

\usepackage{authblk}

\ifdefined\DEBUG
	\newcommand{\fab}[1]{\textcolor{red}{#1}}
	\newcommand{\fabr}[1]{\rem{\textcolor{red}{$\bullet$ #1}}}

	\newcommand{\vca}[1]{\rem{\textcolor{blue}{$\bullet$ #1}}}

	\newcommand{\el}[1]{\rem{\textcolor{gray}{$\bullet$ #1}}}

	\newcommand{\chris}[1]{\rem{\textcolor{magenta}{$\bullet$ #1}}}
	
	\def\rem#1{{\marginpar{\raggedright\scriptsize #1}}}
\else
	\newcommand{\fab}[1]{#1}
	\newcommand{\fabr}[1]{}
	
	\newcommand{\vca}[1]{}
	
	\newcommand{\el}[1]{}
	 
	\newcommand{\chris}[1]{}
\fi

\title{Breaching the 2 LMP Approximation Barrier for Facility Location with Applications to $k$-Median}

\ifdefined\DEBUG
\author[1]{Vincent Cohen-Addad Viallat\thanks{\texttt{vcohenad@gmail.com}}}
\author[2]{Fabrizio Grandoni\thanks{\texttt{fabrizio@idsia.ch}. Partially supported by the SNSF Excellence Grant 200020B\_182865/1.}}
\author[3]{Euiwoong Lee\thanks{\texttt{euiwoong@umich.edu}}}
\author[4]{Chris Schwiegelshohn\thanks{\texttt{cschwiegelshohn@gmail.com}}}
\affil[1]{Google Zurich, Switzerland}
\affil[2]{IDSIA, Switzerland}
\affil[3]{University of Michigan, USA}
\affil[4]{Aarhus University, Denmark}
\else
\author[1]{Vincent Cohen-Addad Viallat\thanks{\texttt{vcohenad@gmail.com}}}
\author[2]{Fabrizio Grandoni\thanks{\texttt{fabrizio@idsia.ch}. Partially supported by the SNSF Excellence Grant 200020B\_182865/1.}}
\author[3]{Euiwoong Lee\thanks{\texttt{euiwoong@umich.edu}}}
\author[4]{Chris Schwiegelshohn\thanks{\texttt{cschwiegelshohn@gmail.com}}}
\affil[1]{Google Research, France}
\affil[2]{IDSIA, Switzerland}
\affil[3]{University of Michigan, USA}
\affil[4]{Aarhus University, Denmark}
\fi

\date{}
\begin{document}

\maketitle

\begin{abstract}

 \noindent The Uncapacitated Facility Location (UFL) problem is one of the most fundamental clustering problems: Given a set of clients $C$ and a set of facilities $F$ in a metric space $(C \cup F, dist)$ with facility costs $\open : F \to \R^+$, the goal is to find a set of facilities $S \subseteq F$ to minimize the sum of the opening cost $\open(S)$ and the connection cost $d(S) := \sum_{p \in C} \min_{c \in S} dist(p, c)$. 
An algorithm for UFL is called a Lagrangian Multiplier Preserving (LMP) $\alpha$ approximation if it outputs a solution $S\subseteq F$ satisfying $\open(S) + d(S) \leq \open(S^*) + \alpha d(S^*)$ for any $S^* \subseteq F$. 
The best-known LMP approximation ratio for UFL is at most $2$ by the JMS algorithm of Jain, Mahdian, and Saberi [STOC'02, J.ACM'03] based on the Dual-Fitting technique.
The lack of progress on improving the upper bound on $\alpha_{LMP}$ in the last two decades raised the natural question whether $\alpha_{LMP}=2$. 

We answer this question negatively by presenting a (slightly) improved LMP approximation algorithm for UFL. This is achieved by combining the Dual-Fitting technique with Local Search, another popular technique to address clustering problems. 
  In more detail, we use the LMP solution $S$ produced by JMS to seed a local search algorithm.
  We show that local search substantially improves $S$ unless a big fraction of the connection cost of $S$ is
  associated with facilities of relatively small opening costs. In the latter case however the
  analysis of Jain, Mahdian, and Saberi can be improved (i.e., $S$ is cheaper than expected). To summarize:
  Either $S$ is close enough to the optimum, or it must belong to the local neighborhood of a good enough local optimum. From a conceptual viewpoint, our result gives a theoretical evidence that local search can be enhanced so as to avoid bad local optima
  by choosing the initial feasible solution with LP-based techniques.

Our result directly implies a (slightly) improved approximation for the related $k$-Median problem, another fundamental clustering problem: Given $(C \cup F, dist)$ as in a UFL instance and an integer $k \in \mathbb{N}$, find $S \subseteq F$ with $|S| = k$ that minimizes $d(S)$. The current best approximation algorithms for $k$-Median are based on the following framework: use an LMP $\alpha$ approximation algorithm for UFL to build an $\alpha$ approximate \emph{bipoint solution} for $k$-Median, and then round it with a $\rho_{BR}$ approximate bipoint rounding algorithm. This implies an $\alpha\cdot \rho_{BR}$ approximation. 
The current-best value of $\rho_{BR}$ is $1.338$ by Byrka, Pensyl, Rybicki, Srinivasan, and Trinh [SODA'15, TALG'17], which yields $2.6742$-approximation. Combining their algorithm with our refined LMP algorithm for UFL (replacing JMS) gives a $2.67059$-approximation.
\end{abstract}

\thispagestyle{empty}
\setcounter{page}{0}

\newpage
\pagenumbering{arabic}

\setcounter{page}{1}

\section{Introduction}

Given a set of points (clients) $C$ and a set of facilities $F$ in a metric space $(C \cup F, dist)$, the classic $k$-Median problem asks to find a subset $S\subseteq F$ of size $k$
(the \emph{centers} or \emph{facilities}), such that the total distance $d(S) := \sum_{p\in C} \underset{c\in S}{~\min~} dist(p,c)$ of points in $C$ to the closest facility is minimized. In the related Uncapacitated Facility Location problem (UFL), 
opening costs $\open : F \to \R^+$ replace $k$ in the input, and we need to compute $S \subseteq F$ that minimizes $\open(S) + d(S)$, $\open(S):=\sum_{f\in S}\open(f)$. 
Initially motivated by operations research applications, the study of the complexity of these two problems 
dates back to the early 60s~\cite{KH63,Sto63} and has been the focus of a great deal of attention over the years.
They are APX-hard in general metric spaces,
and have been well-studied in terms of approximation algorithms. 
For UFL, Li~\cite{Li13} presented an algorithm with the currently best known approximation ratio of $1.488$, while Guha and Khuller~\cite{GuK99} show that it admits no better than an $\gamma_{GK} \approx 1.463$ approximation. 
For $k$-Median, while a hardness-of-approximation bound of $1+2/e \approx 1.73$ has been known for more than 20 years \cite{GuK99, JMS02}, designing an algorithm matching this lower bound remains a major challenge despite ongoing efforts.

The most successful techniques to approximate $k$-Median and UFL are essentially of two kinds: Local search and LP-based techniques. Local search starts with a given feasible solution, and iteratively improves it by enumerating all ``neighbouring'' feasible solutions (obtained by swapping in and out a subset of at most $\Delta=O(1)$ centers) and switching to any such solution if the cost decreases.
\cite{AryaGKMMP04} showed that this approach achieves a $(3+2/\Delta)$ approximation for $k$-Median, and this factor is tight (more precisely, there exists a locally optimal solution which is $(3+2/\Delta)$ approximate). They also showed that local search achieves a $(1+\sqrt{2})\approx 2.41$ approximation for UFL. 
 
 There are two families of LP-based approximation algorithms. The first family is based on the direct rounding of a fractional solution to an LP relaxation (see e.g. \cite{ARS03,CGTS02,ByrkaA10, Li13, CL12}). The second approach, which is the most relevant for this paper, is based on LMP approximations for UFL. 
Suppose that we are given a Lagrangian Multiplier Preserving (LMP) $\alpha_{LMP}$ approximation algorithm for UFL; the solution $S$ produced by this algorithm satisfies $\open(S) + d(S) \leq \open(S^*) + \alpha_{LMP} \cdot d(S^*)$ for any feasible solution $S^*$.\footnote{An equivalent definition also used in the literature is that $\alpha_{LMP}\, \open(S)+d(S) \le \alpha_{LMP} (\open(S^*)+ d(S^*))$. Two definitions are equivalent by scaling facility costs.}
For $k$-Median, there is a framework using {\em bipoint solutions} that takes any LMP $\alpha_{LMP}$ approximation algorithm for UFL and constructs an $\alpha_{LMP} \cdot \rho_{BR}$ approximation, where $\rho_{BR}$ denotes the best known bipoint rounding ratio. (See Section~\ref{subsec:prelim} for details.) 
The fundamental paper by Jain and Vazirani \cite{JaV01} gives $\alpha_{LMP}=3$ and $\rho_{BR}=2$, leading to a $6$ approximation for $k$-Median. The value of $\alpha_{LMP}$ was later improved to $2$ in another seminal paper by Jain, Mahdian, and Saberi~\cite{JMS02, JainMMSV03}, where the authors use the \emph{Dual-Fitting} technique. 
While the lower bound remains $1 +2/e \approx 1.73$~\cite{JMS02}, in the last two decades no progress was made on $\alpha_{LMP}$. Recent progress was however made on $\rho_{BR}$: In a breakthrough result $\rho_{BR}$ was improved to $1.3661$ by Li and Svensson \cite{LiS16}, and later refined to $1.3371$ by~\cite{ByrkaPRST17}. \cite{ByrkaPRST17} also shows that $\rho_{BR}\geq \frac{1+\sqrt{2}}{2}>1.207$. Combining \cite{JMS02} with \cite{ByrkaPRST17} one obtains the current best $2\cdot 1.3371=2.6742$ approximation for $k$-Median.
There is also a simple reduction showing that an $\alpha$ LMP inapproximability for UFL with uniform facility costs implies an $\alpha$ inapproximability for $k$-Median for any $\alpha \geq 1$, so the question $\alpha_{LMP} < 2$ is directly related to whether $k$-Median would eventually admit a strictly better than $2$ approximation. (The standard LP relaxation for $k$-Median has an integrality gap at least $2 - o(1)$~\cite{JMS02}.)

Bicriteria approximation algorithms for UFL, of which LMP approximation algorithms are an extreme case, play an important role for the approximation ratio of UFL.  The best known $1.488$ approximation algorithm for UFL by Li~\cite{Li13} uses the algorithms by Madhian, Ye, and Zhang~\cite{MahdianYZ06} and Byrka and Aardal~\cite{ByrkaA10} that guarantee 
$\open(S) + d(S) \leq 1.11\open(S^*) + 1.78 d(S^*)$ and 
$\open(S) + d(S) \leq 1.67\open(S^*) + 1.37 d(S^*)$ for any $S^*$ respectively. 
\cite{MahdianYZ06} also proves that the optimal LMP $(1+2/e)$ approximation automatically yields the optimal $\gamma_{GK}$ approximation for UFL. 


\subsection{Our Results}
Due to the lack of progress on improving $\alpha_{LMP}$ in the last two decades, it is natural to ask the following question:
\begin{problem}
Is there an LMP $(2-\eta)$ approximation for UFL for some constant $\eta>0$?
\end{problem}   
We answer the above question in the affirmative (though for a rather small $\eta>2.25 \cdot 10^{-7}$), even for the general version of UFL where facilities can have different opening costs. The constant $\eta$ can be substantially improved for the case of uniform opening costs, which is sufficient in our application to $k$-Median. Indeed, for $k$-Median some other technical refinements are possible, altogether leading to a more substantial (still small) improvement of $k$-Median approximation factor from $2.6742$ to $2.67059$\footnote{We did not insist too much on improving $\eta$, but we made a substantial effort to refine the approximation factor for $k$-Median.}.

While showing (though only qualitatively) that the LMP $2$ approximation barrier can be breached is interesting in our opinion, probably the most interesting aspect of our work is at technical level. The dual-fitting LMP 2 approximation by~\cite{JMS02} is tight. In particular, a tight instance can be built via their factor-revealing LP. One might consider also a natural local-search-based algorithm for UFL, where one swaps in and out subsets of up to $\Delta=O(1)$ facilities: however (see Claim~\ref{claim:ls_bad} in Section \ref{sec:fl}) such an algorithm is not an LMP $\alpha_{LMP}$ approximation for any constant $\alpha_{LMP}$ (in the case of non-uniform opening costs).

Somewhat surprisingly, we show that a combination of dual-fitting and local search is better than each such approach in isolation. Consider first the case of uniform opening costs. We consider the UFL solution $S$ produced by the algorithm JMS in \cite{JMS02}. We seed the above mentioned local search UFL algorithm with $S$ (i.e., $S$ is used as the initial feasible solution), hence obtaining a refined solution $S'$. The latter solution turns out to be an LMP $2-\eta<2$ approximation for the problem. More precisely, we can show that either $S$ is already a good enough solution, or it belongs to the local neighbourhood of a good enough solution. Intuitively, the worst case instances for JMS and local search are to a certain
degree mutually exclusive. For general opening costs we augment the standard local-search neighborhood by also considering the solutions obtained by setting to zero the cost of some facilities and running JMS itself on the obtained instance. We believe that our result is conceptually interesting: seeding a local search algorithm with an LP-based approximate solution is a natural heuristic; our result supports this heuristic from a theoretical viewpoint.

\subsection{Related Work}
\label{subsec:related}

\paragraph{$k$-Median.} 

$k$-Median remains APX-hard~\cite{Cohen-AddadS19,CKL21,MeS84} in Euclidean metrics, and here the best-known approximation is $2.406$~\cite{CEMN22}. The problem however admits a local-search-based PTAS for a constant number of dimensions~\cite{Cohen-Addad18}. This extends to metrics with bounded doubling dimension~\cite{FriggstadRS19,CFS19}, and graph metrics with an excluded minor \cite{Cohen-AddadKM19}. Local search computes optimal or nearly optimal solutions also assuming a variety of stability conditions~\cite{BalcanW17,Cohen-AddadS17} (see also \cite{BBLM14,BeT10,DGK02,HaM01,SAD04,YSZUC08} for related clustering results). For fixed $k$, there are a few PTASs, typically using small coresets \cite{BecchettiBC0S19,FL11,HuangV20,KumarSS10}. There exists also a large body of bicriteria approximations, see~\cite{BaV15,CharikarG05,CoM15,KPR00,MakarychevMSW16}.

  A related result by Chen~\cite{Chen08} also combines bipoint solutions and local-search-based approaches to obtain the first $O(1)$-approximation
  algorithm for $k$-Median with outliers (recently significantly improved by Krishnaswamy, Li, and Sandeep~\cite{KrishnaswamyLS18} using a different method). 
  In the $k$-Median with
  outliers problem the goal is to place $k$ centers so as to minimize the sum of the distances from the closest $n-z$ points to their closest center,
  where $z$ is an input integer designating the number of outliers. 
  In his paper, Chen cleverly rounds a bipoint solution for the problem by using a local search algorithm for
  the $k$-Median problem with penalty (where the contribution of each point to the objective is the minimum between the distance to the closest center
  and the penalty) with increasing values of penalty. Though the focus of his local search is different from ours (we use local search to obtain better bipoint solutions instead of rounding them), we believe that these examples suggest that combining local search with other algorithms at various stages has a great potential to obtain improved results.

\paragraph{Facility Location.}
As mentioned previously, following the extensive literature on UFL~\cite{ByrkaA10,CharikarG05,ChudakS03,JainMMSV03,JMS02,KPR00,MahdianYZ06,ShmoysTA97}, the $1.488$ approximation by Li~\cite{Li13} cleverly combines the algorithms of Mahdian, Ye, and Zhang~\cite{MahdianYZ06} and Byrka and Aardal~\cite{ByrkaA10} that achieve different $(\alpha, \beta)$-approximations; an algorithm for UFL (producing a solution $S$) is an $(\alpha,\beta)$ approximation if, for any feasible solution $S'$, one has $\open(S)+d(S)\leq \alpha \open(S')+\beta d(S')$. In this language,~\cite{MahdianYZ06} gives a $(1.11, 1.78)$ approximation and~\cite{ByrkaA10} gives a $(1.6774, 1.3738)$-approximation. Achieving better than $(\gamma, 1 + e^{-\gamma})$ is NP-hard~\cite{JMS02}, and the latter result matches this lower bound for $\gamma = 1.6774$. Note that Guha and Khuller's lower bound $\gamma_{GK} \approx 1.463$ for UFL is the solution of $\gamma = 1 + e^{-\gamma}$. 

There is also a large body of work in Capacitated Facility Location, where each facility location has a limited capacity~\cite{chudak1999improved, korupolu2000analysis, pal2001facility, aggarwal2013, bansal2012, an2017lp}. After many constant factor approximation algorithms using local search, An, Singh, and Svensson~\cite{an2017lp} presented the first constant factor approximation algorithm using a LP relaxation.

\subsection{Preliminaries}
\label{subsec:prelim}

\paragraph{Facility Location and JMS Algorithm.}
In the classical (metric uncapacitated) Facility Location problem (UFL) we are given a set of facilities $F$ and a set of clients $C$ together with a metric $dist(\cdot,\cdot)$ over $F\cup C$. Furthermore, we are given opening costs $\open(f)\geq 0$ over facilities $f\in F$. A feasible solution consists of a subset $S\subseteq F$ of \emph{open} facilities. Let $\open(S):=\sum_{f\in S}\open(f)$ be the opening cost of the solution, and $d(S):=\sum_{j\in C}dist(j,S)$ be its connection cost, where $dist(j,S):=\min_{f\in S}\{dist(j,f)\}$. Intuitively, we open the facilities in $S$ and connect each client $j$ to the closest open facility $f\in S$ (paying the associated connection cost $dist(j,S)$). Our goal is to minimize the total cost $\open(S)+d(S)$. We will also consider the special case of UFL where all opening costs have uniform value $\lambda$. Recall that an algorithm for UFL (producing a solution $S$) is an $(\alpha,\beta)$ approximation if, for any feasible solution $S'$, one has $\open(S)+d(S)\leq \alpha \open(S')+\beta d(S')$. If $\alpha=1$, that algorithm is a \emph{Lagrangian Multiplier Preserving} (LMP) $\beta$ approximation for FL.

As mentioned in the introduction, our algorithm performs local search starting from the solution computed by the classical JMS algorithm~\cite{JMS02, JainMMSV03}.
Roughly speaking, the JMS algorithm works as follows. At any point of time we maintain a subset of \emph{active} clients $A\subseteq C$ and a subset of open facilities $S\subseteq F$. Furthermore, for each $j\in C$, we have a variable $\alpha_j\geq 0$ and a facility $S(j)\in S$ which is defined only for \emph{inactive} clients $j\in C\setminus A$ ($j$ is \emph{connected} or \emph{assigned} to $S(j)$). Initially $A=C$, $S=\emptyset$ and $\alpha=0$. At any point of time each client $j$ makes an \emph{offer} to each \emph{closed} facility $f\in F\setminus S$. If $j$ is active, this offer is $\max\{0,\alpha_j-dist(j,f)\}$. Otherwise ($j$ is inactive), the offer is $\max\{0,dist(j,S(j))-dist(j,f)\}$. We increase at uniform rate the dual variables of active clients until one of the  following events happens. If $\alpha_j=dist(j,f)$ for some $j \in A, f\in S$, we remove $j$ from $A$ and set $S(j)=f$ ($j$ is connected to $f$). If the offers to a closed facility $f\in F\setminus S$ reach its opening cost $\open(f)$, we open $f$ (i.e. add $f$ to $S$), and set $S(j)=f$ for all the clients who made a strictly positive offer to $f$. Notice that $S(j)$ might change multiple times. We iterate the process until $A=\emptyset$, and then return $S$. 

This algorithm is analyzed with the dual-fitting technique. We interpret the variables $\alpha$ as a dual solution for a proper standard LP (that we will omit). From the above construction it should be clear that the sum of the dual variables upper bounds the total cost of the final solution $S$. The dual solution turns out to be infeasible, but it can be made feasible by scaling down all the dual variables by a constant factor $\gamma>1$ which also determines the approximation factor of the algorithm. The value of $\gamma$ is determined by upper bounding the value of a properly defined \emph{factor-revealing} LP, which is parameterized by an integer cluster size $q$. In particular the authors show that $\gamma<1.61$. By modifying the objective function of such LP, one can also show that JMS  is an $(\alpha,\beta)$ approximation algorithm for different pairs $(\alpha,\beta)$, and this is at the heart of the current best approximation algorithm for UFL \cite{Li13}. In particular, we consider the following variant of a factor-revealing LP used in \cite{JMS02} to show that JMS is an LMP $2$-approximation for UFL\footnote{The normalization constraint \eqref{con:sumdi} can be omitted by dividing the objective function by $\sum_{j=1}^{q}d_j$. The maxima in \eqref{con:LBlambda} can be avoided by introducing auxiliary variables.}:
\begin{align}
\mbox{max} & \sum_{i=1}^q \alpha_i - \lambda \nonumber & (LP_{JMS(q)}) \\ 
\mbox{s.t.} 
& \sum_{i=1}^q d_i = 1 \label{con:sumdi} & \\
& \alpha_i \leq \alpha_{i+1} & \forall 1 \leq i < q \label{con:orderalphai}\\
& r_{j,i+1} \leq r_{j,i} & \forall 1 \leq j  \leq i < q\label{con:orderrji}\\
& \alpha_i \leq r_{j,i} + d_i + d_j & \forall 1 \leq j < i \leq q\label{con:triangleIneq} \\  
& r_{j,j} \leq \alpha_j & \forall 1 \leq j \leq q \label{con:addedConstr} \\
& \sum_{j=1}^{i-1} \max\{r_{j,i} - d_j,0\} + \sum_{j=i}^{q} \max\{\alpha_i-d_j,0\} \leq \lambda & \forall 1 \leq i \leq q \label{con:LBlambda} \\
& \alpha, d, \lambda, r\geq 0 \nonumber   
\end{align}
In the above LP, for $1\leq j<i\leq q$, $r_{j,i}$ is interpreted as the distance between client $j$ and the facility it is connected to right before time $\alpha_i$, assuming $\alpha_j<\alpha_i$. Otherwise, $r_{j,i}=\alpha_j$. 
Equation \eqref{con:sumdi} is the normalizing constraint so that the objective value of the program directly becomes the LMP approximation ratio, 
\eqref{con:orderalphai} orders the clients in the increasing order of $\alpha_i$'s,
and \eqref{con:orderrji} indicates that a client will be connected to closer facilities as time progresses. 
W.r.t. the LP in \cite{JMS02}, we add the variables $r_{j,j}$ which denote the distance between client $j$ and the first facility it is connected to. W.r.t. to the LP studied in \cite{JMS02}, we add the constraint $r_{j,j+1}\leq r_{j,j}$ in \eqref{con:orderrji} and the new constraint $\eqref{con:addedConstr}$, namely $r_{j,j}\leq \alpha_j$. Both constraints are clearly satisfied by the \emph{original} algorithm in \cite{JMS02}. 
Equation~\eqref{con:triangleIneq} uses the triangle inequality to show that client $i$ can be connected to the facility client $j$ is already connected to, 
and~\eqref{con:LBlambda} shows that the current facility is never overpaid (i.e., it becomes open as soon as the equality holds). 
Let $\opt_{JMS}(q)$ be the optimal solution to $LP_{JMS}(q)$, and $\opt_{JMS}=\sup_q\{\opt_{JMS}(q)\}$. The analysis in \cite{JMS02} directly implies $\opt_{JMS}\leq 2$ (since we only added constraints). We will use a variant of the above LP to prove our improved LMP bound for UFL\footnote{We remark that JMS produces an LMP $2$-approximation w.r.t. the value of a standard LP, while our approach provides an LMP $2-\eta$ approximation but only w.r.t. the optimal integral solution (due to our use of local search). Finding an LMP $2-\eta$ approximation w.r.t. a natural LP is an interesting open problem.
}.       

Given a facility location solution $S$ and a client $j$, we let $S(j)$ denote the facility serving client $j$ in solution $S$. We also let $C_S(f)$ be the clients served by facility $f\in S$ in solution $S$. For $S'\subseteq S$, $C_S(S')=\cup_{f\in S'}C_S(f)$. For $f\in S$ and $S'\subseteq S$, we let $d_S(f):=\sum_{j\in C_{S}(f)}dist(j,f)$ and $d_S(S')=\sum_{f\in S'}d_S(f)$. In all the mentioned cases, we sometimes omit the subscript $S$ when $S$ is clear from the context. 

\paragraph{Bipoint solution.} Consider a $k$-Median instance $(C,F,dist,k)$ with an optimal solution $\OPT$ of cost $\opt$. Let us assume w.l.o.g. that $k<|F|$, otherwise the problem can be solved optimally in polynomial time by opening the entire $F$. Given an LMP $\alpha_{LMP}$ approximation for UFL, a bipoint solution is constructed as follows. For a given parameter $\lambda\geq 0$, consider the UFL instance with clients $C$, facilities $F$, and uniform opening cost $\lambda$. Let $S(\lambda)\subseteq F$ be the solution produced by the considered UFL algorithm on this instance. We assume that $S(\lambda')$ consists of precisely one facility for a large enough $\lambda'$, and that $S(0)=F$. These properties are either automatically satisfied by the UFL algorithm, or they can be easily enforced. 
We perform a binary search in the interval $[0,\lambda']$ until we find two values $0\leq \lambda_2<\lambda_1\leq \lambda'$ such that $S_1:=S(\lambda_1)$ opens $k_1\leq k$ facilities and $S_2:=S(\lambda_2)$ opens $k_2> k$ facilities. Furthermore, $(\lambda_1-\lambda_2)k\leq \eps\cdot \opt$. Interpreting $S_i$ as an incidence vector over $F$, the bipoint solution $S_B=aS_1+bS_2$ is a convex combination of $S_1$ and $S_2$ with coefficients $a=\frac{k_2-k}{k_2-k_1}$ and $b=1-a=\frac{k-k_1}{k_2-k_1}$. Notice that $S_B$ opens precisely $k$ facilities in a fractional sense. In $S_B$ each client $c$ is connected by an amount $a$ to $S_1(c)$ and by an amount $b$ to $S_2(c)$, hence the connection cost of $S_B$ is $a\,d(S_1)+b\,d(S_2)$.
Recall that by assumption $\lambda_i |S_i| + d(S_i) \leq \lambda_i k+\alpha_{LMP} \cdot \opt$. Thus 
$$
a\,d(S_1)+b\,d(S_2)\leq a\lambda_1 (k-k_1)+b\lambda_2 (k-k_2)+\alpha_{LMP} \cdot \opt\leq (\alpha_{LMP}+\eps)\opt.
$$
We will next neglect the term $\eps\opt$ since it is subsumed by our numerical overestimates. Furthermore, to lighten the notation, we will neglect the slight gap between $\lambda_1$ and $\lambda_2$ and simply assume $\lambda_1=\lambda_2=\lambda$.

\paragraph{Organization.} To illustrate (some of) our main ideas in a simpler setting, in the main body we will focus on approximation algorithms for $k$-Median. To that aim, it is sufficient to consider LMP approximation algorithms for UFL \emph{in the case of uniform opening cost $\lambda$} (which in turn is sufficient to build a convenient bipoint solution). As mentioned earlier, our UFL algorithm starts with the solution produced by JMS, and then refines it by means of local search. In Sections~\ref{sec:LMP} and~\ref{sec:localSearch} (with some proofs in Sections~\ref{sec:LMP_omittedProofs} and~\ref{sec:localSearchOmitted}, resp.) we provide two distinct upper bounds on the cost of the produced UFL solution. The first such bound is based on a refinement of the analysis in \cite{JMS02} which holds in certain special cases (see Section \ref{sec:LMP_old} for an alternative analytical approach). The second bound instead exploits the local optimality of the final solution. Combining these two bounds, we derive an improved approximation for $k$-Median in Section \ref{sec:improvedkMedian} (with some proofs in Section~\ref{sec:omittedSection4}). In Section \ref{sec:improvedFLuniform} we describe how to obtain a better than $2$ LMP approximation for UFL with uniform facility costs. A refined approximation for $k$-Median is sketched in Section \ref{sec:refinedkMedian} (using results from Section \ref{sec:boundsS1}). Finally our improved LMP approximation for UFL with arbitrary opening costs is given in Section \ref{sec:fl}. 

\section{An Improved LMP Bound for Small Facility Cost}
\label{sec:LMP}

Given a Facility Location instance $(C, F, dist, \open)$, let $\OPT \subseteq F$ be some given solution to a UFL instance. Suppose that a constant fraction $\alpha>0$ of the connection cost of $\OPT$ is associated with a subset $\OPT'\subseteq \OPT$ of facilities whose opening cost $\open(\OPT')$ is at most some constant $T$ times the respective connection cost $d(\OPT')$. Then we can show that JMS on this instance produces a solution of cost at most $\open(\OPT)+(2-\eta)d(\OPT)$, for some constant $\eta=\eta(T,\alpha)>0$. Hence in some sense JMS is an LMP $(2-\eta)$ approximation algorithm w.r.t. this type of solutions. This result will be used in our overall win-win strategy for $k$-Median; later in Theorem~\ref{thr:localSearch} and Lemma~\ref{lem:boundLonelyFacilityCost}, we will show that if an outcome of local search with mild additional requirements (satisfied by starting local search from a JMS solution) does not achieve a $(2 - \eta')$ LMP approximation for some $\eta' > 0$, then $\OPT$ will have the above property, so JMS already guarantees a $(2-\eta)$ LMP approximation.

In order to prove this claim, let us have a closer look at the analysis of JMS in \cite{JMS02}. Using the dual-fitting technique, the authors show that, for any $f \in OPT$, $\sum_{j\in C(f)}\alpha_j\leq \open(f)+\opt_{JMS}\cdot d(f)\leq \open(f)+2 d(f)$. 
The worst case value of $\opt_{JMS}$ in their analysis is achieved for $\lambda=\open(f)$ arbitrarily larger than $d(f)=\sum_{j=1}^{q}d_j=1$. However for the facilities in $\OPT'$ this cannot happen on average since $\sum_{f\in \OPT'}\open(f)=\open(\OPT')\leq T\cdot d(\OPT')=T\cdot \sum_{f\in \OPT'}d(f)$. This motivated us to study a variant $LP_{JMS}(q,T)$ of $LP_{JMS}(q)$, for a parameter $T>0$, where we add the following constraint that intuitively captures the condition $\open(f)\leq T\cdot d(f)$
\begin{equation}
\lambda \leq T.\label{con:Tbound-new}
\end{equation}
Let $\opt_{JMS}(q,T)$ be the optimal value of $LP_{JMS}(q,T)$, and $\opt_{JMS}(T)=\sup_q\{\opt_{JMS}(q,T)\}$. Observe that $\opt_{JMS}(q,T)\leq \opt_{JMS}(q)$ where the equality holds for any fixed $q$ and $T$ large enough. A straightforward adaptation of the analysis in \cite{JMS02} using our modified LP $LP_{JMS}(q,T)$ implies the following.
\begin{lemma}\label{lem:modifiedJMSclaim}
Let $\OPT$ be a solution to a UFL instance and $S$ be the JMS solution on the same instance. Then
$$
\open(S)+d(S)\leq \open(\OPT)+\sum_{f^*\in \OPT}\opt_{JMS}(\frac{\open(f^*)}{d(f^*)})\cdot d(f^*).
$$
\end{lemma}
Given Lemma \ref{lem:modifiedJMSclaim} and the above discussion, it makes sense to study the behaviour of $\opt_{JMS}(T)$ for bounded values of $T$. We can show that the supremum defining $\opt_{JMS}(T)$ is achieved for $q$ going to infinity, i.e. $\opt_{JMS}(T)=\lim_{q\to +\infty}\opt_{JMS}(q,T)$:
\begin{restatable}{lemma}{lemjmssplitnew}
\label{lem:jms-split-new}
Fix $T > 0$. For any positive integers $q$ and $c$, 
$\opt_{JMS}(q, T) \leq \opt_{JMS}(cq, T)$. 
\end{restatable}
\begin{proof}[Proof sketch]
Given an optimal solution $S=(\alpha,d,r,\lambda)$ of $LP_{JMS}(q,T)$, consider the following solution $S'=(\alpha',d',r',\lambda')$ of $LP_{JMS}(cq,T)$:
\begin{align*}
\alpha'_{i}:=\frac{\alpha_{\lceil i/c\rceil}}{c},\quad d'_{i}=\frac{d_{\lceil i/c\rceil}}{c},\quad 
r'_{j,i}:=\begin{cases} 
\frac{r_{\lceil j/c\rceil,\lceil i/c\rceil}}{c} & \text{if } \lceil j/c\rceil<\lceil i/c\rceil;\\
\frac{\alpha_{\lceil j/c\rceil}}{c} & \text{if }\lceil j/c\rceil=\lceil i/c\rceil.
\end{cases}
\quad \lambda'=\lambda,\quad \forall 1\leq j\leq i\leq cq.
\end{align*}
Observe that $S$ and $S'$ have exactly the same objective value. The reader can easily verify that $S'$ is feasible (details in Section \ref{sec:LMP_omittedProofs}). The claim follows.
\end{proof}

We also observe that $\opt_{JMS(T)}$ is a concave function.
\begin{restatable}{lemma}{lemJMSconcavity}\label{lem:JMSconcavity}
$\opt_{JMS}(T)$ is a concave function of $T$. 
\end{restatable}

We next show how to compute an upper bound on $\opt_{JMS}(T)$ for any given $T$. Here we slightly adapt an approach presented in \cite{BPRST15}, which avoids to construct a family of feasible dual solutions as in \cite{JMS02}\footnote{\cite{BPRST15} contains a technical bug which was later fixed in the journal version \cite{ByrkaPRST17} (leading to a slightly worse approximation factor). Here we adapt a part of the analysis in \cite{BPRST15} which does not appear in \cite{ByrkaPRST17} (nor in the most recent arXiv version). 
}. This leads to a simpler and easier to adapt analysis.

We define an alternative LP $LP^{+}_{JMS}(q,T)$ which is obtained from $LP_{JMS}(q,T)$ by replacing Constraint \eqref{con:LBlambda} with the following constraint
\begin{equation}\label{con:LBlambda-plus}
\sum_{j=1}^{i} \max\{r_{j,i} - d_j,0\} + \sum_{j=i+1}^{q} \max\{\alpha_i-d_j,0\} \leq \lambda, \quad\quad \forall 1 \leq i \leq q.
\end{equation}
In particular, there is a shift by one of two indexes in the two sums. \fab{Furthermore, we replace Constraint \eqref{con:triangleIneq} with the following constraint
\begin{equation}\label{con:triangleIneq-plus}
\alpha_i \leq r_{j,i-1}+d_i+d_j, \quad\quad \forall 1 \leq j < i \leq q.
\end{equation}
In particular, $r_{j,i}$ is replaced by $r_{j,i-1}$.
}

Let $\opt^+_{JMS}(q,T)$ be the optimal solution to $LP^{+}_{JMS}(q,T)$. The next lemma follows analogously to \cite{BPRST15}.
\begin{restatable}{lemma}{lemjmsUBnew}\label{lem:jms-UB}
Fix $T > 0$. For any positive integers $q$ and $c$, 
$\opt^+_{JMS}(q, T) \geq \opt_{JMS}(cq, T)$. 
\end{restatable}
\begin{proof}[Proof sketch]
Let $S:=(\alpha,d,r,\lambda)$ be an optimal solution to $LP_{JMS}(cq,T)$. Consider the following solution $S^+:=(\alpha^+,d^+,r^+,\lambda^+)$ for $LP^{+}_{JMS}(q,T)$:
$$
\alpha^+_i:=\sum_{\ell=(i-1)c+1}^{ic}\alpha_\ell,\quad d^+_i=\sum_{\ell=(i-1)c+1}^{ic}d_\ell,\quad r^+_{j,i}:= \sum_{\ell=(j-1)c+1}^{jc}r_{\ell,ic},\quad \lambda^+=\lambda, \quad \forall 1\leq j\leq i\leq q.
$$
Clearly the objective values of $S$ and $S^+$ are identical. The reader can easily check that $S^+$ is feasible (details in Section \ref{sec:LMP_omittedProofs}). The claim follows.
\end{proof}
\begin{restatable}{lemma}{lemjmsUB-new}\label{lem:jms-UB-new}
Fix $T > 0$. For any positive integer $q$, $\opt^+_{JMS}(q, T) \geq \opt_{JMS}(T)$. 
\end{restatable}
\begin{proof}
Suppose by contradiction that the claim does not hold. In particular, for some constant $\eps>0$, one has $\opt^+_{JMS}(q, T) \leq \opt_{JMS}(T)-\eps$. There must exist a finite integer value $c$ such that $\opt_{JMS}(c,T)> \opt_{JMS}(T)-\eps$. Hence we obtain the desired contradiction by Lemmas \ref{lem:jms-split-new} and \ref{lem:jms-UB}:
$$
\opt^+_{JMS}(q, T) \overset{\text{Lem. \ref{lem:jms-UB}}}{\geq} \opt_{JMS}(cq,T) \overset{\text{Lem. \ref{lem:jms-split-new}}}{\geq} \opt_{JMS}(c,T)> \opt_{JMS}(T)-\eps\geq \opt^+_{JMS}(q, T).\qedhere
$$
\end{proof}
Lemma \ref{lem:jms-UB-new} provides a convenient way to upper bound $\opt_{JMS}(T)$. Though we will not prove this formally, an empirical evaluation shows that $\opt^+_{JMS}(q, T)$ is a decreasing function of $q$ which (rather quickly) converges to $\opt_{JMS}(T)$. Hence, for a given $T$, a very tight upper bound on $\opt_{JMS}(T)$ is obtained by computing $\opt^+_{JMS}(q, T)$ for a large enough but fixed value of $q$ (in the order of a few hundreds). In the following we will use $q$ as a parameter that we will fix only in the final numerical computation. We are also able to prove an analytical (though weaker) upper bound on $\opt_{JMS}(T)$. Though this is not needed for the rest of our analysis, we present the latter upper bound in Section \ref{sec:LMP_old}.

The next corollary follows from Lemmas \ref{lem:modifiedJMSclaim}, \ref{lem:JMSconcavity}, and \ref{lem:jms-UB-new} (proof in Section \ref{sec:LMP_omittedProofs}). 
\begin{restatable}{corollary}{cormodifiedJMSclaim}\label{cor:modifiedJMSclaim}
Let $\OPT$ be a solution to a UFL instance and $S$ be the JMS solution on the same instance. For a given constant $T>0$, let $\OPT'\subseteq \OPT$ be a subset of facilities satisfying $\open(\OPT')\leq T\cdot d(\OPT')$. Then, for every integer $q>0$, 
$$
\open(S)+d(S)\leq \open(\OPT)+2d(\OPT\setminus \OPT')+\opt^+_{JMS}(q,T)\cdot d(\OPT').
$$
\end{restatable}

\section{A Local-Search-Based Bound}
\label{sec:localSearch}

In this section we present our local-search-based bound for UFL. 
Suppose that we start with a facility location solution $S'$ for some given uniform facility opening cost $\lambda$. 
(In our case $S'$ will be produced by the JMS algorithm, though the result of this section holds for an arbitrary starting solution $S'$.) We refine $S'$ via local search as follows. Let $\eps>0$ be a constant parameter. A feasible \emph{swap pair} $(A,B)$ consists of $A\subseteq S'$ and $B\subseteq F$ such that $|A|,|B|\leq \Delta$. Here $\Delta\geq 1$ is a constant depending on $\eps$ to be fixed later. For any feasible swap pair, we consider the alternative solution $S''=S'\setminus A\cup  B$ and if $S''$ has total (i.e., facility plus connection) cost strictly smaller than $S'$, we replace $S'$ with $S''$ and repeat. 
Lemma~\ref{lem:discrete_localsearch} shows that with a loss of $1+\eps$ in the approximation factor, we can guarantee that the local-search procedure ends after a polynomial number of rounds. 
For the sake of clarity we simply assume here and omit the extra factor $1+\eps$ since it is subsumed by similar factors in the rest of our analysis.  

Let $S'$ be the solution at the end of the process (i.e., $S'$ is a {\em local optimum}), with $k'=|S'|$ and $d'=d(S')$. Let also $\OPT$ denote some given UFL solution, with $k=|OPT|$ and connection cost $\opt=d(\OPT)$. Intuitively, for us $\OPT$ will be the optimal solution to the underlying $k$-Median problem. We will next describe a bound on the total cost of $S'$ which is based on a carefully chosen classification of the facilities in $S'$ and $\OPT$. 

\paragraph{Matched vs. Lonely facilities: a key notion.}

Given two facility location solutions $S$ and $S'$ and $\alpha\geq 1/2$, we say that $f\in S$ $\alpha$-captures $f'\in S'$ if more than an $\alpha$ fraction of the clients served by $f'$ in $S'$ are served by $f$ in $S$, i.e. $|C_{S}(f)\cap C_{S'}(f')|>\alpha|C_{S'}(f')|$. Since $\alpha\geq 1/2$, $f'\in S'$ can be captured by at most one $f\in S$. Similarly, we say that $A\subseteq S$  $\alpha$-captures $f'\in S'$ if more than an $\alpha$ fraction of the clients served by $f'$ in $S'$ are served by some $f\in A$ in $S$.

We classify the facilities in $S'$ and $\OPT$ as lonely and matched as follows. Let $\delta \leq 1/2$ be a parameter to be fixed later. Suppose first that $|S'|=k'\leq k=|\OPT|$ (intuitively corresponding to $S'=S_1$). In this case, consider $f^*\in \OPT$, and let $M(f^*)$ be the facilities in $S'$ which are $1/2$-captured by $f^*$. If $M(f^*)$ $(1-\delta)$-captures $f^*$, then we say that $f^*$ and $M(f^*)$ are \emph{matched}. In the complementary case $k'> k$ (intuitively corresponding to $S'=S_2$) we use a symmetric definition. Let $f'\in S'$, and let $M(f')$ be the facilities in $\OPT$ which are $(1-\delta)$-captured by $f'$. If $M(f')$ $1/2$ captures $f'$ we say that $f'$ and $M(f')$ are \emph{matched}.

All the facilities which are not matched are \emph{lonely}. We use $\OPT^M$ and $\OPT^L$ to denote the facilities in $\OPT$ which are respectively matched and lonely. We define similarly $S^M$ and $S^L$ w.r.t. to $S'$. By $\opt^X$, $X\in \{M,L\}$, we denote the connection cost of clients served by $\OPT^X$ in $\OPT$. Similarly, by $\opt^{XY}$, $X,Y\in \{M,L\}$, we denote the connection cost of clients served by $S^X$ in $S'$ and by $\OPT^Y$ in $\OPT$. We define $d^{XY}$ analogously w.r.t. the connection cost of $S'$.

\paragraph{High-level overview.}
At a high-level, pairs of facilities of $S'$ and $\OPT$ induce \emph{well-served} clients in the solution output by local
search.
Indeed, think about the extreme scenario where there is one facility of $f' \in S'$ and one facility of $f^* \in \OPT$ that
serve the exact same group of clients. Then, by local optimality, one could argue that this group of clients
is served optimally in $S'$ since the swap that closes down $f'$ and opens $f^*$ does not decrease the cost.
Thus generalizing this idea, our analysis shows that clients served by matched facilities of $\OPT$ are served
nearly-optimally (i.e,: much better than the 2 approximation of JMS) while the remaining clients are served
similarly than in the classic local search analysis (namely within a factor 3 of the optimum). 

Our proof then decouples the matched and lonely facilities. We provide an analysis of the local search solution
that is specific to the matched facilities, and use the analysis of Gupta and Tangwongsan~\cite{LSGupta}. We then
show how to mix the two analysis.
Concretely, the proof works as follows.
We define feasible swap pairs $(A_i,B_i)$ so that the $A_i$'s and $B_i$'s are carefully chosen subsets of $S'$ and $\OPT$.
The goal is thus to bound the variations $\Delta_i$ defined by the cost of solution $S' - A_i \cup B_i$ minus the cost
of $S'$. Local optimality of $S'$ implies that $\Delta_i$ is non-negative for all $i$, providing us with an upper bound
on the cost of $S'$. The challenge is thus to provide
the best possible upper bound on the cost of $S' - A_i \cup B_i$, in terms of the cost of $\OPT$ and $S'$.
The usual challenge here is to bound the \emph{reassignment}
cost of the clients served by a facility of $A_i$ in $S'$ but not served
by a facility of $B_i$ in $\OPT$.

The swap pairs are defined so that matched facilities belong to the same swap pair 
(when one facility is matched to a large number of facilities, we carefully break them into multiple swap pairs to
have the same effect).
This way, for matched facilities  $f' \in S'$
and $f^* \in \OPT$, the clients served by $f'$ in $S'$ and $f^*$ in $\OPT$ (call them {\em matched} clients) do not induce a reassignment cost -- they are
served optimally in the swap pair involving $f'$ (and so $f^*$). The clients of $f^*$ that are not served by $f'$ still benefit
from being served by a matched facility $f^*$. Any such client served by a facility $f'' \in S'$ can be reassigned to
$f'$ when $f''$ is swapped out. The reassignment cost in this case can be charged to the matched clients; since the 
clients of $f^*$ served by a facility different from $f'$ are in minority, the charge is tiny.

The swap pairs are also built in a way that lonely facilities of $\OPT$ belong to the same swap with their closest
facility in $S'$, following the approach of  Gupta and Tangwongsan~\cite{LSGupta}. This ensures that the reassignment
cost is overall bounded by a factor 3 (as shown in~\cite{LSGupta}).
Building the swap pairs satisfying the above two properties requires a more careful approach than in~\cite{LSGupta}
and is an important step in our proof.

We are ready to state our local-search based bound (with proof in Section \ref{sec:localSearchOmitted}). 
\begin{theorem}\label{thr:localSearch}
  For $\Delta = (2\eps^{-2\eps^{-7}})$, one has
  \[
 \lambda k'+d' \leq
   \lambda k  +  3 \opt^{L} + \opt^{M} +  \frac{\delta}{1-\delta} (d^{MM} + \opt^{MM})
  +  O(\eps (d' + \opt)).
  \]
\end{theorem}

A high-level corollary of the above theorem is that, if most of the connection cost of $\OPT$ is due to clients served by matched facilities (and for $\delta$ small enough), then $S'$ provides an LMP $\alpha_{LMP}$ approximation (w.r.t. to $\OPT$) for some $\alpha_{LMP}<2$ (hence improving on the JMS bound). Since this result is an intrinsic property of local search and does not use the properties of the starting solutions (e.g., those produced by the JMS algorithm), we hope that it may be useful in the future to show that local search can improve the quality of the solution of other algorithms.

Handling non-uniform facility opening costs via local search requires a new notion of local neighborhood because the disparity in facility costs forces some swap pairs $(A_i, B_i)$ to contain a superconstant number of facilities to result in a meaningful guarantee. For general UFL, we introduce a new local search algorithm where each local move closes a constant number of open facilities and runs a variant of the JMS algorithm to find a new solution that can possibly open a large number of new facilities. Specifically, we set the opening cost of the currently open facilities to zero, and run the JMS algorithm to possibly open more facilities. 
Theorem~\ref{thr:extendJMS} provides a guarantee for this new local search; essentially the guarantee of Theorem~\ref{thr:localSearch} holds except $3\opt^L$ is replaced by $4\opt^L$ due to additional approximations created by the use of JMS.
In our opinion defining a local neighbourhood via an LP-based algorithm is a promising approach which deserves further investigation. See Section~\ref{subsec:newlocal} for the new local search and its guarantee.


\section{An Improved Approximation for $k$-Median}
\label{sec:improvedkMedian}

In this section we present the claimed improved approximation algorithm for $k$-Median. A refinement (for which we did not compute the explicit approximation factor) is sketched in Section \ref{sec:refinedkMedian}.

We consider the LMP UFL algorithm described in previous sections, where we start with the solution produced by the JMS algorithm and then apply local search to improve it. We use this algorithm to build a bipoint solution $S_B=aS_1+bS_2$ via binary search as described in Section~\ref{subsec:prelim}. Let $k_i=|S_i|$ and $d_i=d(S_i)$. Recall that $k_1\leq k<k_2$, $1-b=a=\frac{k_2-k}{k_2-k_1}$, and $ad_1+bd_2\leq 2\opt$. 
Since $S_1$ and $S_2$ are the outcomes of local search starting from solutions produced by the JMS algorithm, their total costs ($\lambda k_1 + d_1$ for $S_1$ and $\lambda k_2 + d_2$ for $S_2$) are at most the total cost of the starting JMS solutions, which are at most $\lambda k+2\opt$ by the standard $2$ LMP approximation guarantee. Furthermore, since $k_2 \geq k$, one can conclude that $d_2 \leq 2\opt$. 

Based on the results from Sections \ref{sec:LMP} and \ref{sec:localSearch}, we will prove the following improved upper bound on the total cost of $S_2$.
\begin{lemma}\label{lem:S2bound}
$
\lambda k_2+d_2\leq \lambda k+(2-\eta_2)\opt
$ for some absolute constant $\eta_2>0.00536$. 
\end{lemma}
Using $\lambda k_1 + d_1 \leq \lambda k + 2\opt$ with Lemma \ref{lem:S2bound}, we derive an improved bound on the connection cost of $S_B$:
$$
ad_1+bd_2\leq (2 -(1-a)\eta_2)\opt. 
$$ 
Thus combining the bipoint solution $S_B$ with the $\rho_{BR}<1.3371$ approximate bipoint rounding procedure in \cite{ByrkaPRST17}, one obtains a feasible solution with the approximation factor
$$
\rho_{BR}(2 -(1-a)\eta_2).
$$
For $a$ bounded away from $1$, this is clearly an improvement on the $2\rho_{BR}<2.6742$ approximation achieved in \cite{ByrkaPRST17}. 
(Indeed, in Section~\ref{sec:improvedFLuniform}, we also prove $\lambda k_1+d_1\leq \lambda k+(2-\eta_1)\opt$ for some $\eta_1 > 0$ as well, which establishes $ad_1 + bd_2 < (2-\eta)opt$ for some $\eta >0$ for all values of $a$.)
To get an improvement only from Lemma~\ref{lem:S2bound}, 
Li and Svensson \cite{LiS16} present a $2(1+2a)+\eps$ approximation for any constant $\eps > 0$ (with extra running time $n^{\poly(1/\eps)}$), hence we can also assume that $a$ is bounded away from $0$ (since we aim at an approximation factor strictly larger than $2$). For $a$ bounded away from $0$ and $1$, the same authors present a $\frac{2(1+2a)}{(1+2a^2)}+\eps$ approximation that we will next use. Taking the best of the two mentioned solutions, and neglecting the $\eps$ in the approximation factor, one obtains the  following approximation factor for $k$-Median
$$
\rho_{kMed}=\max_{a\in [0,1]}\min\{\frac{2(1+2a)}{1+2a^2},\rho_{BR}(2 -(1-a)\eta_2)\}.
$$
The first term in the minimum is decreasing for $a\geq \frac{\sqrt{3}-1}{2}$, hence $\rho_{kMed}<2\rho_{BR}$, thus improving on \cite{ByrkaPRST17}. We numerically\footnote{The worst case can be computed analytically, but the formula is complicated and we omit it.} obtained that the worst case is achieved for $a\simeq 0.4955391$, leading to $\rho_{kMed}<2.67059$.

\subsection{Proof of Lemma \ref{lem:S2bound}}
\label{sec:improvedkMedian:boundS2}

We next use $S^{M}_2$ instead of $S^{M}$ and similarly for related quantities to stress that we are focusing on the locally optimal solution $S'=S_2$. We let $\alpha^L:=\opt^L/\opt$, $\alpha^M:=\opt^M/\opt$, $\alpha^{MM}=\opt^{MM}/\opt$, $\beta_2=d_2/\opt$, and $\beta_2^{MM}=d^{MM}_2/\opt$. Let also $k^L=|\OPT^L|$, $k^M=|\OPT^M|$, $k^L_2=|S_2^L|$, and $k^M_2=|S_2^M|$.

We apply to $S_2$ the local-search-based bound from Theorem \ref{thr:localSearch} in Section \ref{sec:localSearch}\footnote{Theorem \ref{thr:localSearch} is proved in Section \ref{sec:localSearchOmitted}. However for this part of the analysis it is sufficient to consider the simpler subcase $|S'|>k$ discussed in Section \ref{sec:localSearchOmitted:S2}.}. In the following we will neglect the term depending on $\eps$: indeed the latter parameter can be chosen arbitrarily small, independently from the other parameters in the proof (at the cost of a larger, yet polynomial running time). Taking that term into account would involve an extra additive term $O(\eps)$ in the approximation factor, which is however absorbed by the overestimations in our numerical analysis. Thus we obtain
\begin{equation}\label{eqn:S2_boundLS}
\lambda k_2+d_2 \leq \lambda k + \big(1+2\alpha^{L} +  \frac{\delta}{1-\delta} (\beta^{MM}_2 + \alpha^{MM}) \big)\cdot \opt =:\lambda k +\rho^A\cdot \opt.
\end{equation}

Observe that the above bound already implies Lemma \ref{lem:S2bound} when $\alpha^L$ and $\delta$ are sufficiently small (notice that $\beta^{MM}_2\leq \beta_2\leq 2$). We next derive an alternative bound which implies the claim in the complementary case. Here we exploit the LMP result from Section \ref{sec:LMP}. More specifically, we will show that the facility cost $\lambda |OPT^L|$ due to lonely facilities in $\OPT$ is upper bounded by $O(\opt^L)$. This will allow us to exploit Corollary \ref{cor:modifiedJMSclaim}.
By construction for each $f'\in S^M_2$, there is at least one distinct $f^*\in \OPT^M$, hence $k^M_2\leq k^M$. Thus
\begin{equation}\label{eqn:boundNumberLonely_S2}
k^L_2=k_2-k^M_2\geq k_2-k^M=k_2-k+k^L.
\end{equation}

The local optimality of $S_2$ implies the following bound (proof in Section \ref{sec:omittedSection4}):
\begin{lemma}\label{lem:boundLonelyFacilityCost}
One has 
\begin{equation}
\lambda k^L_2 \leq \big(\frac{2}{\delta}(1-\alpha^{MM})+2\frac{1-\delta}{\delta}(\beta_2-\beta_2^{MM})+2\frac{\delta}{1-\delta}(\beta_2^{MM}+\alpha^{MM}) \big)\cdot \opt.
\label{eq:boundLonelyFacilityCost}
\end{equation}
\end{lemma}
The idea of the proof is that removing any lonely facility $f'$ from $S_2$ determines an increase of the connection cost of the clients $C_{S_2}(f')$ which is at least equal to the facility cost $\lambda$. It is therefore sufficient to upper bound the cost of connecting each $c\in C_{S_2}(f')$ to some facility $f''\in S_2\setminus \{f\}$. Here we exploit the fact that $f'$ is lonely. This implies that if $f'$ $(1-\delta)$-captures certain facilities $M(f')$ in $\OPT$, the latter facilities do not $1/2$-capture $f'$. In particular, at least one half $C'(f')$ of the clients $C_{S_2}(f')$ are served in $\OPT$ by facilities not captured by $f'$. The idea is then to define a path that goes from $c$ to some $c'\in C'(f')$, from there to the facility $f^*\in \OPT$ serving $c'$ and from $f^*$ to the desired facility $f''$ (which in particular exists since $f^*$ is not captured by $f'$). 
This idea is inspired by PTASes for $k$-Median and $k$-Means in constant-dimensional Euclidean spaces~\cite{Cohen-AddadKM19} where the authors remove lonely facilities to convert a bicriteria approximation to a true approximation. Also note that while Lemma~\ref{lem:boundLonelyFacilityCost} uses local optimality as Theorem~\ref{thr:localSearch} and holds for any local optimum as well, the upper bound on $\lambda k_2^L$ (i.e., the RHS of~\eqref{eq:boundLonelyFacilityCost}) depends on $\beta_2 = d_2 / \opt$, which can be upper bounded by $2$ using the fact that the starting solution is already $2$ LMP approximate and $k_2 \geq k$.

We are ready to prove our second upper bound on the total cost of $S_2$.  One has
\begin{align*}
\lambda k^L & \overset{\eqref{eqn:boundNumberLonely_S2}}{\leq} \lambda k^L_2 \overset{Lem. \ref{lem:boundLonelyFacilityCost}}{\leq} \frac{2}{\delta \alpha^L} \big( 1+(1-\delta)\beta_2-(1-\frac{\delta^2}{1-\delta})\alpha^{MM}-(1-\frac{\delta}{1-\delta})\beta^{MM}_2 \big)\opt^L
=:T_{L} \opt^L.
\end{align*}
We can thus apply Corollary \ref{cor:modifiedJMSclaim}
with $\OPT'=\OPT^L$ and $T=T_{L}$ to infer
\begin{equation}\label{eqn:S2_boundLMP}
\lambda k_2+d_2 \leq \lambda k+(2\cdot (1-\alpha^L)+\opt^{+}_{JMS}(q,T_{L})\cdot \alpha^L)\opt =:\lambda k +\rho^B\cdot \opt.
\end{equation}
Combining \eqref{eqn:S2_boundLS} and \eqref{eqn:S2_boundLMP} we obtain 
$$
\lambda k_2+d_2\leq \lambda k+(2-\eta_2)\opt =:\lambda k+\min\{\rho^A,\rho^B\}\opt.
$$
Notice that $\opt^{+}_{JMS}(q,T_{L})$ is a non-decreasing function of $T_L$.
Since increasing $\beta_2\leq 2$ makes $T_L$, hence $\rho^B$, bigger without decreasing $\rho^A$, we can pessimistically assume $\beta_2=2$. We can choose $\delta\in [0,1/2]$ freely so as to minimize $\min\{\rho^A,\rho^B\}$. Given $\delta$, we can pessimistically choose $\alpha^L\in [0,1]$, $\alpha^{MM}\in [0,1-\alpha^L]$, and $\beta^{MM}_2\in [0,\beta_2]=[0,2]$ so as to maximize $\min\{\rho^A,\rho^B\}$. Altogether we get:
$$
2-\eta_2\leq \min_{\delta\in [0,1/2]}\max_{\alpha^L\in [0,1]}\max_{\alpha^{MM}\in [0,1-\alpha^L]}\max_{\beta^{MM}_2\in [0,\beta_2]}\min\{\rho^A,\rho^B\}.
$$
It should now be clear that $\eta_2>0$. Indeed, by choosing $\delta$ sufficiently close to zero (but positive), we achieve $\rho^A<2$ for $\alpha^L$ small enough, say $\alpha^L<1/2$. Otherwise we can assume that both $\alpha^L$ and $\delta$ are bounded away from zero. In that case $T_{L}$ is upper bounded by some constant, and it turns out that $\opt^+_{JMS}(q,T_L)<2$ in that case 
(implying $\rho^B<2$ since by assumption $\alpha^L>0$). This concludes the proof  of Lemma \ref{lem:S2bound}.
We numerically obtained (setting $q=400$) that $\eta_2>0.005360$. The worst case is achieved for $\delta\simeq 0.49777$, $\alpha^L\simeq 0.4948527$, $\alpha^{MM}\simeq 0.004005537$, $\beta^{MM}_2\simeq 0.00097229266$ (leading to $T_L\simeq 16.25852$).

\newpage

\bibliographystyle{alpha}
\bibliography{literature}

\newpage

\appendix

\section{Omitted proofs from Section \ref{sec:LMP}}
\label{sec:LMP_omittedProofs}

\lemjmssplitnew*
\begin{proof}
Given an optimal solution $S=(\alpha,d,r,\lambda)$ of $LP_{JMS}(q,T)$, we build a feasible solution $S'=(\alpha',d',r',\lambda')$ of $LP_{JMS}(cq,T)$ with the same objective value. We set
\begin{align*}
\alpha'_{i}:=\frac{\alpha_{\lceil i/c\rceil}}{c},\quad d'_{i}=\frac{d_{\lceil i/c\rceil}}{c},\quad 
r'_{j,i}:=\begin{cases} 
\frac{r_{\lceil j/c\rceil,\lceil i/c\rceil}}{c} & \text{if } \lceil j/c\rceil<\lceil i/c\rceil;\\
\frac{\alpha_{\lceil j/c\rceil}}{c} & \text{if }\lceil j/c\rceil=\lceil i/c\rceil.
\end{cases}
\quad \lambda'=\lambda,\quad \forall 1\leq j\leq i\leq cq.
\end{align*}
Observe that $S$ and $S'$ have exactly the same objective value. Let us show that $S'$ is feasible. Clearly all the variables in $S'$ are non-negative since the same holds for $S$. To avoid possible confusion, we use $(x)'$ to denote constraint $(x)$ in $LP_{JMS}(cq,T)$. The following inequalities show that Constraints $\eqref{con:sumdi}'$, $\eqref{con:orderalphai}'$, $\eqref{con:addedConstr}'$, and $\eqref{con:Tbound-new}'$, resp., hold for every feasible choice of the parameters:
$$
\sum_{i=1}^{cq}d'_i=\sum_{\ell=1}^{q}d_i \overset{\eqref{con:sumdi}}{=} 1; \quad\quad
\alpha'_i = \frac{\alpha_{\lceil i/c\rceil}}{c} \overset{\eqref{con:orderalphai}}{\leq} \frac{\alpha_{\lceil (i+1)/c\rceil}}{c}=\alpha'_{i+1}; \quad\quad
r'_{j,j}=\frac{\alpha_{\lceil j/c\rceil}}{c}=\alpha'_j; \quad\quad
\lambda'=\lambda\overset{\eqref{con:Tbound-new}}{\leq} T.
$$
Consider next Constraint $\eqref{con:orderrji}'$. We distinguish three subcases. If $\lceil j/c\rceil<\lceil i/c\rceil$, one has 
$$
r'_{j,i+1}=\frac{r_{\lceil j/c\rceil,\lceil (i+1)/c\rceil}}{c} \overset{\eqref{con:orderrji}}{\leq} \frac{r_{\lceil j/c\rceil,\lceil i/c\rceil}}{c}=r'_{j,i}.
$$
If $\lceil j/c\rceil=\lceil i/c\rceil<\lceil (i+1)/c\rceil$, one has
$$
r'_{j,i+1}=\frac{r_{\lceil j/c\rceil,\lceil (i+1)/c\rceil}}{c} \overset{\eqref{con:orderrji}}{\leq} \frac{r_{\lceil j/c\rceil,\lceil i/c\rceil}}{c}=\frac{r_{\lceil j/c\rceil,\lceil j/c\rceil}}{c} \overset{\eqref{con:addedConstr}}{\leq} \frac{\alpha_{\lceil j/c\rceil}}{c}=r'_{j,i}.
$$
Finally if $\lceil j/c\rceil=\lceil i/c\rceil=\lceil (i+1)/c\rceil$, one has
$$
r'_{j,i+1}=\frac{\alpha_{\lceil j/c\rceil}}{c}=r'_{j,i}.
$$  
Consider now Constraint $\eqref{con:triangleIneq}'$. We distinguish two subcases. If $\lceil j/c\rceil<\lceil i/c\rceil$,
$$
\alpha'_i=\frac{\alpha_{\lceil i/c\rceil}}{c} \overset{\eqref{con:triangleIneq}}{\leq} \frac{r_{\lceil j/c\rceil,\lceil i/c\rceil}+d_{\lceil i/c\rceil}+d_{\lceil j/c\rceil}}{c}=r'_{j,i}+d'_i+d'_j.
$$
Otherwise (i.e., $\lceil j/c\rceil=\lceil i/c\rceil$)
$$
\alpha'_i=\frac{\alpha_{\lceil i/c\rceil}}{c} = \frac{\alpha_{\lceil j/c\rceil}}{c}=r'_{j,i}\leq r'_{j,i}+d'_i+d'_j.
$$
It remains to consider Constraint $\eqref{con:LBlambda}'$. For the considered index $i$, let $i=c \ell +h$ with $\ell\in \{0,q-1\}$ and $h\in \{1,\ldots,c\}$. Notice that 
\begin{equation}\label{lem:jms-split-new:eqn1}
r'_{j,i}= \frac{\alpha_{\lceil i/c\rceil}}{c} =\alpha'_i,\quad\quad \forall c \ell +1 \leq j\leq i=c\ell+h.
\end{equation}
Then
\begin{align*}
& \sum_{j=1}^{i-1} \max\{r'_{j,i} - d'_j,0\} + \sum_{j=i}^{cq} \max\{\alpha'_i-d'_j,0\} \\
= & \sum_{j=1}^{c \ell} \max\{r'_{j,i} - d'_j,0\}+\sum_{j=c\ell +1}^{c \ell +h}\max\{r'_{j,i} - d'_j,0\}+\sum_{j=c\ell +h+1}^{cq} \max\{\alpha'_i-d'_j,0\}\\
\overset{\eqref{lem:jms-split-new:eqn1}}{\leq} & \sum_{j=1}^{c \ell} \max\{r'_{j,i} - d'_j,0\}+\sum_{j=c \ell+1}^{cq} \max\{\alpha'_i-d'_j,0\}\\
= & \sum_{h=1}^{\ell} c\cdot \max\{\frac{r_{h,\ell+1}}{c} - \frac{d_h}{c},0\} + \sum_{h=\ell+1}^{q} c\cdot \max\{\frac{\alpha_{\ell+1}}{c}-\frac{d_{h}}{c},0\}\overset{\eqref{con:LBlambda}}{\leq} \lambda=\lambda'.\qedhere
\end{align*}
\end{proof}

\lemJMSconcavity*
\begin{proof}
We will show that, for any fixed $q$, $\opt_{JMS}(q,T)$ is concave. The claim then follows. Indeed, assume by contradiction that there exist values $0<T_1<T_2$ and $0<\alpha<1$ such that 
$$
\alpha \opt_{JMS}(T_1)+(1-\alpha)\opt_{JMS}(T_2)\geq \opt_{JMS}(T)+\delta
$$ 
where $T=\alpha T_1+(1-\alpha)T_2$ and $\delta>0$. There must exist a finite $q$ such that $\opt_{JMS}(q,T_i)\geq \opt_{JMS}(T_i)-\delta/2$ for $i\in \{1,2\}$. Then we get the contradiction
\begin{eqnarray*}
\opt_{JMS}(T) & \geq & \opt_{JMS}(q,T)\geq \alpha \opt_{JMS}(q,T_1)+(1-\alpha)\opt_{JMS}(q,T_2)\\
                       & \geq & \alpha(\opt_{JMS}(T_1)-\frac{\delta}{2})+(1-\alpha)(\opt_{JMS}(T_2)-\frac{\delta}{2})\\
                       & \geq & \opt_{JMS}(T)+\delta-\frac{\delta}{2} > \opt_{JMS}(T).
\end{eqnarray*}
It remains to show that $\opt_{JMS}(q,T)$ is concave for any given $q\geq 1$. Consider any values $0<T_1<T_2$ and $\alpha\in (0,1)$. Let $T=\alpha T_1+(1-\alpha)T_2$ and $\OPT_{JMS}(q,T_i)$ be some optimal solution for $LP_{JMS}(q,T_i)$. We use $\lambda(q,T_i)$ to refer to the value of $\lambda$ in $\OPT_{JMS}(q,T_i)$ and similarly for the other variables. Consider the solution 
$$
APX_{JMS}(q,T)=\alpha \OPT_{JMS}(q,T_1)+(1-\alpha)\OPT_{JMS}(q,T_2)
$$ 
for $LP_{JMS}(q,T)$. Trivially $APX_{JMS}(q,T)$ if feasible. Indeed, it is the convex combination of two feasible solutions restricted to Constraints \eqref{con:sumdi}-\eqref{con:LBlambda}. Furthermore its value of $\lambda$ satisfies Constraint \eqref{con:Tbound-new} since, by the feasibility of $\OPT_{JMS}(q,T_i)$, 
$$
\lambda =\alpha \lambda(q,T_1)+(1-\alpha)\lambda(q,T_2)\leq \alpha T_1+(1-\alpha)T_2=T.
$$
The value of $APX_{JMS}(q,T)$ is $\alpha \opt_{JMS}(q,T_1)+(1-\alpha)\opt_{JMS}(q,T_2)$, hence $\opt_{JMS}(q,T)$ is at least the latter amount (since $LP_{JMS}(q,T)$ is a maximization LP).   
\end{proof}

\lemjmsUBnew* 
\begin{proof}
Let $S:=(\alpha,d,r,\lambda)$ be an optimal solution to $LP_{JMS}(cq,T)$. We show how to build a feasible solution $S^+:=(\alpha^+,d^+,r^+,\lambda^+)$ for $LP^{+}_{JMS}(q,T)$ with the same objective value. We set
$$
\alpha^+_i:=\sum_{\ell=(i-1)c+1}^{ic}\alpha_\ell,\quad d^+_i=\sum_{\ell=(i-1)c+1}^{ic}d_\ell,\quad r^+_{j,i}:= \sum_{\ell=(j-1)c+1}^{jc}r_{\ell,ic},\quad \lambda^+=\lambda, \quad \forall 1\leq j\leq i\leq q.
$$
Clearly the objective values of $S$ and $S^+$ are identical. Let us show that $S^+$ is a feasible solution. Clearly all the variables in $S^+$ are non-negative since the same holds for $S$. In order to avoid possible confusion, we will use $(x)^+$ instead of $(x)$ to denote the $(x)$-constraint in $LP^{+}_{JMS}(q,T)$. The following inequalities show that Constraints $\eqref{con:sumdi}^+$, $\eqref{con:orderalphai}^+$, $\eqref{con:orderrji}^+$, $\eqref{con:addedConstr}^+$, and $\eqref{con:Tbound-new}^+$, resp., hold for every valid choice of the indexes: 
$$
\sum_{i=1}^{q}d^+_i=\sum_{j=1}^{cq}d_i\overset{\eqref{con:sumdi}}{=}1.
$$
$$
\alpha^+_i=\sum_{\ell=(i-1)c+1}^{ic}\alpha_\ell \overset{\eqref{con:orderalphai}}{\leq} \sum_{\ell=(i-1)c+1}^{ic}\alpha_{\ell+c}=\sum_{\ell=ic+1}^{(i+1)c}\alpha_{\ell}=\alpha^+_{i+1}.
$$
$$
r^+_{j,i+1}=\sum_{\ell=(j-1)c+1}^{jc}r_{\ell,(i+1)c}\overset{\eqref{con:orderrji}}{\leq}\sum_{\ell=(j-1)c+1}^{jc}r_{\ell,ic}=r^+_{j,i}.
$$
$$
r^+_{j,j}=\sum_{h=1}^{c}r_{(j-1)c+h,jc}\overset{\eqref{con:orderrji}}{\leq}\sum_{h=1}^{c}r_{(j-1)c+h,(j-1)c+h}\overset{\eqref{con:addedConstr}}{\leq} \sum_{h=1}^{c}\alpha_{(j-1)c+h}=\alpha^+_j.
$$
$$
\lambda^+=\lambda\overset{\eqref{con:Tbound-new}}{\leq} T.
$$
\fab{Consider next constraint \eqref{con:triangleIneq-plus} that replaces \eqref{con:triangleIneq}:}
\begin{align*}
\alpha^+_i & =\sum_{h=1}^{c}\alpha_{(i-1)c+h} \overset{\eqref{con:triangleIneq}}{\leq} \sum_{h=1}^{c} (r_{(j-1)c+h,(i-1)c+h}+d_{(i-1)c+h}+d_{(j-1)c+h}) \\
& =d^+_i+d^+_j+\sum_{h=1}^{c} r_{(j-1)c+h,(i-1)c+h}  \overset{\eqref{con:orderrji}}{\leq} d^+_i+d^+_j+\sum_{h=1}^{c} r_{(j-1)c+h,\fab{(i-1)}c}=d^+_i+d^+_j+r^+_{j,\fab{i-1}}.
\end{align*}
It remains to consider Constraint $\eqref{con:LBlambda-plus}$ that replaces \eqref{con:LBlambda}. Let us prove some intermediate inequalities
\begin{align}\label{lem:jms-UB:eqn1}
\sum_{j=1}^{t}\max\{r^+_{j,t}-d^+_j,0\} & =\sum_{j=1}^{t}\max\left\{\sum_{\ell=(j-1)c+1}^{jc}r_{\ell,tc}-\sum_{\ell=(j-1)c+1}^{jc}d_\ell,0\right\}\nonumber \\
& \leq \sum_{j=1}^{t}\sum_{\ell=(j-1)c+1}^{jc}\max\{r_{\ell,tc}-d_{\ell},0\}=\sum_{\ell=1}^{tc}\max\{r_{\ell,tc}-d_{\ell},0\}.
\end{align}
\begin{align}\label{lem:jms-UB:eqn2}
\sum_{j=t+1}^{q}\max\{\alpha^+_{t}-d^+_j,0\} & =\sum_{j=t+1}^{q}\max\left\{\sum_{\ell=(j-1)c+1}^{jc}\alpha_{\ell}-\sum_{\ell=(j-1)c+1}^{jc}d_\ell,0\right\}\nonumber \\
& \leq \sum_{j=t+1}^{q}\sum_{\ell=(j-1)c+1}^{jc}\max\{\alpha_{\ell}-d_{\ell},0\}=\sum_{\ell=tc+1}^{cq}\max\{\alpha_{\ell}-d_{\ell},0\}.
\end{align}
Constraint $\eqref{con:LBlambda-plus}$ follows since:
\begin{align*}
& \sum_{j=1}^{i-1} \max\{r^+_{j,i} - d^+_j,0\} + \sum_{j=i}^{q} \max\{\alpha^+_i-d^+_j,0\} \overset{\eqref{lem:jms-UB:eqn1}+\eqref{lem:jms-UB:eqn2}}{\leq}\sum_{\ell=1}^{tc}\max\{r_{\ell,tc}-d_{\ell},0\}+\sum_{\ell=tc+1}^{cq}\max\{\alpha_{\ell}-d_{\ell},0\}\\
\overset{\eqref{con:addedConstr}}\leq & \sum_{\ell=1}^{tc-1}\max\{r_{\ell,tc}-d_{\ell},0\}+\max\{\alpha_{tc}-d_{tc},0\}+\sum_{\ell=tc+1}^{cq}\max\{\alpha_{\ell}-d_{\ell},0\} \overset{\eqref{con:LBlambda}}{\leq } \lambda=\lambda^+.\qedhere
\end{align*}
\end{proof}

\cormodifiedJMSclaim*
\begin{proof}
Recall that $\opt_{JMS}(T)\leq \opt_{JMS}\leq 2$. By Lemma \ref{lem:JMSconcavity}
\begin{eqnarray}
& & \sum_{f^*\in \OPT'}\opt_{JMS}(\frac{\open(f^*)}{d(f^*)})\cdot \frac{d(f^*)}{d(\OPT')}\leq \opt_{JMS}\left(\sum_{f^*\in \OPT'}\frac{d(f^*)}{d(\OPT')}\cdot \frac{\open(f^*)}{d(f^*)}\right) \nonumber \\
& = & \opt_{JMS}(\frac{\open(\OPT')}{d(\OPT')})\leq \opt_{JMS}(T) \overset{Lem. \ref{lem:jms-UB-new}}{\leq} \opt^+_{JMS}(q,T) \label{cor:modifiedJMSclaim:eqn1},
\end{eqnarray}
where the second-last inequality follows since $\opt_{JMS}(T)$ is non-decreasing in $T$. Then 
\begin{align*}
\open(S)+d(S) & \overset{Lem. \ref{lem:modifiedJMSclaim}}{\leq} \open(\OPT)+\sum_{f^*\in \OPT\setminus \OPT'}2d(f^*)+\sum_{f^*\in \OPT'}\opt_{JMS}(\frac{\open(f^*)}{d(f^*)})d(f^*) \\
& \overset{\eqref{cor:modifiedJMSclaim:eqn1}}{\leq} \open(\OPT)+2d(\OPT\setminus \OPT')+\opt^+_{JMS}(q,T)\cdot d(\OPT').\qedhere
\end{align*}
\end{proof}

\section{An Alternative LMP Bound for Small Facility Cost}
\label{sec:LMP_old}

Here we provide an alternative, analytical, way to upper bound the approximation factor of JMS. We define a variant of $LP_{JMS}(q)$ as follows:
\begin{align}
\mbox{max} & \sum_{i=1}^q \alpha_i - \lambda \nonumber \\ 
\mbox{s.t.} 
& \sum_{i=1}^q d_i = 1 \label{eq:zeron} \\
& \forall 1 \leq i < q: \alpha_i - \alpha_{i+1} \leq 0 \label{eq:onen} \\
& \forall 1 \leq j  < i < q: r_{j,i+1} - r_{j,i} \leq 0 \label{eq:twon} \\
& \forall 1 \leq j < i \leq q: \alpha_i - r_{j,i} - d_i - d_j \leq 0 \label{eq:threen} \\ 
& \forall 1 \leq j < i \leq q: r_{j,i} - d_j - g_{i,j} \leq 0 \label{eq:fourn} \\ 
& \forall 1 \leq i \leq j \leq q: \alpha_i - d_j - h_{i,j} \leq 0 \label{eq:fiven} \\ 
& \forall 1 \leq i \leq q: \sum_{j=1}^{i-1} g_{i,j} + \sum_{j=i}^q h_{i,j} - \lambda \leq 0 \label{eq:sixn} \\
& \forall 1 \leq j < q: r_{j, j+1} \leq \alpha_j \label{eq:eightn} \\
& \alpha, d, f, r, g, h \geq 0 \nonumber 
\end{align}
In particular notice that we do not use the variables $r_{j,j}$. Furthermore, we replace $r_{j,j}\leq \alpha_j$ with $r_{j,j+1}\leq \alpha_j$. Clearly the latter constraint is also guaranteed by JMS. Indeed, since $\alpha_j$ is an upper bound on the distance between client $j$ and its first facility and client $j$ only switches to closer facilities, $r_{j, i} \leq \alpha_j$ holds for all possible LP values that correspond to a run of the JMS algorithm.

Consider the variant $LP_{JMS}(q,T)$ of $LP_{JMS}(q)$, for a parameter $T>0$, where we add the following constraint that intuitively captures the condition $\open(f)\leq T\cdot d(f)$
\begin{equation}
\lambda \leq T.\label{eq:sevenn}
\end{equation}
Let $\opt_{JMS}(q,T)$ be the optimal value of $LP_{JMS}(q,T)$, and $\opt_{JMS}(T)=\sup_q\{\opt_{JMS}(q,T)\}$. Observe that $\opt_{JMS}(q,T)\leq \opt_{JMS}(q)$ where the equality holds for any fixed $q$ and $T$ large enough.
We are able to show that, for any finite $T$, $\opt_{JMS}(T)<2$. More precisely, we obtained the following result. 
\begin{restatable}{lemma}{lemjmsmain}
\label{lem:jms-main}
For any $T > 0$, 
$
\opt_{JMS}(T) \leq \min_{z \in [0, 1/3]} ( V(z) + T(M(z) - 1) ),
$ 
where
\begin{align*}
V(z)& = \max\{2 - \frac{z}{1 - z}, 2 - \frac{2z}{1 - z} + \ln ( 1 + \frac{z}{1 - 2z}) + \frac{4z^2}{(1-z)(1-2z)} \}, \mbox{ and}  \\
M(z) - 1 & = \ln(1 + \frac{z}{1 - 2z}) - \frac{z}{1 - z} + \frac{2z^2}{(1-z)(1 - 2z)}.
\end{align*}
\end{restatable}
\noindent
The following simple corollary shows that the improvement with respect to $2$ is $\Theta(\frac{1}{c + T})$ for some constant $c$. 
Note that this is asymptotically tight because \cite{JMS02} proved that $opt_{JMS}(T) \geq 2 - \frac{2}{T + 2}$ for any positive even integer $T$.

\begin{corollary}
For any $T > 0$, 
$\opt_{JMS}(T) \leq 2 - \frac{1}{4(7 + 3T)}.$
\label{cor:jms-general}
\end{corollary}
\begin{proof}
Using the approximation $\ln(1 + \frac{z}{1 - 2z}) \leq \frac{z}{1 - 2z}$, simple calculations show that 
$V(z) \leq 2 - \frac{z - 7z^2}{(1-z)(1-2z)}$ and $M(z) - 1 \leq \frac{z^3}{(1-z)(1-2z)}$. 
Therefore, $V(z) + T(M(z) - 1) \leq 2 - \frac{z-(7+3T)z^2}{(1-z)(1-2z)}\leq 2-(z-(7+3T)z^2)$. Taking $z = \frac{1}{2(7 + 3T)}$ yields the claim. 
\end{proof}

%

It remains to prove Lemma \ref{lem:jms-main}. To that aim we first show that $\opt_{JMS}(T)=\lim_{q\to +\infty}\opt_{JMS}(q,T)$. This is a straightforward consequence of the following lemma.
\begin{restatable}{lemma}{lemjmssplit}
\label{lem:jms-split}
Fix $T > 0$. For any positive integers $q$ and $m$, 
$\opt_{JMS}(q, T) \leq \opt_{JMS}(mq, T)$. 
\end{restatable}
\begin{proof}
Create $qm$ clients $(i, a) \in [q] \times [m]$ with the lexicographical ordering (first coordinate first). Define the values as follows.
\begin{itemize}
\item $\alpha_{i, a} = \alpha_{i} / m$. 
\item $d_{i, a} = d_i / m$.
\item $r_{(j, b), (i, a)}$: if $j < i$, then $r_{(j, b), (i, a)} = r_{j, i} / m$. Otherwise, $i = j$ and $r_{(j, b), (i, a)} = \alpha_j / m$. 
\item As usual, $g_{(i, a), (j, b)} = \max(r_{(j,b),(i,a)} - d_{(j,b)}, 0)$ and $h_{(i, a), (j, b)} = \max(\alpha_{i,a} - d_{(j,b)}, 0)$.
\item $f$ stays the same. 
\end{itemize}
Equations~\eqref{eq:zeron} and~\eqref{eq:onen} are easy to check. \eqref{eq:twon} is also true because new $r_{(j,b), (j, a)} = \alpha_j / m \geq r_{(j,b), (i, c)}$ if $i > j$ (by the assumption).
Consider \eqref{eq:threen} for $(j, b) < (i, a)$, if $j < i$, then it just follows from the old inequality with $i$ and $j$. 
If $i = j$ and $b < a$, since $r_{(j, b), (i,a)} = \alpha_j / m$, 
\[
\alpha_{(i,a)} - r_{(j, b), (i,a)} - d_{i, a} - d_{j, b} = -d_{i,a} - d_{j, b} \leq 0. 
\]
Finally, let us look at~\eqref{eq:sixn}. Fix $(i, a)$. We first compute relevant $g$ and $h$. 
\begin{itemize}
\item For $(j, b) < (i, a)$, if $j < i$, then $g_{(i, a), (j, b)} = g_{i,j} / m$. Otherwise, $i = j$ and $g_{(i, a), (j, b)} = \max(r_{(j, b), (i, a)} - d_j, 0) = \max(\alpha_j - d_j, 0) / m$. 
\item For $(j, b) \geq (i, a)$, $h_{(i, a), (j, b)} = \max(\alpha_{i, a} - d_{j,b}, 0) = h_{i, j} / m$. Note that if $i = j$, it is again exactly $\max(\alpha_j - d_j, 0) / m$. 
\end{itemize}

Since $g_{(i, a), (i, b)} = h_{i, i} / m = h_{(i, a), (i, c)}$ and $b < a \leq c$, so 

\begin{align*}
& \sum_{(j, b) < (i, a)} g_{(i, a), (j, b)} 
+ 
\sum_{(j, b) \geq (i, a)} h_{(i, a), (j, b)}  \\
=&
 \sum_{(j, b) : j < i} g_{(i, a), (j, b)} 
+ 
\sum_{(j, b) : j \geq i} h_{(i, a), (j, b)} \\
=&  \sum_{j < i} g_{i, j} 
+ 
\sum_{j \geq i} h_{i, j} \leq f. 
\end{align*}

\end{proof}


The above lemma motivated us to study the continuous version of the factor-revealing LP and its dual. 
\cite{JMS02} showed that $\opt_{JMS} \leq 2$ by constructing a feasible solution of value $2 - 1/q$ for the dual of $LP_{JMS}(q)$ for every $q \geq 1$. 
For $LP_{JMS}(q, T)$, letting $\{ A_{i,j} \}_{1 \le j < i \leq k}$, $\{ B_{i,j} \}_{1 \le j < i \leq k}$, $\{ C_{i,j} \}_{1 \leq i \leq j \leq k}$, $\{ N_i \}_{1 \leq i \leq q}$, $V$, and $M$ be the dual variables corresponding to~\eqref{eq:threen},~\eqref{eq:fourn},~\eqref{eq:fiven},~\eqref{eq:sixn},~\eqref{eq:zeron}, and~\eqref{eq:sevenn} and optimizing the other dual variables depending on them, we obtain the dual LP called $DP_{JMS}(q,T)$ shown below. 
The \emph{Continuous Dual LP}, dubbed $CDP_{JMS}(T)$, is obtained by letting $q \rightarrow \infty$ and 
replacing the counting measure on $\{ 1, \dots, q \}$ by the Lebesgue measure on $[0, 1]$; 
so variables are represented as functions and sums become integrals. 
In particular, we have $A, B : L \to \RP$, $C : U \to \RP$ where 
$L := \{ (i,j) \in [0,1]^2 : j \leq i \}$, and $U := \{ (i,j) \in [0,1]^2 : j \geq i \}$.
We use subscripts to denote function arguments (e.g., $A_{i,j} = A(i, j)$). 
See Appendix~\ref{sec:jms-setup} for a full derivation of both $DP_{JMS}(q,T)$ and $CDP_{JMS}(T)$.

{\footnotesize
\begin{align*}
\mbox{min }\quad & V - (M-1)T \hspace{2cm}(DP_{JMS}(q, T))
& \mbox{min }\quad & V - (M-1)T \hspace{2cm}(CDP_{JMS}(T)) \\
\mbox{s.t. }\quad & \forall i \in [q]: \sum_{j=1}^{i-1} A_{i, j}  + \sum_{j=i}^q C_{i,j}  = 1 
& \mbox{s.t. }\quad & \forall i \in [0, 1]: \int_{j=0}^i A_{i, j}  + \int_{j=i}^1 C_{i,j}  = 1 \\
& \forall j \in [q]:
& & \forall j \in [0, 1]:\\
& \sum_{i=j+1}^q A_{i,j}  + \sum_{i=1}^{j-1} A_{j,i}  + \sum_{i=j+1}^q B_{i,j}  + \sum_{i=1}^j C_{i,j}  \leq V 
& & \int_{i=j}^1 A_{i,j}  + \int_{i=0}^j A_{j,i}  + \int_{i=j}^1 B_{i,j}  + \int_{i=0}^j C_{i,j}  \leq V \\
& \forall \, 1 \leq j \leq i \leq q: \sum_{\ell=j+1}^i B_{\ell, j} \geq \sum_{\ell=j+1}^i A_{\ell, j},  
& & \forall \, 0 \leq j \leq i \leq 1: \int_{\ell=j}^i B_{\ell, j} \geq \int_{\ell=j}^i A_{\ell, j},  \\
& N_i = \max \bigg( \max_{j \in [1, i-1]} B_{i,j}, \quad \max_{j \in [i+1, q]} C_{i,j} \bigg), 
& & N_i = \max \bigg( \max_{j \in [0, i]} B_{i,j}, \quad \max_{j \in [i, 1]} C_{i,j} \bigg), \\ 
& M = \sum_{i=1}^q N_i 
& & M = \int_{i=0}^1 N_i 
\end{align*}
}

Let $cdopt_{JMS}(T)$ be the optimal value of $CDP_{JMS}(T)$. 
Note that setting $A(\cdot)$, $B(\cdot)$, $C(\cdot)$, and $N(\cdot)$ to the constant 1 function and letting $V = 2$ and $M = 1$, one obtains a feasible solution to $CDP_{JMS}(T)$, showing that $cdopt_{JMS}(T)\leq 2$ for any $T$. We are able to construct a better dual solution for finite $T$, showing that
$$
cdopt_{JMS}(T)\leq \min_{z \in [0, 1/3]}(V(z) + (M(z) - 1)T).
$$

\noindent The mentioned solution to $CDP_{JMS}(T)$, parameterized by $z$, is constructed by dividing the overall domain $[0, 1] \times [0, 1]$ into several regions defined by the lines including $x = z, y = z, x = 1-z, y = 1-z, x = y$ and by assigning different values to $A, B$, and $C$ for each such region. In order to obtain a feasible solution for $DP_{JMS}(q, T)$, we discretize our solution for $CDP_{JMS}(T)$. In more detail, we use the natural strategy of partitioning $[0, 1] \times [0, 1]$ into squares of side length $1/q$ and integrating the continuous variables over each square to obtain the appropriate discrete variables.
(E.g., $A_{q, 1}$ of $DP_{JMS}(q, T)$ is defined to be $\int_{i=(q-1)/q}^1 \int_{j=0}^{1/q} A_{i, j}$ where $A_{i, j}$ is taken from our $CDP_{JMS}(T)$ solution).
When $z$ is an integer multiple of $1/q$, the lines defining the small squares ``align nicely'' with the lines defining the continuous dual solution, 
so one can obtain the following lemma. 

\begin{restatable}{lemma}{lemjmsdisc}
\label{lem:jms-disc}
For any $T > 0$, $q\in \N$ and $z \in [0, 1/3]$ where $z = d/q$ for some integer $d$, 
$\opt_{JMS}(q, T) \leq V(z) + T(M(z) - 1)$. 
\end{restatable}
\noindent 
We now have all the ingredients to prove Lemma \ref{lem:jms-main}.
\begin{proof}[Proof of Lemma \ref{lem:jms-main}]
Let $z^* = \argmin_{z \in [0, 1/3]} (V(z) + T(M(z) - 1))$. For any $\eps > 0$, one can choose $c, d \in \N$ such that $z' := \frac{d}{cq}\in [0,1/3]$ is close enough to $z^*$ to satisfy
$V(z') + T(M(z') - 1)< V(z^*) + T(M(z^*) - 1) + \eps$.
$$
\opt_{JMS}(q,T)\overset{Lem. \ref{lem:jms-split}}{\leq} \opt_{JMS}(cq,T)\overset{Lem. \ref{lem:jms-disc}}{\leq} V(z') + T(M(z') - 1)< V(z^*) + T(M(z^*) - 1)+\eps. 
$$
Since this holds for any $\eps$, we have $\opt_{JMS}(q,T)\leq \min_{z \in [0, 1/3]} ( V(z) + T(M(z) - 1) )$, hence the claim.
\end{proof}

It remains to prove Lemma ~\ref{lem:jms-disc}. As mentioned before, the proof of this lemma first considers the {\em continuous case} where the set of clients becomes a continuous set $[0, 1]$ followed by a discretization step. 
Section~\ref{sec:jms-setup} and~\ref{sec:jms-cont} will discuss the factor-revealing LP for the continuous case. 

\subsection{Reduction to dual continuous LP}
\label{sec:jms-setup}
In this subsection, we drive the dual LP $DP_{JMS}(q, T)$ and the continuous dual LP $CDP_{JMS}(T)$ from the factor-revealing LP $LP_{JMS}(q, T)$ introduced before.

We will primarily use dual variables for~\eqref{eq:threen},~\eqref{eq:fourn}, and~\eqref{eq:fiven}. 
(Dual variables for other constraints will be implicitly defined depending on them, but we will not use \eqref{eq:onen} at all.)
Let $\{ A_{i,j} \}_{1 \le j < i \leq q}$, $\{ B_{i,j} \}_{1 \le j < i \leq q}$, $\{ C_{i,j} \}_{1 \leq i \leq j \leq q}$ be the dual variables corresponding to~\eqref{eq:threen},~\eqref{eq:fourn}, and~\eqref{eq:fiven}. 
Note that these are the only variables that contain $\alpha_i, d_i$ primal variables apart from~\eqref{eq:zeron} and~\eqref{eq:onen}. 
Next, we will go over each dual constraint and give some conditions that our dual solutions need to satisfy: 

\begin{itemize}
\item $\alpha$ and $d$: 
In our dual solution, we will satisfy the constraints induced by $\alpha_i$'s with equality, which implies
\begin{equation}
\forall i: \sum_{j=1}^{i-1} A_{i, j} + \sum_{j=i}^q C_{i,j} = 1 \label{eq:alpha}
\end{equation}

For $d_i$'s, we have 
\begin{equation}
\forall j: 
\sum_{i=j+1}^q A_{i,j} + \sum_{i=1}^{j-1} A_{j,i} + 
\sum_{i=j+1}^{q} B_{i,j} + \sum_{i=1}^{j} C_{i,j} \leq V \label{eq:d}
\end{equation}
where $V$ will be the dual value corresponding to \eqref{eq:zeron}. (Note that $i,j$ is swapped only for the second term because~\eqref{eq:threen} has both $d_i$ and $d_j$.) 

\item $r$: 
The coefficient of $r_{j,i}$ is increased by $B_{i,j}$ and decreased by $A_{i,j}$. Also,~\eqref{eq:twon} increases the coefficient of $r_{j,i+1}$ and decrease the coefficient of $r_{j,i}$. (Intuitively, for same $j$, we can increase the coefficient for higher $i$ and decrease the coefficient for smaller $i$.)
Therefore, as long as we satisfy 
\begin{align}
& \forall 1 \leq j < q: \mbox{ the sequence } (B_{j+1,j}, ..., B_{q,j}) \mbox{ ``dominates'' } (A_{j+1,j}, ..., A_{q,j}) \nonumber \\
\Leftrightarrow \quad & \forall 1 \leq j < i \leq q: \sum_{\ell=j+1}^i B_{\ell,j} \geq \sum_{\ell=j+1}^i A_{\ell,j}, 
\label{eq:r}
\end{align}
then we can make sure that the dual constraints for $r_{j,i}$'s are satisfied. 

\item $g$ and $h$: 
Once we fix $A, B, C$, then for each $i \in [q]$, the dual variable corresponding to the $i$th equation of~\eqref{eq:sixn} should be 
\[
N_i := \max \bigg( \max_{j=1}^{i-1} B_{i,j}, \quad \max_{j=i}^q C_{i,j} \bigg).
\]
due to the constraints corresponding to $g, h$.

\item $f$: 
Let $M := \sum_i N_i$. (We will make sure $M \geq 1$ always.) Then from the dual constraint for $f$, we will have the dual variable for $f \leq T$ exactly $(M - 1)$. 
And the dual objective function value is $V + T(M-1)$. 
\end{itemize}

So that's the rule of the game: find $A, B, C$ that satisfy~\eqref{eq:alpha},~\eqref{eq:d}, and~\eqref{eq:r} to minimize the value $V + T(M-1)$. 
The JMS gave every $A,B,C$ variable $1/q$, so that $V = 2-1/q$ and $M = 1$. 

To ignore the diagonal issues, let us consider the {\em continuous setting} where for $A,B,C$, we divide each index by $q$ and multiply each value by $q$. As $q$ goes to infinity, then we are basically in the setting where we want to give a map $A, B : L \to \RP$ and $C : U \to \RP$ where 
\begin{align*}
L := \{ (i,j) \in [0,1]^2 : j \leq i \} \\
U := \{ (i,j) \in [0,1]^2 : j \geq i \}
\end{align*}
Then the three conditions~\eqref{eq:alpha},~\eqref{eq:d}, and~\eqref{eq:r} will be translated as

\begin{equation}
\forall i \in [0, 1]: \int_{j=0}^i A_{i, j}  + \int_{j=i}^1 C_{i,j}  = 1 \label{eq:alpha:c}
\end{equation}

\begin{equation}
\forall j \in [0, 1]: 
\int_{i=j}^1 A_{i,j}  + \int_{i=0}^j A_{j,i}  + 
\int_{i=j}^1 B_{i,j}  + \int_{i=0}^j C_{i,j}  \leq V \label{eq:d:c}
\end{equation}

\begin{align}
&\forall j \in [0, 1]: \mbox{ the function } B(\cdot, j) \mbox{ ``dominates'' } A(\cdot, j) \nonumber \\
\Leftrightarrow \quad & \forall \, 0 \leq j \leq i \leq 1: \int_{\ell=j}^i B_{\ell, j} \geq \int_{\ell=j}^i A_{\ell, j}, 
\label{eq:r:c}
\end{align}
and 
\[
N_i := \max \bigg( \max_{j \in [0, i]} B_{i,j}, \quad \max_{j \in [i, 1]} C_{i,j} \bigg), \qquad M := \int_{i=0}^1 N_i, 
\]
\noindent
and the objective value is $V - (M-1)T$. Note that making $A$, $B$, $C$ the constant 1 functions satisfy every constraint and give the value $2$ for any $T$. 
In Section~\ref{sec:jms-cont}, we construct a dual solution for this continuous dual program.
In Section~\ref{sec:jms-disc}, we will show how to convert it back to a dual solution for the discrete dual LP. 

\subsection{Solution to dual continuous LP}
\label{sec:jms-cont}
Let $z \in [0, 1/3]$ be a small constant, $y := z^2/(1-z)$, and 
\begin{align*}
V(z)& = \max\{2 - \frac{z}{1 - z}, 2 - \frac{2z}{1 - z} + \ln ( 1 + \frac{z}{1 - 2z}) + \frac{4z^2}{(1-z)(1-2z)} \}, \mbox{ and}  \\
M(z) - 1 & = \ln(1 + \frac{z}{1 - 2z}) - \frac{z}{1 - z} + \frac{2z^2}{(1-z)(1 - 2z)}.
\end{align*}
We are ready to construct a slightly better dual solution parameterized by $z$. 

\begin{itemize}
\item $A(i, j)$ is equal to
\begin{itemize}
\item $1/(1-z)$ when $i \in [z, 1-z]$, $j \in [z, i]$. 
\item $1/z$ when $i \in [1 - z, 1 - y], j \in [0, z]$. 
\item $1/(0.5-z)$ when $i \in [1 - y, 1], j \in [z, 0.5]$. 
\item $0$ otherwise. 
\end{itemize}

\item $B(i, j)$ is equal to 
\begin{itemize}
\item $1/(1-z)$ when $i \in [z, 1-z], j \in [0, i]$. 
\begin{itemize}
\item But, when $i \in [1-z- y, 1-z]$ and $j \in [z, 0.5]$, let $B(i,j) = 1/(1-z) + 1/(0.5-z)$. 
\end{itemize}
\item $0$ otherwise. 
\end{itemize}

\item $C(i, j)$ is equal to
\begin{itemize}
\item $1 / (1 - z - i)$ when $i \in [0, z]$ and $j \in [i, 1-z]$. 
\item $1/(1-z)$ when $i \in [z, 1-z]$ and $j \in [i, 1]$. 
\item $0$ otherwise. 
\end{itemize}
\end{itemize}

Let us first check~\eqref{eq:alpha:c}
\begin{equation*}
\forall i \in [0, 1]: \int_{j=0}^i A_{i, j}  + \int_{j=i}^1 C_{i,j}  = 1.
\end{equation*}
\begin{itemize}
\item $i \in [0, z]$: $0$ from $A$, $1$ from $C$, for each $i$, $C(i, j) = 1/(1-z-i)$ for $j \in [i, 1-z]$. 
\item $i \in [z, 1-z]$: $(i-z)/(1-z)$ from $A$ and $(1-i)/(1-z)$ from $C$. 
\item $i \in [1-z, 1-y]$: $1$ from $A$ (value $1/z$, interval length $z$). 
\item $i \in [1-y, 1]$: $1$ from $A$ (value $1/(0.5-z)$, interval length $0.5-z$). 
\end{itemize}

Let us check~\eqref{eq:r:c}
\begin{equation*}
\forall j \in [0, 1]: B(\cdot, j) \mbox{ ``dominates'' } A(\cdot, j), 
\label{eq:r:c*}
\end{equation*}
\begin{itemize}
\item $j \in [0, z]$: $A_{i,j} = 1/z$ when $i \in [1-z, 1-y]$. $B(i, j) = 1/(1-z)$ when $i \in [z, 1-z]$: For both $A$ and $B$, the integral is 
\begin{equation}
\frac{1}{z} \cdot(z -  y) = 
\frac{1}{z} \cdot(z -  \frac{z^2}{1-z}) = 
\frac{1-2z}{1-z},
\label{eq:vertical1}
\end{equation}
and clearly $B(\cdot, j)$ dominates $A(\cdot, j)$. 

\item $j \in [z, 0.5]$: $A_{i,j} = 1/(1-z)$ when $i \in [j, 1-z]$ and $1/(0.5-z)$ when $i \in [1-y, 1]$. 
$B_{i,j} = 1/(1-z)$ when $i \in [j, 1-z-y]$ and $1/(1-z) + 1/(0.5-z)$ when $i \in [1-z-y, 1-z]$. For both $A$ and $B$, 
the integral is 
\begin{equation}
\frac{1-z-j}{1-z} + \frac{y}{0.5-z}
\label{eq:vertical2}
\end{equation}
and clearly $B(\cdot, j)$ dominates $A(\cdot, j)$. 

\item $j \in [0.5, 1-z]$: $A_{i, j} = B_{i,j} = 1/(1-z)$ for $i \in [j, 1-z]$ and $0$ otherwise. 
The sum is 
\begin{equation}
\label{eq:vertical3}
(1 - z - j) / (1- z)
\end{equation}
for both $A$ and $B$. 

\item $j \in [1-z, 1]$: Both $A$ and $B$ are all zero here. 
\end{itemize}

We finally check~\eqref{eq:d:c}.

\begin{equation*}
\forall j \in [0, 1]: 
\int_{i=j}^1 A_{i,j}  + \int_{i=0}^j A_{j,i}  + 
\int_{i=j}^1 B_{i,j}  + \int_{i=0}^j C_{i,j}  \leq V
\end{equation*}

\begin{itemize}

\item For $j \in [1-z, 1]$: Both $A(\cdot, j)$ and $B(\cdot, j)$ are zero here. 
Since $C(j, \cdot) = 0$ here, $\int_{i} A(j, i) = 1$, 
$C(i, j)$ is $1/(1-z)$ when $i \in [z, 1-z]$. So the total sum is 
\[
\int_{i=0}^j A_{j,i} +  \int_{i=0}^j C_{i,j} = 1 + (1 - 2z) / (1-z) = 2 - \frac{z}{1 - z} \leq V(z).
\]

\item For $j \in [0.5, 1 - z]$: 
By~\eqref{eq:vertical3}, 
\[
\int_{i=j}^1 A_{i,j} + \int_{i=j}^1 B_{i,j} = 2(1-z-j)/(1-z).
\]
$C(i, j) = 1/(1-z)$ for $i \in [z, j]$ and $1/(1-z-i)$ for $i \in [0, z]$. 
Finally, $A(j, i) = 1/(1-z)$ for $i \in [z, j]$. Therefore, the total sum is 
\begin{align*}
&\int_{i=j}^1 A_{i,j}  + \int_{i=0}^j A_{j,i}  + 
\int_{i=j}^1 B_{i,j}  + \int_{i=0}^j C_{i,j}
= \frac{2(1-z-j)}{1-z} + \frac{2(j - z)}{1-z} + \int_{i=0}^z \frac{1}{1-z-i} \\
= &\frac{2(1-2z)}{1-z} + (\log(1-z) - \log (1-2z)) =
 \frac{2-4z}{1-z} + (\log(1 + \frac{z}{1-2z}))
\leq V(z).
\end{align*} 

\item For $j \in [z, 0.5]$: By~\eqref{eq:vertical2},
\[
\int_{i=j}^1 A_{i,j} + \int_{i=j}^1 B_{i,j} = \frac{2(1-z-j)}{1-z} + \frac{2y}{0.5-z}, 
\]
which is $\frac{2y}{0.5-z} = \frac{4z^2}{(1 - z)(1 - 2z)}$ more than the $j \in [0.5, 1-z]$ case. Everything else is identical. It is still at most $V(z)$. 

\item For $j \in [0, z]$: From~\eqref{eq:vertical1}, 
\[
\int_{i=j}^1 A_{i,j} + \int_{i=j}^1 B_{i,j} = \frac{2(1-2z)}{1 - z} = 2 - \frac{2z}{1-z}. 
\]
And $A_{j, i} = 0$. Therefore, 
\begin{align*}
&\int_{i=j}^1 A_{i,j}  + \int_{i=0}^j A_{j,i}  + 
\int_{i=j}^1 B_{i,j}  + \int_{i=0}^j C_{i,j} \\
=& 2 - \frac{2z}{1-z} + \int_{i = 0}^j \frac{1}{1-z-i}
= 2 - \frac{2z}{1-z} + (\log(1-z) - \log(1-z-y)) \\
\leq& 2 - \frac{2z}{1-z} + (\log(1-z) - \log(1-2z)) 
\leq V(z).
\end{align*}
\end{itemize}

Finally, let us compute $N_i$ and $M$. 
\begin{itemize}
\item For $i \in [0, z]$, $N_i := 1 / (1 - z - i)$. 
\item For $i \in [z, 1-z-y]$, $N_i := 1 / (1 - z)$. 
\item For $i \in [1-z-y, 1-z]$, $N_i := 1/(1-z) + 1 / (0.5 - z)$. 
\end{itemize}
Therefore, 
\[
M = \bigg( \log(1-z)-\log(1-2z) \bigg) + \frac{1 - 2z - y}{1 - z} + y \bigg( \frac{1}{1-z} + \frac{1}{0.5 - z} \bigg)  = M(z).
\]

\subsection{Conversion to Discrete Dual}
\label{sec:jms-disc}
Given the solutions to the continuous LP, the discretization step shown in the following claim proves Lemma~\ref{lem:jms-disc}. 
\begin{claim}
Fix $T > 0$, and consider the solution $A, B, C, V, N, M$ to the continuous dual program constructed in Section~\ref{sec:jms-cont}, which is parameterized by $z \in [0, 1/3]$.
If $z = d/q$ for some integer $d$, the value of the discrete factor-revealing LP with $q$ clients is at most the dual value for the continuous program, which is $V + T(M - 1)$. 
\end{claim}

\begin{proof}
We first construct dual solutions $\{ a_{i,j} \}_{1 \leq j < i \leq q}$, 
$\{ b_{i,j} \}_{1 \leq j < i \leq q}$, $\{ c_{i,j} \}_{1 \leq i \leq j \leq q}$ to the discrete factor-revealing LP; $V$ remains the same. 

They are defined as follows.
For $1 \leq i \leq q$, let $s_i := (i-1)/q$ and $t_i := i / q$. 
To simplify notations, we allow integrating $A, B,$ or $C$ over regions they are not defined; their value is $0$ in those regions. 
Also, we always use the standard Lebesgue measure on $\R$ and $\R^2$ and omit $dx dy$ notations. 
(Variables $y$ and $x$ correspond to $i$ and $j$ respectively.) 

\begin{itemize}
\item For all $1 \leq j < i \leq q$: $a_{i,j} = q \cdot \big( \int_{y=s_i}^{t_i} \int_{x=s_j}^{t_j} A_{y,x} \big)$. 
\item For all $1 \leq j < i \leq q$: $b_{i,j} = q \cdot \big( \int_{y=s_i}^{t_i} \int_{x=s_j}^{t_j} B_{y,x} \big)$. 
\item For all $1 \leq i \leq j \leq q$, $c_{i,j} = q \cdot \big( \int_{y=s_i}^{t_i} \int_{x=s_j}^{t_j} C_{y,x} \big)$. 
\begin{itemize}
\item If $i = j$, additionally increase $c_{i,j}$ by $q \cdot \big( \int_{y=s_i}^{t_i} \int_{x=s_i}^{t_i} A_{y,x} \big)$. 
\end{itemize}
\end{itemize}
We first check they satisfy~\eqref{eq:alpha},~\eqref{eq:d}, and~\eqref{eq:r}. 
\begin{itemize}
\item \eqref{eq:alpha}: for every $i \in [q]$, 
\begin{align*}
& \sum_{j=1}^{i-1} a_{i, j} + \sum_{j=i}^q c_{i,j} \\
=& q \cdot \bigg( \sum_{j=1}^{i-1} \int_{s_i}^{t_i} \int_{s_j}^{t_j} A_{y,x}  + \sum_{j=i}^q \int_{s_i}^{t_i} \int_{s_j}^{t_j} C_{y,x}  + \int_{s_i}^{t_i} \int_{s_i}^{t_i} A_{y, x}  \bigg) \\
=& q \cdot \int_{y=s_i}^{t_i} \bigg(\int_{x=0}^y A_{y,x} + \int_{x=y}^1 C_{y,x} \bigg) \\
=& 1,
\end{align*}
where the last equality follows from~\eqref{eq:alpha:c}. 

\item \eqref{eq:d}: for all $j \in [q]$, first note that 
\begin{align*}
& \frac{1}{q} \cdot \bigg( \sum_{i=j+1}^q a_{i,j} + \sum_{i=1}^{j-1} a_{j,i} + 
\sum_{i=j+1}^{q} b_{i,j} + \sum_{i=1}^{j} c_{i,j}  \bigg) \\
=& \sum_{i=j+1}^q \int_{s_i}^{t_i} \int_{s_j}^{t_j} A_{y,x} + \sum_{i=1}^{j-1} \int_{s_i}^{t_i} \int_{s_j}^{t_j} A_{x,y} + 
\sum_{i=j+1}^{q} \int_{s_i}^{t_i} \int_{s_j}^{t_j} B_{y,x} + \sum_{i=1}^{j} \int_{s_i}^{t_i} \int_{s_j}^{t_j} C_{y,x}  
+ \int_{s_j}^{t_j} \int_{s_j}^{t_j} A_{y,x},
\end{align*}
where the last $\int_{s_i}^{t_i} \int_{s_j}^{t_j} A_{y,x}$ term appears because of the 
additional increases for $c_{j, j}$. This is at most 
\begin{align*}
&  \int_{x=s_j}^{t_j} \int_{y = x}^{1} A_{y,x} + \int_{x=s_j}^{t_j} \int_{y=0}^{x} A_{x,y} + 
\int_{x=s_j}^{t_j} \int_{y = x}^{1} B_{y,x} + \int_{x=s_j}^{t_j} \int_{y = 0}^{x} C_{y,x} \\
\leq & \int_{x=s_j}^{t_j} 
\bigg( 
\int_{y = x}^{1} A_{y,x} + \int_{y=0}^{x} A_{x,y} + \int_{y = x}^{1} B_{y,x} + \int_{y = 0}^{x} C_{y,x}
\bigg) 
\leq \frac{V}{q}. 
\end{align*}
where the last inequality follows from~\eqref{eq:d:c}.

\item \eqref{eq:r}: we use~\eqref{eq:r:c} that proved $B(\cdot, x)$ dominates $A(\cdot, x)$ for every $x \in [0, 1]$, 
which almost immediately implies~\eqref{eq:r}. 
The only worry is that we {\em ignored} diagonals; for $j \in [q]$, since $a_{j, j}$ and $b_{j, j}$ are not variables,
the integrals $a'_j := \int_{s_j}^{t_j} \int_{s_j}^{t_j} A_{y, x}$ and $b'_j := \int_{s_j}^{t_j} \int_{s_j}^{t_j} B_{y, x}$
were not added to any discrete variable, and it can possibly violate \eqref{eq:r} if the latter is strictly larger than the former.
So below we check that $a'_j$ and $b'_j$ are exactly same for every $j \in [q]$. 
Here we use that $z = d/q$ for some integer $d$. 

\begin{itemize}
\item $1 \leq j \leq d$: both $a'_j$ and $b'_j$ are zero. 
\item $d <  j \leq q - d$: both $a'_j$ and $b'_j$ are $(1/q^2) \cdot (1/(1-z))$. 
\item $q - d < j \leq q$: both $a'_j$ and $b'_j$ are zero. 
\end{itemize}
\end{itemize}

Let $m_i = \max\big( \max_{j=1}^{i-1} b_{i,j}, \, \, \max_{j=i}^q c_{i,j} \big)$ for $i \in [q]$ and $N_y = \max\big( \max_{x \in [0, y]} B_{y,x}, \, \, \max_{x \in [y, 1]} C_{y, x} \big)$ for $y \in [0, 1]$. 
We show that 
\begin{equation}
m_i \leq \int_{y=s_i}^{t_i} N_y,
\label{eq:cont-disc-m}
\end{equation}
which implies that the same $M = \int_{y=0}^1 N_y$ is an upper bound of $\sum_{i=1}^q m_i$ and finishes the proof.
By definition, for $i \in [q]$, $m_i = b_{i, j}$ for $j < i$ or $c_{j ,i}$ for $j \geq i$. 
If $m_i = b_{i, j}$ for $j < i$, 
\[
m_i = b_{i,j} = q\cdot \int_{y=s_i}^{t_i} \int_{x=s_j}^{t_j} B_{y,x} \leq \int_{y=s_i}^{t_i} N_y. 
\]
The other case $m_i = c_{i, j}$ can be shown similarly. 
The only possible worry is that $c_{i, i}$ for $i \in [d + 1, q - d]$ added an additional integral $\int_{s_i}^{t_i} \int_{s_i}^{t_i} A_{y, x}$.
But in those regions all $A$, $B$, and $C$ have the same value $1 / (1-z)$, and $A$ and $C$ are defined only in the lower triangle and upper triangle respectively, 
$c_{i, i} = c_{i, i+1} = \dots c_{i, q} = 1/(1-z)$ for $i \in [d+1, q - d]$. Therefore, adding an additional integral to $c_{j, j}$ does not increase $m_i$ and~\eqref{eq:cont-disc-m} still holds. 
\end{proof}

\section{Proof of Theorem \ref{thr:localSearch}}
\label{sec:localSearchOmitted}

In this section we prove Theorem \ref{thr:localSearch}.
We assume that the size of the swaps done by the algorithm is $\Delta = (2/\eps)^{2/\eps^7}$.
We split the proof in two parts. In Section \ref{sec:localSearchOmitted:S2} we will consider
the simpler case $|S'|>k$, which also assumes swaps of smaller size.
This helps to introduce part of the ideas. In Section \ref{sec:localSearchOmitted:S1} we will consider the more complex case
$|S'| \le k$.
We start with the following lemma showing that local search can be made polynomial-time.
\begin{lemma}
  For any $\eps > 0$,
  the results of Lemmas~\ref{lem:localSearchBound_S2} and \ref{lem:improvedS1} can be obtained in polynomial time up to
  losing an additive factor of $O(\eps (\opt + d'))$.
  \label{lem:discrete_localsearch}
\end{lemma}
\begin{proof}
  We apply the following standard preprocessing of the graph. Let $\eta$ be the connection cost of the optimum solution.
  Our algorithm partitions the input metric by creating a graph whose vertex set is the set of points of the input metric space and
  whose edge set is the set of all pairs of
  points at distance at most $\eta$;
  and applies the local search algorithm independently in each subinstance defined by the set of points that are in the same connected component.

  Since the optimum connection cost is at most $\eta$, it is immediate that any optimum clustering is such that no two points of different connected
  components are in the same cluster, and so, obtaining the approximation guarantees of Lemmas~\ref{lem:localSearchBound_S2}\ref{lem:improvedS1} for
  each subinstance is enough to recover the same approximation guarantees for the whole instance.

  Next, for each subinstance of size $n'$, let us assume that the minimum distance between any pair of points is $\eps\eta/n'^2$ -- this can
  be achieved up to losing a factor $(1+\eps)$
  in the connection cost of the optimum solution by ``contracting'' each pair of points at distance smaller than $\eps\eta/n'^2$.
  Moreover, the maximum distance between any two points of the subinstance
  is at most $n' \cdot \eta$. Finally, we can assume without loss of generality that the opening cost is at most $2n'^2 \cdot \eta$ (since otherwise
  we know that the optimum solution opens at most 1 center and we can solve the instance optimally) and at least $\eps^2\eta/n'^2$
  (otherwise, since the minimum
  distance between a pair of points is $\eps\eta/n'^2$,
  we can obtain a $(1+\eps)$-approximate solution by opening a facility for each input client, at the candidate center location that is the closest).

  We thus have that for each subinstance on $n'$ points, the minimum cost is at least $\eps^2\eta/n'^2$ and the maximum is $\eta n'^2$.
  The local search is run such that a new solution
  is picked if its cost improves the cost of the current solution by a factor at least $(1+\eps^3/n'^5)$. Since the initial solution has cost at most
  $\eta n'^2$ and the final solution at most $\eps^2\eta/n'^2$, the number of steps is at most $O(n'^9\eps^{-5})$.
  Finally, two solutions of cost at least $\eps^2\eta/n'^2$ and at most $\eta n'^2$ differ in cost by a multiplicative factor at most
  $(1+\eps^3/n'^5)$ which is an additive factor of at most $\eps^3/n'^5 \cdot \eta n'^2 = \eps^3 / n'^3 \le \eps/n \cdot  \eps^2\eta/n'^2$.
  Since the analysis of Lemmas~\ref{lem:localSearchBound_S2}\ref{lem:improvedS1} uses at most $k$ potential swaps, the difference in cost between the analysis of the algorithm
  assuming it reaches a local optimum or a solution which is within a factor $(1+\eps^3/n'^5)$ is indeed at most $\eps \eta$ as desired.

  To conclude, the algorithm is run with $\eps^{-1} \log \frac{\text{maximum distance}}{\text{minimum distance}}$ different estimates of $\eta$,
  namely, all the powers of $(1+\eps)$ times the
  minimum distance between any pair of points, and the best solution is returned.
  \end{proof}

\subsection{Case $|S'|> k$.}
\label{sec:localSearchOmitted:S2}
We first describe a family of feasible swap pairs $(A_i,B_i)$ for $S'$, with $A_i \subseteq S'$ and $B_i \subseteq \OPT$.
Let $k' = S'$ and $d'$ be the service cost of $S'$.
The upper bound derives from the fact that, by the local optimality of $S'$, for each such pair $(A_i,B_i)$ one has
\[
\lambda k'+d' \leq \lambda |S'\setminus A_i \cup B_i|+d(S'\setminus A_i \cup B_i).
\] 
We start by creating (possibly infeasible) swap pairs as follows. 
Each matched facility $f_2\in S'^M$ and the corresponding matched facilities $M(f_2)\in \OPT^M$ define a swap pair $(\{f_2\},M(f_2))$. Each $f^*\in \OPT^L$ is added to the swap pair containing the closest facility $cl(f^*)\in S'$, creating one such pair $(\{cl(f^*)\},\{f^*\})$ if $cl(f^*)$ is not already contained in some swap  pair. Notice that each swap pair now contains exactly one facility $f_2\in S'$ and one or more facilities in $\OPT$: such facility $f_2$ is the \emph{leader} of the corresponding pair.
Notice also that each $f^*\in \OPT$ belongs to precisely one pair. Now we add the remaining lonely facilities $S'^L$ in $S'$ to the swap pairs, enforcing that each pair $(A,B)$ satisfies $|A|\geq |B|$. Notice  that this is always possible since $|S'|\geq |\OPT|$. We also observe that the sum of the discrepancies $disc(A,B)=|B|-|A|$ is initially exactly $k-k'$. Note that these swap pairs may be such that $|A|>1/\eps$, and so may be infeasible. If this happens, we say that the pair is large. We break large pairs into smaller ones as follows. We take any subset $A'\subseteq A$ of $-disc(A,B)=|A|-|B|$ facilities and let $A''=A\setminus A'$. We partition $A''$ and $B$ into sets $A''_1,\ldots,A''_q$ and $B_1,\ldots,B_q$ such that $|A''_i|=|B_i|$ for all $i$ and $|A''_i|=\frac{1}{4\eps}$ for $i<q$. We partition $A'$ into sets $A'_1,\ldots,A'_\ell$ where $|A'_i|=1/(4\eps)$ for $i<\ell$. We replace $(A,B)$ with the swap pairs $(A''_1,B_1),\ldots,(A''_q,B_q),(A'_1,\emptyset),\ldots,(A'_\ell,\emptyset)$. We observe that all these swap pairs are obviously feasible. 

However, we still need  to perform one last modification to the swap pairs. Let $a$ be the leader in $A$, and assume w.l.o.g. that $a\in A''_1$. For $q\geq 2$, when performing the swap $(A''_1,B_1)$, we remove the leader of $A$ without including in the solution all the facilities in $B$. This is problematic in the rest of our analysis. To resolve this issue, if $q\geq 3$, we replace $a$ with a random facility $a'\in A''_2$. Otherwise, we simply remove $a$ from $A''_1$. We remark that in both cases we keep the feasibility of the pair $(A''_1,B_1)$.

Let $(A_1,B_1),\ldots,(A_h,B_h)$ be the resulting overall set of feasible swap pairs. The following three properties, which will be helpful in the following, trivially hold. From the local optimality of $S'$ one has:
\begin{equation}\label{eqn:localOptimality_S2}
\forall i: \lambda k' + d' \leq \lambda (k_2+|B_i|-|A_i|)+d(S'\setminus A_i \cup B_i).
\end{equation}
Each time we decrease by one the cardinality of some set $A_1$ derived from a large pair $(A,B)$ to remove the leader $a$ of that pair, this means that $\ell\geq 2$, and in particular the corresponding set $A'$ satisfies $|A'|\geq \frac{1}{2\eps}$. Notice that the sum of the cardinalities of the sets $A'$ is precisely $k'-k$, and such sets are disjoint. Hence this event can happen at most $2\eps(k'-k)$ times. Thus: 
\begin{equation}\label{eqn:totalDiscrepancy_S2}
k-k'\leq \sum_{i}disc(A_i,B_i)=\sum_i (|B_i|-|A_i|)\leq k-k'+2\eps(k'-k).
\end{equation}
Finally, by construction each facility $f_2\in S'$ is contained in at most $2$ pairs. When $f_2$ belongs to two pairs, this is because $f_2$ was sampled from a set $A_2$ of size $\frac{1}{4\eps}$. Thus
\begin{equation}\label{eqn:probability_S2}
\forall f_2\in S': Pr[\text{$f_2$ is contained in $2$ swap pairs}]\leq 4\eps.
\end{equation}
We are now ready to prove the following upper bound. 

Recall that given $S'$ and $\opt$, we let $\alpha^L:=\opt^L/\opt$, $\alpha^M:=\opt^M/\opt$, $\alpha^{MM}=\opt^{MM}/\opt$, $\beta=d'/\opt$, and $\beta^{MM}=d'^{MM}/\opt$.
Let also $k^L=|\OPT^L|$, $k^M=|\OPT^M|$, $k'^L =|S'^L|$, and $k'^M=|S'^M|$.

\begin{lemma}\label{lem:localSearchBound_S2}
One has
\begin{align*}
& \lambda k'+d' \leq \lambda k+2\eps \lambda(k'-k)+\frac{\delta(1+4\eps)}{1-\delta}(d'^{MM}+\opt^{MM})+3(1+4\eps)\opt^L+(1+4\eps)\opt^M \\
= &\lambda k+\eps\left(2\lambda(k'-k)+\frac{4\delta}{1-\delta}(d'^{MM}+\opt^{MM})+12\opt^L+4\opt^M\right)+\rho^A(\delta,\alpha^L,\beta,\alpha^{MM},\beta^{MM})\opt.
\end{align*}
\end{lemma}
\begin{proof}
Let us consider the swap pairs $(A_1,B_1),\ldots,(A_h,B_h)$ as described before. Let $S(i)=S'\setminus A_i\cup B_i$. One has
\begin{equation}
  h d' \overset{\eqref{eqn:localOptimality_S2}}{\leq} \sum_{i} \lambda (k'+|B_i|-|A_i|)+\sum_i d(S(i))\overset{\eqref{eqn:totalDiscrepancy_S2}}{\leq}
  \lambda(k-k')+2\eps \lambda(k'-k)+\sum_i d(S(i)).
\end{equation}
Therefore it is sufficient to show that
\begin{equation}\label{eqn:desiredBound_S2}
\E[\sum_i d(S(i))]\leq (h-1)d'+\frac{\delta(1+4\eps)}{1-\delta}(d'^{MM}+\opt^{MM})+3(1+4\eps)\opt^L+(1+4\eps)\opt^M.
\end{equation}

We next focus on a specific client $c$. Let $f^*(c)$ and $f'(c)$ be the facilities serving $c$ in  $\OPT$ and $S'$, resp. We use the shortcuts $d'(c)=d(c,S')$ and $\opt(c)=d(c,\OPT)$. The contribution of $c$ to $hd'$ is precisely $h\, dist(c,S')=h\, d'(c)$. Let us bound the expected contribution $\sum_{i}dist(c,S(i))$ of $c$ to $\sum_i d(S(i))$.  Define $j^*(c)$ such that $f^*(c)\in B_{j^*(c)}$. Define also $J'(c)$ as the set of indexes $i$ such that $f'(c)\in A_{i}$. Notice that $J'(c)$ is empty only if $f'(c)$ is the leader of a large group.  Otherwise,  $J'(c)$ has cardinality $2$ with probability at most $4\eps$, and otherwise cardinality $1$. 

We next assume that $J'(c)$ is empty or $J'(c)=\{j''(c)\}$. We later show how to fix the analysis for the case $|J'(c)|=2$. We distinguish a few cases (for each case, we assume that the conditions of previous cases do not apply):
\begin{enumerate}
\item\label{case:leaderSmallGroup} $f'(c)$ is the leader of a large group. Then 
$$
\sum_{i}dist(c,S(i))=\sum_{i\neq j^*(c)}dist(c,S(i))+dist(c,S(j^*(c)))\leq (h-1)d'(c)+\opt(c).
$$
From now on we can assume that $j''(c)$ is defined. 
\item\label{case:swapWithOPT} $j^*(c)=j''(c)$ (i.e., both $f^*(c)$ and $f'(c)$ belong to the same swap pair). This is similar to the previous case
$$
\sum_{i}dist(c,S(i))=\sum_{i\neq j^*(c)}dist(c,S(i))+dist(c,S(j^*(c)))\leq (h-1)d'(c)+\opt(c).
$$
\item\label{case:matched} $f^*(c)\in \OPT^M$. Let $f^*(c)$ be matched with $f_1(c)\in S'^M$, and $C'(f^*(c))$ be the clients served by $f^*(c)$ which are also served by $f_1(c)$. Let also $C(f^*(c))$ be all the clients served by $f^*(c)$. We observe that $|C'(f^*(c))|\geq (1-\delta)|C(f^*(c))|$ by the definition of matching, and $c\in C(f^*(c))\setminus C'(f^*(c))$ since we are not  in case \ref{case:swapWithOPT}.  
We have the upper bound 
\begin{align*}
dist(c,S(j''(c))) & \leq dist(c,f_1(c))\leq dist(c,f^*(c))+dist(f^*(c),f_1(c))\\
& \leq \opt(c)+\frac{1}{|C'(f^*(c))|}\sum_{c'\in C'(f^*(c))}(dist(f^*(c),c')+dist(c',f_1(c)))\\
& = \opt(c)+\frac{1}{|C'(f^*(c))|}\sum_{c'\in C'(f^*(c))}(\opt(c')+d'(c'))
\end{align*}
Let us define $\Delta(c)=\frac{1}{|C'(f^*(c))|}\sum_{c'\in C'(f^*(c))}(\opt(c')+d'(c'))$. Then
\begin{align*}
\sum_{i}dist(c,S(i)) & =\sum_{i\neq j''(c)}dist(c,S(i))+dist(c,S(j''(c))) \leq (h-1)d'(c)+\opt(c)+\Delta(c).
\end{align*}
\item\label{case:lonely} $f^*(c)\in \OPT^L$.
Let us upper bound $dist(c,S(j''(c)))$. Let $f_1(c)\in S'$ be the closest facility to $f^*(c)$ (i.e. the leader of the group of $f^*(c)$). By the previous cases it must happen that $f_1(c)\notin A_{j''(c)}$, hence $f_1(c)\in S(j''(c))$. Thus
\begin{align*}
dist(c,S(j''(c))) & \leq dist(c,f_1(c))\leq dist(c,f^*(c))+dist(f^*(c),f_1(c))\leq \opt(c)+dist(f^*(c),f'(c))\\
& \leq \opt(c)+dist(f^*(c),c)+dist(c,f'(c))= d'(c)+2\opt(c).
\end{align*}
Thus
\begin{align*}
\sum_{i}dist(c,S(i)) & =\sum_{i\neq j''(c),j^*(c)}dist(c,S(i))+dist(c,S(j''(c)))+dist(c,j^*(c))\\
& \leq (h-2)d'(c)+\opt(c)+d'(c)+2\opt(c)=(h-1)d'(c)+3\opt(c).
\end{align*}
\end{enumerate}
We observe that
\begin{align*}
\sum_{c \text{ in case 3}}\Delta(c) & \leq \sum_{c' \in C^{MM}}\frac{|C(f^*(c))\setminus C'(f^*(c))|}{|C'(f^*(c))|}(\opt(c')+d'(c'))\\
& \leq 
\frac{\delta}{1-\delta}\sum_{c' \in C^{MM}}(\opt(c')+d'(c'))\leq \frac{\delta}{1-\delta}(\opt^{MM}+d'^{MM}).
\end{align*}
Summing the different cases over the different clients one obtains
\begin{align*}
\sum_{i}d(S(i)) & \leq (h-1)d'+\opt+\sum_{c \text{ in case 4}}2\opt(c)+\sum_{c\text{ in case 3}}\Delta(c)\\
& \leq (h-1)d'+\opt+2\opt^L+\frac{\delta}{1-\delta}(\opt^{MM}+d'^{MM})\\
 & = (h-1)d'+\frac{\delta}{1-\delta}(d'^{MM}+\opt^{MM})+3\opt^L+\opt^M
\end{align*}
Recall that with probability $4\eps$ for each client not in case 1 above it might happen that there exists a second index $j''_2(c)$ such that $f'(c)\in A_{j''_2(c)}$.
In that case we might need to add some extra terms in cases 3 and 4 above, and also add a term $\min\{0,\opt(c)-d'(c)\}$. We pessimistically assume that the latter term is $0$, and that the mentioned probability is precisely $4\eps$ for each such client $c$. Summing over all clients one gets in expectation an extra additive term of value at most
$$
4\eps(\frac{\delta}{1-\delta}(d'^{MM}+\opt^{MM})+3\opt^L+\opt^M).
$$ 
This gives the desired inequality \eqref{eqn:desiredBound_S2}.
\end{proof}

\subsection{Case $|S'|\leq k$.}
\label{sec:localSearchOmitted:S1}
We complete the proof of Theorem~\ref{thr:localSearch}.
In fact we are going to prove a slightly stronger theorem that implies Theorem~\ref{thr:localSearch}.
Let $\delta', \delta^*\in (0,1/2)$ be parameters to be fixed later. We say that $f^*\in \OPT$ \emph{captures} $f'\in S'$ if more than a $1-\delta^*$ fraction of the clients served by $f'$ in $S'$ are served by $f^*$ in $\OPT$. We say that $f'\in S'$ captures $f^*\in \OPT$ if more than a $1-\delta'$ fraction of the clients served by $f^*$ in $\OPT$ are served by $f'$ in $S'$. Notice that the definition is \emph{not} symmetric. We moreover use a generalization of the definition of isolated regions
from \cite{Cohen-AddadKM19}: We say that a set of facilities $F' \subseteq S'$ and a set of facilities $F^*\subseteq \OPT$ such that $|F^*| \le 1/\epsilon$
are in \emph{isolation} if at least a $1-\delta'$ fraction of the clients served by $F^*$ in $\OPT$ are served by $F'$ in $S'$ and if at least
a $1-\delta^*$ fraction of the clients served by $F'$ in $S'$ are served by $F^*$ in $\OPT$.
In the following, we consider an arbitrary set of disjoint pairs of facilities $(F', F^*)$ in isolation.

We let $\OPT^M$ be an arbitrary set of facilities of $\OPT$ that satisfy either of the following: 
(1) they are captured by some facility in $S'$, or
(2) they are in at least one chosen group that is in isolation.
We let $\OPT^L = \OPT \setminus \OPT^M$. We let $S^M$ be the set of facilities of $S'$ that are captured by some facility in $\OPT$ or
in a chosen group that is in isolation, 
and set $S^L = S' \setminus S^M$.
We also let $\opt^M=d(\OPT^M)$, $\opt^L=d(\OPT^L)$, $d^M =d(S^M)$ and $d^L=d(S^L)$.
We let $C^{XY}=\{c\in C \mid S'(c)\in S^X \text{ and }\OPT(c)\in \OPT^Y\}$, e.g.: $C^{LM}$ is the set of clients served by a lonely facility in $S'$ and by a matched facility in $\OPT$. We let $C^{\star M} = C^{LM} \cup C^{MM}$, etc..
We let $d^{XY}$ and $\opt^{XY}$ be the connection cost in $S'$ and $\OPT$, resp., of the clients in $C^{XY}$.
We let $d^{\star X} = d^{L X} + d^{M X}$, etc..

We are ready to state our local-search based bound. 
\begin{theorem}\label{thr:localSearch_generalized}
  We have
  \begin{align*}
  & \lambda k' +d' \leq\\
   &  k\lambda  +  3 \opt^{LL} + \opt^{M} +  \frac{\delta'}{1-\delta'} (d^{MM} + \opt^{MM}) + O(\eps (d' + \opt))
 +\min \begin{cases} 3 \opt^{ML} \\
    \opt^{ML} +  \frac{\delta^*}{1-\delta^*} (d^{M} + \opt^{M})
  \end{cases}  
  \end{align*}
\end{theorem}

Assuming the above theorem, we can conclude the proof of Theorem~\ref{thr:localSearch}.
\begin{proof}[Proof of Theorem~\ref{thr:localSearch}]
  In the case where the solution $S'$ produced by the algorithm is such that $|S'| > k$,
  then we invoke Lemma~\ref{lem:localSearchBound_S2} and obtain the desired bounds.
  
  In the case where  $|S'| \le k$, then we have that the matching of facilities defined
  in Section~\ref{sec:localSearch} slightly differs from the more general
  one we have in this section; it is the following:
  Given the parameter $\delta$ of Theorem~\ref{thr:localSearch},
  consider $f^*\in \OPT$, and let $M(f^*)$ be the facilities in $S'$ which are
  $1/2$-captured by $f^*$. If $M(f^*)$ $(1-\delta)$-captures $f^*$, then we say that $f^*$ and $M(f^*)$ are \emph{matched}.
  If $M(f^*)$ $(1-\delta)$-captures $f^*$, then we say that $f^*$ and $M(f^*)$ are \emph{matched}.

  We now translate the above matching in the matching defined for Theorem~\ref{thr:localSearch_generalized}.
  For each such pair above, we create a pair of facilities in isolation $(f^*, M(f^*))$, note that by definition
  these pairs are pairwise disjoint (i.e.: for all pairs $(f^*, M(f^*))$, $(f_1^*, M(f_1^*))$, we have
  $f_1^* \neq f^*$ and $M(f_1^*) \cap M(f^*) = \emptyset$).
  We can thus apply Theorem~\ref{thr:localSearch_generalized}, with $\delta' := \delta$ and $\delta^* := 1/2$,
  and define $\OPT^M$ to be the set containing the chosen groups of facilities in isolation described above.
  It follows that the set of facilities of
  $\OPT$ and the set of facilities of $S'$ that are matched according to the definition of Theorem~\ref{thr:localSearch}
  and matched according to the definition of Theorem~\ref{thr:localSearch_generalized} are the same.
  We can thus write:
  \begin{align*}
    &\lambda k' +d' \leq  k\lambda  +  3 \opt^{L} + 
    \opt^{M}
    +  \frac{\delta'}{1-\delta'} (d^{MM}+ \opt^{MM})
    + O(\eps (d' + \opt)),
  \end{align*}
  as claimed.
\end{proof}

We complete the proof by showing Theorem~\ref{thr:localSearch_generalized}.
\begin{proof}[Proof of Theorem~\ref{thr:localSearch_generalized}]
We form swap sets $(A_i,B_i)$, $A_i\subseteq S'$ and $B_i\subseteq \OPT$, using the following directed graph $G$.
  For each chosen group $(F_1,F^*)$ of facilities in isolation, create a vertex $v_f$ for each $f \in F_1$ and a vertex $v_{F^*}$ for $F^*$
  and add edges $\langle v_f, v_{F^*} \rangle$ then repeat on the remaining facilities. 
  For each remaining facility $f$ of $\OPT$ or $S'$ create a vertex $v_f$.
  Then, for each $f' \in \OPT^M$ and $f\in S'$, create an edge from $v_f$ to $v_{f'}$ if $f'$ captures $f$.
  Namely, if $f'$ is in $\OPT^M$ and
  $f$ is in $S'$ then at least a $1-\delta^*$ fraction of the clients served by $f$ in $S'$ are served
  by $f'$ in $\OPT$.
  Also create an edge $v_{f'}$ to $v_{f}$  if $f'$ is in $S'$ and $f \in \OPT$
  and at least a $1-\delta'$ fraction of the clients served by $f$ in $\OPT$ are served by $f'$ in $S'$.
  Moreover, for each facility $f$
  of $\OPT-\OPT^M$,
  create an edge from $f$ to the vertex corresponding to the closest facility of $S'$.

  Note that the out-degree of each vertex of the graph is at most 1 and the out-degree of any vertex $v_{F^*}$ is 0.
  For each vertex $v_f$ with outgoing edge $\langle v_f, v_{f'} \rangle$, we define $\phi(f) = f'$. If $v_f$ has
  out-degree 0 we let $\phi(f) = \emptyset$.

  We let $\tF^*$ be the set of facilities $f$ of $\OPT$ that satisfy (1) $\phi(f)$ has indegree at least $2$; or (2)
  $f$ is part of a chosen group of facilities in isolation $(F_1,F^*)$ such that $|F^*| > |F_1|$.

  We now choose a collection of subsets $\{(F_i, F^*_i)\}$, where $F_i \subseteq S'$ and $F^*_i \subseteq \OPT$ and modify $\OPT$ to obtain a solution
  $\wO = \OPT \bigcup_i F_i  \setminus \bigcup_i F^*_i$  that satisfies the following:
  \begin{itemize}
  \item Each client $c$ served by a facility in $F^*_i$ in $\OPT$
  it is served by a facility in $F_i$ in $\wO$. 
  \item The service cost of $\wO$ is at most the service cost of $\OPT$ plus $O(\eps (d' + \opt))$.
  \item $\wO$ opens at most $|\OPT| - 2\eps^2 |\tF^*|$ centers.
  \end{itemize}

  To define a solution $\wO$ that satisfies the above, we proceed as follows. Consider each facility $f \in S'$ with
  indegree at least $2$, and let $\OPT(f)$ be the set of facilities of $\OPT$ with an outgoing edge toward $f$.
  Consider the group $f, \OPT(f)$. We bound the change in cost induced by the following \emph{group exchange}:
  add $f$ and remove $\OPT(f)$ from $\OPT$ and serve
  all the clients served by a facility in $\OPT(f)$ by $f$. Since for each $f^*$ we have that either $f$ captures $f^*$ or
  $f$ is the closest facility to $f^*$, we have that the change in cost is at most 
  \[\frac{1}{1-\delta'}\sum_{c : \text{ served by $f^*$ in } \OPT}  \dist(c, S') + \dist(c, \OPT).\]
  Consider also any chosen group of facilities in isolations $F_1, F^*$ and define a group exchange which consists
  in replacing $F^*$ with $F_1$ in $\OPT$. Replacing $F^*$ with $F_1$ and serving the clients
  served by $F^*$ in $\OPT$ with $F_1$ also increases the cost for these clients by at most the same amount as above.
  
  Therefore, consider all the chosen groups of facilities in isolation $(F_1,F^*)$, and all pairs $(f, F^*)$ where $f \in S'$
  and $F^* \subseteq \OPT$, $|F^*| > 1$, and where each facilities $f^* \in F^*$ has an edge toward $f$. Let $\tilde{k}$ be the number of groups.
  Pick a random subset of $2\eps \tilde{k}$ groups and perform the group exchange for each picked group.
  The expected cost increase is thus at most $O(\eps (d' + \opt))$.
  Moreover, each group exchange reduces the number of centers in the solution by at least 1 and since for each group $|F^*| \le 1/\eps$, we have
  that the number of groups is at least $\eps |\tF^*|$.  Therefore, the resulting
  solution contains at most $|\OPT| - 2\eps^2 |\tF^*|$ centers. Denote by $I = S' \cap \wO$.    Observe that the service cost for the
  clients served by any $f^* \in S' \cap \wO$ in $\wO$ is as good in $S'$ since
  $f^* \in S'$. 
  Importantly, we keep the chosen group of facilities in isolation unchanged in spirit: if a group of a facilities in isolation
  $(F^*, F)$ has been chosen above and so $F \subseteq \wO$, we replace $(F^*, F)$ by $(F, F)$.

  We now describe how to adapt the graph $G$.
  Each vertex representing a facility or a group of facilities that is both in $S'$ and $\wO$ is removed from the $G$ and
  so the resulting graph only contains facilities of $S'-I$ and $\wO - I$.
  Moreover, for each facility $f$ and group of facilities $F^*$ of $\OPT$ with outgoing edges towards $f$ for which the exchange
  was performed, we remove $f$ and $F^*$ from the graph (and all the adjacent edges).

  
  We now want to build a set of \emph{swaps} $\{(A_i, B_i)\}$, namely a collection of pairs $A_i \subset S' \setminus I$, $B_i \subset \wO - I$, where
  $|A_i|+|B_i| \le \Delta$. From there we can analyze the cost of the solutions
  $S' \setminus A_i \cup B_i$ and obtain a bound on the cost of $S'$.
  We will give a randomized procedure that produces a set of feasible swaps that satisfies the following property.\\
  \textbf{Swap consistency:}
  For any feasible swap $(A_i,B_i)$, we have:
  \begin{enumerate}
  \item\label{swap:1} Property \ref{swap:1}: For any facility $f \in S'$, if $f \in A_i$, 
    we have $\phi(f) \in B_i$ with probability 1 if $f$ is in a chosen group of facilities in isolation, and at least $(1-\eps)$ otherwise.
    In other words, if the swap
    \emph{closes} $f$ then it must open $\phi(f)$ with probability at least $(1-\eps)$ if $f$ is not in an isolated group
    and with probability 1 otherwise;
  \item\label{swap:2} Property \ref{swap:2}: For any facility $f^* \in \wO \setminus I$, if $f^* \notin B_i$, then $\phi(f^*) \notin A_i$. In other words,
    if the swap does not \emph{open} $f^*$ then it must leave $\phi(f^*)$ open;
  \item\label{swap:3} Property \ref{swap:3}: Each $f^* \in \wO \setminus I$ belongs to at least one group and $\sum_i |B_i| \le k + \eps^2|\tF^*|$;
  \item\label{swap:4} Property \ref{swap:4}: Each facility of $S'$ belongs to at most one group. The number of facilities of $S' \setminus I$
    that do not belong to any group is at most $\eps^2 |\tF^*|$. The facilities of $I$ do not belong to any group.
  \end{enumerate}
  \begin{claim}
    \label{claim:swapstructure}
    There exists a set of feasible swaps that satisfies the swap consistency properties.
  \end{claim}
  Let's first assume that the claim holds and proceed to the analysis. Let $h$ be the number of pairs $(A_i,B_i)$ in the
  set of swaps prescribed by Claim~\ref{claim:swapstructure}.

  For each swap set $(A_i, B_i)$ we consider the local update where the algorithm opens $B_i$ and removes $A_i$ from $S'$.
  Note this is a feasible local update since we assume $|B_i| + |A_i| \le \Delta$.
  Therefore, by the local optimality of $S'$, it must be the case that, for $S(i)=S'\cup B_i-A_i$, one has
  \begin{equation}\label{eqn:localSearchCondition}
    \lambda k'+d' \leq  \lambda  (k'+|B_i|-|A_i|)+d(S(i)).
  \end{equation}

  Summing over $i$, from \eqref{eqn:localSearchCondition}  one gets
  \begin{equation}\label{eqn:initialGlobalBound}
     \lambda h k'+hd'\leq \sum_{i} \lambda  (k'+|B_i|-|A_i|)+\sum_i d(S(i))=h  \lambda  k'+ \lambda (k-k')+\sum_i d(S(i)),
  \end{equation}
  where we have used Property~\ref{swap:3} of the swap consistency, the definition of $\wO$, and the fact that the swaps are
  defined on $S' \setminus I$ and $\wO \setminus I$.
  We next focus on a specific client $c$. Let $f^*(c)$ and $f_1(c)$ be the facilities serving $c$ in $\wO$ and $S'$, respectively.
  We use the shortcuts $d' (c)=\dist(c,S')$ and $\wo(c)=\dist(c,\wO)$. The contribution of $c$ to $hd'$ is precisely $h\, dist(c,S')=h\, d'(c)$.
  Let us bound the contribution $\sum_{i}\dist(c,S(i))$ of $c$ to $\sum_i d(S(i))$.
  
  First consider a client $c$ served by a facility in $I$ in $\wO$. Then, for the swap that closes the facility serving it in $S'$ (if there is such,
  and there can be at most one by Property~\ref{swap:3}),
  we can always reassign $c$ to the facility serving it in $I$. In this case, the cost is $\wo(c)$. In any other swap it is at most $d'(c)$.
  Hence, the contribution is always at most $(h-1)d'(c) + \wo(c)$.
  Additionally, for any client $c$ such that there is no swap containing the facility serving $c$ in $S'$, then since there must be a swap containing
  $f^*$ (by Property~\ref{swap:3}), we have that the cost contribution of $c$ is at most $(h-1)d'(c) + \wo(c)$. We will
  see in the remaining that the bound provided dominates this one and we henceforth focus on the case where there is
  at least one swap containing the facility serving $c$ in $S'$.

  We now turn to the rest of clients; We have 3 types of clients among the rest:
  \begin{itemize}
  \item $c \in C^{\star M}$: an \emph{$\OPT$-matched} client.
  \item $c \in C^{ML}$: an \emph{$S'$-matched} client.
  \item $c \in C^{LL}$: a \emph{lonely} client.
  \end{itemize}

    

  \paragraph{$\OPT$-matched clients.}
  Consider a client $c$ served in $\OPT$ by facility $f^*$ that is captured or in a chosen isolation group. Let $f$ be the facility serving $c$
  in $S'$.
  Consider the swap $(A_i,B_i)$ where $f \in A_i$. 
  We bound the cost paid by $c$ in $S(i)$ by one of the following
  \begin{enumerate}
  \item $\wo(c)$ if $f^* \in B_i$,
  \item Otherwise, the distance from $c$ to the facility $f'$ of $S'$ that
    is captured by $f^*$ that is the closest to $f^*$ if $f^*$ is in a chosen isolation group, or
  \item $\phi(f^*)$ if $f^*$ is captured.
  \end{enumerate}
  We argue that these are valid reassignments. (1) is obviously valid since it assumes that $f^* \in B_i$.
  For (2), observe that by Property~\ref{swap:1}, if $f^* \notin B_i$ then,
  $f' \notin A_i$ (since $f'$ is captured by $f^*$ and $f^*$ is in isolation, $f'$ is in an isolated group with $f^*$
  and so this holds with probability 1).
  For (3), by Property~\ref{swap:2}, if $f^* \notin B_i$ then $\phi(f^*) \notin A_i$.
  It follows that $\dist(c, f')$ is a valid upper bound on the cost of $c$ in $S(i)$ if $f^*$ is in a chosen
  isolation group and
  $\dist(c, \phi(f^*))$ is a valid upper bound on the cost of $c$ if $f^*$ is captured.
  
  In the case of (2), let us now provide an upper bound on $\dist(c, f')$, and therefore on the cost of $c$ in $S(i)$.
  We write $\dist(c, f') \le \dist(c, f^*) + \dist(f^*, f')$. Then, denote by $U(f^*)$ the set of clients that
  are served by $f^*$ in $\OPT$ and by a facility captured by $f^*$ in $S'$.
  Thus, we have that
  $\dist(f^*, f') \le \frac{1}{|U(f^*)|} \sum_{c' \in U(f^*)} \dist(c', f^*) + \dist(c', S')$.
  Hence
  \[\dist(c, f') \le \dist(c, f^*) +\frac{1}{|U(f^*)|} \sum_{c' \in U(f^*)} \dist(c', f^*) + \dist(c', S')\]
  Finally, for each such $f^*$, the number of clients served by $f^*$ in $\OPT$ that are not served by a facility of $S'$
  captured by $f^*$ is at most $\frac{\delta'}{1-\delta'} |U(f^*)|$.

  In the case of (3), let's then provide an upper bound on the $\dist(c, \phi(f^*)) \le \dist(c, f^*) + \dist(f^*, \phi(f^*))$.
  Denote by $U(f^*)$ the set of clients that
  are served by $f^*$ in $\OPT$ and by a $\phi(f^*)$ in $S'$. We have that 
  $\dist(f^*, \phi(f^*)) \le \frac{1}{|U(f^*)|} \sum_{c' \in U(f^*)} \dist(c', f^*) + \dist(c', S')$.
  Here again, for each such $f^*$, the number of clients served by $f^*$ in $\OPT$ that are not served by $\phi(f^*)$ in $S'$
  is at most $\frac{\delta'}{1-\delta'} |U(f^*)|$.

  Hence, the total service cost contribution of the $\OPT$-matched clients to $\sum_i d(S(i))$
  is at most
  \[(h-1)d'^{\star M} +  \opt^{\star M} + \frac{\delta'}{1-\delta'}\left(\opt^{MM} +  d'^{MM} \right) + O(\eps (d' + \opt)).\]
  
  \paragraph{$S'$-matched clients.} 
  Consider a client $c$ served in $S'$ by a facility $f$ that is captured by some facility of $\wO$ and
  served in $\wO$ by a facility $f^*$ not in isolation.
  Let $(A_i, B_i)$ be the swap such that
  $f^* \in B_i$, and $(A_j, B_j)$ be the swap such that $f \in A_j$.
  We have two options: (1) $i = j$ or (2) $i \neq j$.
  In the first case, we have that the cost for the swap $(A_i,B_i)$ is at most $\wo(c)$, and $d'(c)$ for any other swap.
  In the second case, we bound the cost as follows. For the swap $(A_i, B_i)$ the cost is at most $\wo(c)$.
  For any swap $(A_\ell, B_\ell)$, $\ell \notin \{i,j\}$, we have that the service cost for $c$ is at most $d'(c)$.
  Now, for the swap $(A_j, B_j)$ we have two options which we will bound by the minimum of two quantities since the
  local search equation holds for both cases.
  \begin{enumerate}
  \item Case 1: We bound the cost of $c$ by the distance from $c$ to the facility of $S'$ that is the closest to $f^*$.
    Since $f^*$ is not in isolation, it has an edge to the closest facility $f'$ in $S'$. Since $f^* \notin B_j$ it must be
    that $f' \notin A_j$. Hence this is a valid reassignment. In this case, the cost is by triangle inequality
    at most $\dist(c, f^*) + \dist(f^*, f') \le \dist(c, f^*)  + \dist(c, f^*) + \dist(c, f)$, since $\dist(f^*, f') \le \dist(f^*, f)$,
    and so at most $2\wo(c) + d'(c)$.
    Summing over all such clients, we have a total cost contribution of at most
    $(h-1)d'^{ML} + 3 \opt^{ML} +O(\eps d')$. Note that this bound holds with probability 1.
  \item Case 2: Recall that $f$ is captured (possibly in isolation).
    Let $F^*$ denote the set of facilities capturing $f$ (it is the set
    of facilities of $\wO$ if $f$ is in isolation or a single facility otherwise).
    We bound the cost of $c$ by the distance from $c$ to the facility $\tilde{f} \in F^*$ that is the closest to $f$.
    By Property~\ref{swap:1},
    $f \in A_j$ implies that $F^* \subseteq B_j$ with probability $1-\eps$ and so this is a valid upper bound on the cost
    only with this probability
    and we use Case 1 to provide an upper bound otherwise.
    Thus, this cost is at most $\dist(c, f) + \dist(f, F^*) + \eps 3 \wo(c)$. Let $U(f)$ be the set of clients served
    by $f$ in $S'$ and $F^*$
    in $\OPT$. Again, by triangle inequality we have that
    $\dist(f, F^*) \le \frac{1}{|U(f)|} \sum_{c' \in U(f)} \dist(c', f) + \dist(c', F^*)$.
    It follows that the cost for $c$ is at most 
    $\dist(c, f) + \frac{1}{|U(f)|} \sum_{c' \in U(f)} \dist(c', f) + \dist(c', F^*)$.
    Now since $f$ is captured by $F^*$, we have that the number of clients served by $f$ and not in $U(f)$ is at most
    $\frac{\delta^*}{1-\delta^*}|U(f)|$.
    Therefore, summing over all such $f$, the service cost contribution for the $S'$-isolated clients served by $f$ is
    at most $(h-1)d'^{ML} + \opt^{ML}+ \frac{\delta^*}{1-\delta^*} (d'^{M \star}  + \opt^{M\star} ) +O(\eps d').$
  \end{enumerate}

  \paragraph{Lonely clients.}
  Consider a client $c$ served in $S'$ by facility $f$ that is not captured and served in $\OPT$ by a facility $f^*$ that is not captured.
  Let $(A_i, B_i)$ be the swap such that
  $f^* \in B_i$.  For the swap $(A_i, B_i)$ the cost is at most $\wo(c)$ since $f^* \in B_i$.
  Now, for any swap $(A_j, B_j)$ where $f \in A_j$.
  We bound the cost of $c$ by the distance from $c$ to the facility of $S'$ that is the closest to $f^*$.
  Since $f^*$ is not in isolation, it has an edge to the closest facility $f'$ in $S'$. By Property~\ref{swap:2}, $f^* \notin B_j$ implies
  that $f' \notin A_j$. Hence this is a valid reassignment. In this case, the cost is by triangle inequality
  at most $\dist(c, f^*) + \dist(f^*, f') \le \dist(c, f^*)  + \dist(c, f^*) + \dist(c, f)$, since $\dist(f^*, f') \le \dist(f^*, f)$,
  and so at most $2\wo(c) + d'(c)$.
  Thus, summing over all such clients, we have a total cost contribution of at most
  $(h-1)d'^{LL} + 3(1+\eps) \opt^{LL} + \eps d'$.
  Combining the three above bounds yield the bound of the theorem.

  We now turn to the proof of Claim~\ref{claim:swapstructure}.
  \begin{proof}[Proof of Claim~\ref{claim:swapstructure}]
    We will work with the graph $G = (A,B, E)$ defined above, where $A,B$ are the vertices representing the facilities
    of $S'-I$ and $\wO-I$ respectively.
    We first preprocess the graph and then partition its vertices (hence partitioning the associated
    facilities). For each part $P_i$ of the partition
    we will create a swap $(A_i, B_i)$ by setting $A_i$ (resp. $B_i$) to be the set of facilities corresponding to vertices in
    $P_i \cap A$ (resp. $P_i \cap B$). Before defining the partition of $G$, we make the following modifications to the graph.
    
    \paragraph{Degree reduction.}
    For each vertex $v \in B$ with indegree $d > 1/\eps^2$, called a \emph{heavy vertex}
    create $\lfloor d \eps^2 \rfloor$ new vertices $v_1,\ldots, v_{\lfloor d\eps^2 \rfloor}$ that will go in the $B$ side of the graph
    and remove $v$. We say that
    the new vertices also represent the facilities of $\wO$ that $v$ was representing (the facilites are now represented
    $\lfloor d \eps^2 \rfloor$ times in the graph).
    Let us now bound the number of vertices in $B$ at the end of this procedure. Let $B_{>\eps^{-2}}$ be the
    set of vertices of $B$ with indegree greater than $1/\eps^2$. For each vertex $v \in B_{>\eps^{-2}}$ with indegree $d$ the contribution
    to the number of new vertices is $\lfloor d \eps^2 \rfloor - 1$.
    Recall that $|A| \le |B| + \eps^2|\tF^*|$ and $\tF^*$ is the set of facilities $f$ of $\wO$ that satisfy
    (1) $\phi(f)$ has indegree at least $2$; or (2)
    $f$ is part of a chosen group of facilities in isolation $(F_1,F^*)$ such that $|F^*| > |F_1|$.
    Thus the number of added vertices is at most $(1+\eps)\eps^2|\tF^*|$.
    Moreover, any vertex $v_f$ where $f \in S'$ that has indegree larger than $1/\eps^2$ is placed into a new set $\bar{A}$, we immediately
    infer that $|\bar{A}| \le \eps^2|\tF^*|$. 
    
    We now define edges adjacent to the new vertices. For each new vertex, the outgoing edge of
    $v$ is present in each new vertex as well. The incoming edges are arbitrarily partitioned into $\lfloor d\eps^2 \rfloor$
    groups of size at most $2/\eps^2$ and each new vertex receives a different group of incoming edges.
    The graph consisting of $G$ plus the new vertices and the new edges, and resulting from deleting vertices of $\bar{A}$,
    and deleting the heavy vertices has thus degree at most $1/\eps^2$.
  
    \paragraph{Defining $U$ and $\calP$.}
    Consider the graph where the above degree reduction procedure has been applied and the vertices of
    $\bar{A}$ have been deleted. We will now provide a partitioning of the above graph into
    parts $(A_i,B_i)$ such that $|A_i| + |B_i| \le (1/\eps+1)^{O(1/\eps^7)}$.
    
    Then, contract all edges from $f^* \in \wO$ to $f \in S'$ and replace the vertex by one vertex
    representing $f$ and all the facilities of $\wO$ now contracted to $f$.
    Moreover, for each pair of isolated facilities $F_1,F^*$, contract the edges $\langle v_f, v_{F^*} \rangle$.
    Since the outdegree of $v_{F^*}$ is by definition 0,
    the degree of this graph is thus at most $1/\eps^3$. Note that the only remaining edges in this graph
    are of the form $\langle f, f^* \rangle$, where $f \in S', f^* \in \wO$. Let $G'$ be the resulting graph.
    
    Therefore, consider the following randomized procedure on $G'$. For each connected component of $G'$,
    start at an arbitrary vertex of $A \setminus \bar{A}$ and performs a BFS on the undirected version of the graph,
    labeling the vertices of $A \setminus \bar{A}$ at hop-distance $i$ with label $i \mod 1/\eps^4+1$. Then pick a random
    integer $i^*$ in $[0, 1/\eps^4]$ and remove all the outgoing edges of the vertices with label equal to $i^*$
    -- the probability of any edge of being removed is thus at most $\eps^4$.
    
    Since the degree of each vertex is at most $1/\eps^3$, the connected components of the graph have thus size
    at most $(1/\eps+1)^{3/\eps^4}$. We now create a part in $\calP$ for each of them.
    Since each vertex may represent up to $1/\eps^3$ vertices, we have that the total number of
    vertices represented in a given connected component is at most $(1/\eps+1)^{3/\eps^7}$. For each connected component $(A_i, B_i, E_i)$ 
    we create a pair $(U_i, V_i)$ where $U_i$ is the set of facilities represented by vertices in $A_i$ and
    $V_i$ is the set of facilities represented by vertices in $B_i$.
    We now show that this set of swaps satisfies the swap consistency properties.
    
    Let us start with Property~\ref{swap:3}. Observe that each vertex of $G$ representing a facility of $\wO$ and
    of degree at most $1/\eps^2$ is also in $G'$ and  each vertex of $G$ representing a facility of $\wO$ and of degree $> 1/\eps^2$
    appears with multiplicity $\lfloor d \eps^2 \rfloor$. So since each facility of $\wO$ is initially represented in $G$ it is also
    represented in $G'$, and since we take a partition of the vertices of $G'$ it must appear in at least one swap created.
    Moreover, since the number of new vertices in $G'$ is at most $\eps^2 |\tF^*|$,
    we have that $\sum_i |B_i| \le k + \eps^2|\tF^*|$, as desired.
    
    Property~\ref{swap:4} follow immediately from the fact that we have a partition of $G'$ which contains
    vertices that represent all the facilities of $S'$ with multiplicity 1 (note that we did not create copies of vertices of $A$)
    except for the facilities of $S'$ whose corresponding vertices is in $\bar{A}$, in which case the facility is not present
    in any swap. Recall moreover that $|\bar{A}| \le  \eps^2|\tF^*|$.
    
    We now turn to Property~\ref{swap:1}. First, any edge of $G$ that has been contracted is satisfied (since all the facilities
    are represented by one single vertex in $G'$ and so necessarily in the same swap). Then, as discussed above, each edge of
    $G'$ is cut with probability at most $\eps^4$. If an edge $\langle f, f^* \rangle$
    is not cut then $f,f^*$ are in the same component and so for any swap $(A_i,B_i)$, if $f \in A_i$ then $f^* = \phi(f) \in B_i$ as desired.
    Moreover, since $G$ does not contain the facilities of $I$ anymore, we have that no facility in $I$ is in any swap.

    We finally turn to Property~\ref{swap:2}. Observe that all edges from $f^*$ to $f$ have been contracted, or the vertex corresponding to
    $f$ has been deleted. In the first case, we have that $f^*$ and $f$ are represented by the same vertex in $G'$ and so are necessarily
    in the same swap. In the latter case, $f$ does not belong to any swap.
    Hence, for any swap $(A_i,B_i)$, if $f^* \notin B_i$ then $f \notin A_i$.
      \end{proof}

\end{proof}

\section{Omitted Proofs from Section \ref{sec:improvedkMedian}}
\label{sec:omittedSection4}

\begin{proof}[Proof of Lemma \ref{lem:boundLonelyFacilityCost}]
Consider any $f'\in S_2^L$. Let $\Delta(f')$ be the increase of the connection cost of $S_2$ due to the removal of $f'$. The local optimality of $S_2$ implies $\lambda\leq \Delta(f')$, hence
$$
\lambda k^L_2\leq \sum_{f'\in S^L_2}\Delta(f').
$$
Thus it is sufficient to upper bound the righthand side of the above inequality. In order to do that, we consider each client $c$ in the set $C_{S_2}(f')$ of clients served by $f'$ in $S^L_2$, and describe a random path $P(c)$ from $c$ to some facility in $S_2$ distinct from $f'$. Let $\ell(P)$ be the length of a path $P$. The sum of values $\ell(P(c))-d_2(c)$ over clients $c\in C_{S_2}(f')$ gives a valid upper bound on $\Delta(f')$, and we will sum these values over all clients served by any $f'\in S^L_2$. The path $P(c)$ is defined as follows. Let $f^*(c)=\OPT(c)$ be the facility serving $c$ in $\OPT$. We distinguish two cases: (a) $f^*(c)$ is not $(1-\delta)$-captured by $f'$, and (b) the complementary case. In case (a), let $rand(c)$ be a random client served by $f^*(c)$ which is not served by $f'$ in $S_2$. The random path $P(c)$ is given by the concatenation of the shortest paths from $c$ to $f^*(c)$, from $f^*(c)$ to $rand(c)$, and from $rand(c)$ to the facility $f'(rand(c))$ serving $rand(c)$ in $S_2$. Notice that
$$
\ell(P(c))=\opt(c)+\opt(rand(c))+d_2(rand(c)).
$$
Case $(b)$ is slightly more complex. Let $C'(f')\subseteq C_{S_2}(f')$ be the clients which in $\OPT$ are served by a facility not captured by $f'$. We remark that $|C'(f')|\geq |C_{S_2}(f')|/2$ since otherwise $f'$ would be matched. Let $next(c)$ be a random client in $C'(f')$. The path $P(c)$ is the concatenation of the shortest path from $c$ to $f'$, the shortest path from $f'$ to $next(c)$, and the random path $P(next(c))$ already defined earlier. In this case we have
\begin{align*}
\ell(P(c)) & =d_2(c)+d_2(next(c))+\ell(P(next(c)))\\
& =d_2(c)+d_2(next(c))+\opt(next(c))+\opt(rand(next(c)))+d_2(rand(next(c)).
\end{align*}
Let $next^{-1}(c)$ be the clients $c'$ such that $c=next(c')$. We have that, for $c\in C'(f')$, $Ex[|next^{-1}(c)|]\leq 1$ (and $next^-1(c)=\emptyset$ otherwise). Let $\ell'(P(c))$ be the part of the lengths $\ell(P(c))$ not involving $rand(\cdot)$. Let also $\ell''(P(c))$ be  the remaining part of each such length. One has
\begin{align*}
 Ex[\sum_{f'\in S^L_2}\sum_{c\in C_{S_2}(f')}(\ell'(P(c))-d_2(c)]
= & \sum_{f'\in S^L_2}\sum_{c\in C'(f')}(\opt(c)-d_2(c)+Ex[|next^{-1}(c)|](d_2(c)+\opt(c)))\\
\leq & \sum_{f'\in S^L_2}\sum_{c\in C'(f')}2\opt(c) \leq 2(\opt^{LM}+\opt^{LL}).
\end{align*}
In order to upper bound the sum of the cost $\ell''(P(c))$, let us define $rand^{-1}(c)$ as the clients $c'$ in some set $C'(f')$ such that $c=rand(c')$. We wish to upper bound $Ex[|rand^{-1}(c)|]$. Consider the facility $f^*$ serving $c$ in $\OPT$, and let us partition the clients $C_{\OPT}(f^*)$ served by $f^*$ into subsets $C^1,\ldots,C^h$ depending on the facility $f^1,\ldots,f^h$ serving them in $S_2$. Assume w.l.o.g. that $c\in C^1$. For each $c'\in C^j$, $j>1$, $c'$ 
belongs to $rand^{-1}(c)$ with probability $\frac{1}{|C_{\OPT}(f^*)|-|C^j|}$. Thus
$$
Ex[|rand^{-1}(c)|]=\sum_{i\geq 2}\frac{|C^i|}{|C_{\OPT}(f^*)|-|C^i|}.
$$
Observe that for $x,y\geq 0$ and $x+y\leq z>0$, $\frac{x}{z-x}+\frac{y}{z-y}\leq \frac{x+y}{z-x-y}$. By repeatedly applying this inequality, one obtains that the largest value of the above sum is achieved when $h=2$, being in particular $|C^2|/|C^1|$. Notice however that $f^2$ does not $(1-\delta)$-capture $f^*$ by construction. Indeed otherwise clients $c'\in C^2$ would not consider selecting $rand(c)$ in $C^1$. This implies $|C^2|\leq (1-\delta)|C(f^*)|$,  hence $|C^2|/|C^1|\leq \frac{1-\delta}{\delta}$. Thus
$$
Ex[|rand^{-1}(c)|]\leq \frac{1-\delta}{\delta}
$$ 
Now let us define $\widetilde{rand}^{-1}(c)$ as the set of clients in $rand^{-1}(c)$ plus any client $c'$ in some set $C_{S_2}(f')\setminus C'(f')$ such that $next(c')\in rand^{-1}(c)$. By the above discussion and the independence of the sampling processes defining $next()$ and $rand()$, one has that
$$
Ex[|\widetilde{rand}^{-1}(c)|]\leq 2Ex[|rand^{-1}(c)|]\leq 2\frac{1-\delta}{\delta}
$$
We can conclude that
$$
Ex[\sum_{f'\in S^L_2}\sum_{c\in C_{S_2}(f')}(\ell''(P(c))]\leq \sum_{c}Ex[|\widetilde{rand}^{-1}(c)|](d_2(c)+\opt(c))\leq  2\frac{1-\delta}{\delta}(d_2+\opt).
$$
We can refine the above analysis for the clients $C^{LM}\cup C^{MM}$ as follows. Suppose that $f^*\in \OPT^M$, which implies that $c\in C^{LM}\cup C^{MM}$. In this case at least a $1-\delta$ fraction of the clients $C_{\OPT}(f^*)$ are served by matched facilities $M(f^*)\subseteq S^{M}_{2}$. Those clients cannot belong to $rand^{-1}(c)$. We can adapt the above analysis by letting the sets $C^i$ not include clients in $C^{MM}$ (while possibly $c\in C^{MM}$). This gives   
$$
Ex[|rand^{-1}(c)|]\leq \sum_{i\geq 2}\frac{|C^i|}{|C_{\OPT}(f^*)|-|C^i|}\leq \frac{1}{(1-\delta)|C_{\OPT}(f^*)|}\sum_{i\geq 2}|C^i|\leq \frac{\delta}{1-\delta},  
$$ 
hence 
$$
Ex[|\widetilde{rand}^{-1}(c)|]\leq 2Ex[|rand^{-1}(c)|]\leq 2\frac{\delta}{1-\delta}.
$$
This gives
\begin{align*}
& Ex[\sum_{f'\in S^L_2}\sum_{c\in C_{S_2}(f')}(\ell''(P(c))]
\leq \sum_{c}Ex[|\widetilde{rand}^{-1}(c)|](d_2(c)+\opt(c))\\
\leq &  2\frac{1-\delta}{\delta}(d_{2}^{LL}+\opt^{LL}+d_{2}^{ML}+\opt^{ML})+2\frac{\delta}{1-\delta}(d^{LM}_{2}+\opt^{LM}+d_{2}^{MM}+\opt^{MM}).
\end{align*}
Altogether
\begin{align*}
& Ex[\sum_{f'\in S^L_2}\sum_{c\in C_{S_2}(f')}(\ell(P(c))-d_2(c)]\\
\leq & 2(\opt^{LM}+\opt^{LL})+2\frac{1-\delta}{\delta}(d_{2}^{LL}+\opt^{LL}+d_{2}^{ML}+\opt^{ML})+2\frac{\delta}{1-\delta}(d_{2}^{LM}+\opt^{LM}+d_{2}^{MM}+\opt^{MM})\\
= & \frac{2}{\delta}\opt^{LL}+2\frac{1-\delta}{\delta}(d_{2}^{LL}+d_{2}^{ML}+\opt^{ML})+\frac{2}{1-\delta}\opt^{LM}+2\frac{\delta}{1-\delta}(d_{2}^{MM}+\opt^{MM})\\
\leq & \frac{2}{\delta}(\opt-\opt^{MM})+2\frac{1-\delta}{\delta}(d_{2}-d_{2}^{MM})+2\frac{\delta}{1-\delta}(d_{2}^{MM}+\opt^{MM}).
\end{align*}
\end{proof}

\section{An Improved LMP Approximation for Facility Location:\\ Uniform Opening Costs}
\label{sec:improvedFLuniform}

In this section we describe how to derive a better than $2$ LMP approximation for UFL in the case of uniform facility costs (without giving an explicit bound on the approximation factor). Suppose we are given a bipoint solution $S_B=aS_1+(1-a)S_2$ for $k$-Median as described in Section \ref{sec:improvedkMedian}, for some given value of $\lambda$. Recall that $d_i=d(S_i)$, $k_i=|S_i|$, $k_1\leq k<k_2$ and $k=ak_1+(1-a)k_2$. Lemma \ref{lem:S2bound} shows that $\lambda k_2+d_2\leq \lambda k+(2-\eta_2)\opt$ for some fixed constant $\eta_2>0$. In Section \ref{sec:boundsS1} we will prove the following (weaker) version of Lemma \ref{lem:S2bound} which applies to $S_1$.
\begin{lemma}\label{lem:improvedS1}
For some function $\eta_1(a)$ which is strictly positive for $a>0$, one has $\lambda k_1+d_1\leq \lambda k+(2-\eta_1(a))\opt$.
\end{lemma}
 As a corollary we immediately obtain that $S_B$ is strictly better than $2$ approximate.
\begin{corollary}\label{cor:betterBipoint}
For some absolute constant $\eta>0$, one has $ad_1+bd_2 \leq (2-\eta)\opt$. 
\end{corollary}  
\begin{proof}
From Lemmas \ref{lem:S2bound} and \ref{lem:improvedS1}, a valid choice for $\eta$ is
$$
\eta:=\min_{a\in [0,1]}\{a\eta_1(a)+(1-a)\eta_2\}.
$$ 
It is clear that $\eta>0$. Indeed, for $a\leq 1/2$ this is at least $\eta_2/2>0$ while for $a\geq 1/2$ it is at least $\frac{1}{2}\min_{a\in [1/2,1]}\{\eta_1(a)\}>0$.
\end{proof}

We can use Corollary \ref{cor:betterBipoint} to achieve an LMP $(2-\eta)$-approximation for UFL with uniform facility costs $\lambda$ as follows. We consider any possible number $k$ of facilities in a feasible solution. For each such $k$, we use the construction from Section \ref{sec:improvedkMedian} to obtain a bipoint solution $a(k)S_1(k)+(1-a(k))S_2(k)$ for $k$-Median, for some value $\lambda(k)$ of the Lagrangian multiplier. We return the cheapest solution $S_i(k)$. 
\begin{theorem}\label{thr:improvedFLuniform}
There is a deterministic LMP $(2-\eta)$-approximation for facility location with uniform facility costs, where $\eta>0$ is the constant from Corollary \ref{cor:betterBipoint}.  
\end{theorem}
\begin{proof}
Consider the above algorithm, and let $\OPT$ be any feasible solution with total cost $\lambda k+\opt$, where $\opt=d(\OPT)$ and $k=|\OPT|$. Consider the random solution $S'$ which is equal to $S_1(k)$ with probability $a(k)$ and $S_2(k)$ otherwise. Clearly the cost $\lambda |S'|+d(S')$ of $S'$ (notice that here we consider the original facility cost $\lambda$ rather than $\lambda(k)$)  is not worse than the cost of the solution returned by our algorithm. The claim follows since the expected cost of $S'$  is  
$$
a(k) (\lambda |S_1(k)| + d(S_1(k))+(1-a(k))(\lambda |S_2(k)| + d(S_2(k))\leq \lambda\, k+(2-\eta)\opt. 
$$ 
\end{proof}

\subsection{Improved Upper Bounds on $S_1$}
\label{sec:boundsS1}

We next prove Lemma \ref{lem:improvedS1}. To that aim, we next assume $a>0$.

Recall that, for a parameter $\delta\in [0,1/2]$ to be fixed later, we defined a many-to-one matching between $S_1$ and $\OPT$, where a facility $f^*\in OPT$ is matched with a set $M(f^*)$ of facilities in $S_1$ which are $1/2$-captured by $f^*$ and which together $(1-\delta)$-capture $f^*$. We next use $S^M_1$ instead of $S^M$ and similarly for related quantities to stress that we are focusing on the locally optimal solution $S'=S_1$. With the usual notation, we let $\alpha^L:=\opt^L/\opt$, $\alpha^M:=\opt^M/\opt$, $\alpha^{MM}=\opt^{MM}/\opt$, $\beta_1=d_1/\opt$, and $\beta_1^{MM}=d^{MM}_2/\opt$. Let also $k^L=|\OPT^L|$, $k^M=|\OPT^M|$, $k^L_1=|S_1^L|$, and $k^M_1=|S_1^M|$. Finally, we let $d_1^L$ be the connection cost in $S_1$ related to the clients served by lonely facilities, and set $\beta_1^L=d_1^L/\opt$.

We next apply Theorem \ref{thr:localSearch} to $S_1$, neglecting the term depending on $\eps$ for the usual reasons. Hence we get
\begin{equation}\label{eqn:S1_boundLS}
 \lambda k_1+d_1 \leq \lambda k  +  (1+2 \alpha^{L} +  \frac{\delta}{1-\delta} (\beta_1^{MM} + \alpha^{MM}))\cdot \opt=:\lambda k+\rho^A_1(\delta,\alpha^L,\alpha^{MM},\beta_1^{MM})\opt.
\end{equation}
Observe that the above bound already implies an improved upper bound on the cost of $S_1$ (w.r.t. $\lambda k+2\opt$) when $\alpha^L$ and $\delta$ are sufficiently small and $a>0$. Indeed, from $ad_1+(1-a)d_2\leq 2\opt$ and recalling that $\beta=d_2/\opt$, we derive: 
\begin{equation}\label{eqn:UBbeta1}
\beta_1^{MM}\leq \beta_1=\frac{d_1}{\opt}\leq \frac{2\opt-(1-a)d_2}{a\cdot \opt}=\frac{2-(1-a)\beta_2}{a}\leq \frac{2}{a}.
\end{equation}

We next derive an alternative bound which implies the claim in the complementary case. Similarly to Section \ref{sec:improvedkMedian:boundS2}, we show that the facility cost $\lambda |OPT^L|$ due to lonely facilities in $\OPT$ is upper bounded by $O(\opt^L)$. This will allow us to exploit Corollary \ref{cor:modifiedJMSclaim}.

Let us partition $\OPT^L$ in two sets $D_A,D_B$ as follows. We define a directed graph $G^L$ over the node set $\OPT^L$. Let $cl(f)$ be the closest facility to $f$ in $\OPT$. For each $f^*\in \OPT^L$, we  add a directed edge $(f^*,cl(f^*))$ iff  $cl(f^*)\in \OPT^L$. In particular each node $f^*$ in $G^L$ has outdegree $1$ or $0$, the latter case happening if $cl(f^*)\in \OPT^M$. 
\begin{lemma}\label{lem:bipartite}
By breaking ties properly,  one can guarantee that cycles of $G^L$ have length at most $2$ (in particular, $G^L$ is bipartite). 
\end{lemma}
\begin{proof}
Suppose that we sort the nodes arbitrarily, and in case of ties we direct the edge leaving $v$ to the nearest node $u$ that appears first in this sorted list. Assume by contradiction that there exists a directed cycle $C=(v_1,v_2,\ldots,v_k,v_1)$ with $k\geq 3$. Assume w.l.o.g. that $v_1$ appears last in the sorted list among nodes of $C$. Let $d_{i,j}$ denote the distance between $v_i$ and $v_j$. By the minimality of the distances corresponding to the arcs of $G^{L}$ and by the symmetry of the distances between nodes in the original graph we have:
$$
d_{1,2}\leq d_{1,k}\leq  d_{k-1,k}\leq d_{k-2,k-1}\leq \ldots \leq d_{2,3}\leq d_{1,2}.
$$ 
This implies that all the distances corresponding to the arcs in $C$ are identical. This is however a contradiction since node $v_{k}$ should have its outgoing arc directed towards $v_{k-1}$ rather than towards $v_1$ since $v_{k-1}$ appears earlier than $v_1$ in the sorted list and $d_{1,k}=d_{k-1,k}$.
\end{proof}
The sets $D_A$ and $D_B$ are defined by computing a bipartition of $G^L$. Let $k_A=|D_A|$ and $k_B=|D_B|$. Notice that $k^L:=k_A+k_B$. 
\begin{remark}
$k_A$ and $k_B$ are potentially larger than any constant, hence removing $D_A$ or $D_B$ from $S_1$ might not be an option considered by the local search algorithm. Nonetheless, we can to upper bound the increase of the connection cost of $S_1$ due to the removal of such sets.
\end{remark}

The proof of the following lemma is analogous to the proof of Lemma \ref{lem:boundLonelyFacilityCost}, with some subtle technical differences due to the fact that, unlike $S_2$, $\OPT$ is not a locally optimal facility location solution (hence leading to a slightly weaker bound).
\begin{lemma}\label{lem:boundDeletionOPT}
Let $\OPT_X=\OPT\setminus D_X$ and $d_X=d(\OPT_X)$ for $X\in \{A,B\}$. Then 
$$
d_A+d_B \leq 2\opt+\frac{1}{\delta}(\opt+d_1+\opt^L+d^L_1).
$$
\end{lemma}
\begin{proof}
Consider a given $D_X$, and any client $c$ served by $f^*\in D_X$ in $\OPT$. We upper bound the connection cost of $c$ in $\OPT_X=\OPT\setminus D_X$ by upper bounding the distance between $c$ and $cl(f^*)$. Notice that $cl(f^*)\in \OPT_X$ by construction. In order to do that we describe a random path $P(c)$ between $c$ and some facility $f'\in \OPT$ distinct from $f^*$ (possibly $f'\in D_X$). This path $P(c)$ always includes $f^*$ (this is a technical difference w.r.t. Lemma \ref{lem:boundLonelyFacilityCost}). Since 
$$
dist(c,cl(f^*))\leq dist(c,f^*)+dist(f^*,cl(f^*)) \leq dist(c,f^*)+dist(f^*,f'),
$$
the expected length $\ell(P(c))$ of $P(c)$ is a valid upper bound on the connection cost of $c$ in $\OPT_X$. 

The path $P(c)$ is defined as follows. Let $C'(f^*)$ be the clients served by $f^*$ in $\OPT$ which in $S_1$ are served  by a facility which is not $1/2$-captured by $f^*$. Notice that $|C'(f^*)|\geq \delta |C_{\OPT}(f^*)|$ since otherwise $f^*$ would be matched. Let $next(f^*)$ be a random client in $C'(f^*)$. Let $f_1=S_1(next(f^*))$ be the facility serving $next(f^*)$ in $S_1$. Let $rand(f^*)$ be a random client served by $f_1$ in $S_1$ but not by $f^*$ in $\OPT$. Finally let $f'\neq f^*$ be the facility serving $rand(f^*)$ in $\OPT$. Then $P(c)$ is the concatenation of the shortest paths from $c$ to $f^*$, from $f^*$ to $next(f^*)$, from $next(f^*)$ to $f_1$, from $f_1$ to $rand(f^*)$, and from $rand(f^*)$ to $f'$.

By the above discussion, the increase of the connection cost of $\OPT_A$ and $\OPT_B$ together w.r.t $2\opt$ 
is upper bounded by
\begin{align*}
& \sum_{f^*\in \OPT^L}\sum_{c\in C_{\OPT}(f^*)}\ell(P(c))-\opt(c)\\
= & \sum_{f^*\in \OPT^L}|C_{\OPT}(f^*)|\left(\opt(next(f^*))+d_1(next(f^*))+d_1(rand(f^*))+\opt(rand(f^*))\right).
\end{align*}
Let us bound the terms in the above sum involving $next(f^*)$ and $rand(f^*)$ separately. For the part involving $next(f^*)$, we observe that each $c'\in C'(f^*)$ is selected as $next(f^*)$ with probability $\frac{1}{|C'(f^*)|}\leq \frac{1}{\delta|C_{\OPT}(f^*)|}$, and it that case contributes to the sum with $|C_{\OPT}(f^*)|(\opt(c')+d_1(c'))$. Thus
\begin{align*}
& Ex[\sum_{f^*\in \OPT^L}|C_{\OPT}(f^*)|\opt(next(f^*))+d_1(next(f^*))]  \leq \sum_{f^*\in \OPT^L}\sum_{c'\in C'(f^*)}\frac{1}{\delta}(\opt(c')+d_1(c'))\\
\leq & \frac{1}{\delta}\sum_{f^*\in \OPT^L}\sum_{c\in C(f^*)}\opt(c)+d_1(c)\leq \frac{1}{\delta}(\opt^{L}+d_1^{L}).
\end{align*}
Consider next the terms depending on $rand(f^*)$. For a client $c$, define $rand^{-1}(c)$ as the facilities $f^*$ with $rand(f^*)=c$. Then the expected contribution of $c$ to the considered terms in the sum is
$$
(\opt(c)+d_1(c))\cdot  Ex[\sum_{f^*\in rand^{-1}(c)}|C_{\OPT}(f^*)|]
$$ 
Let $f_1=S_1(c)$ be the facility serving $c$ in $S_1$. Let us partition $C_{S_1}(f_1)$ in subsets $C^1,\ldots,C^h$ depending on the facilities $f^1,\ldots,f^h$ serving them in $\OPT$. W.l.o.g., assume $c\in C_1$. Notice that for each $f^i$, $i\geq 2$, $f^i$ selects $c$ as $rand(f^i)$ with probability $\frac{1}{|C_{S_1}(f_1)|-|C^i|}$. Then
$$
Ex[\sum_{f^*\in rand^{-1}(c)}|C_{\OPT}(f^*)|]=\sum_{i\geq 2}\frac{|C(f^i)|}{|C^i|}\frac{|C^i|}{|C_{S_1}(f_1)|-|C^i|}\leq \frac{1}{\delta}\sum_{i\geq 2}\frac{|C^i|}{|C_{S_1}(f_1)|-|C^i|}.
$$ 
As in the proof of Lemma \ref{lem:boundLonelyFacilityCost}, the righthand side of the above inequality is maximized for $h=2$, in particular being $|C^2|/|C^1|$. Since $f^2$ does not $1/2$-capture $f_1$ by construction, we have that $|C^2|\leq \frac{1}{2}|C_{S_1}(f_1)|$ and hence $|C^2|/|C^1|\leq 1$. Altogether
$$
Ex[\sum_{f^*\in rand^{-1}(c)}|C_{\OPT}(f^*)|]\leq \frac{1}{\delta}.
$$
Thus the terms involving $rand(f^*)$ altogether contribute in expectation with at most
$$
\sum_{c\in C}\frac{1}{\delta}(\opt(c)+d_1(c))=\frac{1}{\delta}(\opt+d_1).
$$
The claim follows.
\end{proof}
As a corollary, we achieve the desired alternative bound on the total cost of $S_1$.
\begin{corollary}\label{cor:boundDeletionOPT}
One has 
$$
\lambda k_1+d_1\leq \lambda (k-\frac{k^L}{2})+(2+\frac{1+\beta_1+\beta_1^L+\alpha^L}{\delta})\cdot \opt. 
$$
\end{corollary}
\begin{proof}
Since $S_1$ is LMP 2-approximate, from the existence of the solutions $\OPT_A$ and $\OPT_B$ one gets $\lambda k_1+d_1\leq \lambda (k-k_A)+2d_A$ and $\lambda k_1+d_1\leq \lambda (k-k_B)+2d_B$. The claim follows by averaging and Lemma \ref{lem:boundDeletionOPT}.
\end{proof}

Let $\eta\in [0,1]$ be a parameter to be fixed later.  Next we distinguish two cases. If $\frac{\lambda k^{L}}{2\opt}\geq \frac{1+\beta_1+\beta_1^L+\alpha^L}{\delta}+\eta$, then 
\begin{equation}\label{eqn:S1_boundLargeD}
\lambda k_1+d_1\overset{Cor.\, \eqref{cor:boundDeletionOPT}}{\leq} \lambda k +(2-\eta)\opt.
\end{equation}
Otherwise, we have
$$
\lambda k^L\leq 2(\frac{1+\beta_1+\beta_1^L+\alpha^L}{\delta}+\eta)\opt = 2(\frac{1+\beta_1+\beta_1^L}{\alpha^L\delta}+\frac{1}{\delta}+\eta)\opt^L=:T_1 \opt^L.
$$
In this case we can apply Corollary \ref{cor:modifiedJMSclaim}
with $\OPT'=\OPT^L$ and $T=T_1$ to infer
\begin{equation}\label{eqn:S1_boundLMP}
\lambda k_1+d_1 \leq \lambda k+(2\cdot (1-\alpha^L)+\opt^{+}_{JMS}(q,T_{1})\cdot\alpha^L)\opt=:\lambda k+\rho^B_1(\delta,\alpha^L,\beta_1,\beta^L_1,\eta).
\end{equation}

From the above discussion we get 
$$
\lambda k_1+d_1\leq \lambda k+(2-\eta_1(a))\opt
$$ 
where
$$
2-\eta_1(a)\geq \min\{\rho^A_1(\delta,\alpha^L,\alpha^{MM},\beta_1^{MM})),\max\{2-\eta,\rho^B_1(\delta,\alpha^L,\beta_1,\beta^L_1,\eta)\}\}.
$$
We can simplify the above formula by making some pessimistic choices. From \eqref{eqn:UBbeta1} and assuming pessimistically that $\beta_2=0$, we can set $\beta_1=2/a$. We can pessimistically upper bound $\alpha^{MM}$ with $\alpha^M=1-\alpha^L$ and $\beta_1^{MM}$ with $\beta_1^{MM}=\beta_1-\beta^L_1=\frac{2}{a}-\beta^L_1$. Then we can choose freely the value of $\delta$ to minimize  the lower bound on $2-\eta_1(a)$. For a given $\delta$, we can pessimistically choose the values of $\alpha^L\in [0,1]$ and $\beta_1^L\in [0,\frac{2}{a}]$ that maximize the minimum. Given that, we can freely choose $\eta$ to minimize the maximum. Altogether we obtain
$$
2-\eta_1(a) = \min_{\delta\in [0,1/2]}\max_{\alpha^L\in [0,1],\beta^L_1\in [0,\frac{2}{a}]}\min_{\eta\in [0,1]}\min\{\rho^A_1(\delta,\alpha^L,1-\alpha^L,\frac{2}{a}-\beta^L_1),\max\{2-\eta,\rho^B_1(\delta,\alpha^L,\frac{2}{a},\beta_1^L,\eta)\}\}
$$
It is not hard to see that for $a>0$ one has $\eta_1(a)>0$. Indeed, for positive and small enough $\delta$ and $\alpha^L$, one has $\rho^A_1(\delta,\alpha^L,1-\alpha^L,\frac{2}{a}-\beta^L_1)<2$. Otherwise, we can assume that $\delta$ and $\alpha^L$ are both bounded away from zero. For any positive $\eta$ trivially $2-\eta<2$. Finally, for $\eta$ small enough one has that $T_1$ is upper bounded by a constant. It turns out that $\opt^{+}_{JMS}(q,T_{1})$ is strictly smaller than $2$ for any constant $T_1$ and large enough $q$. This implies $\rho^B_1(\delta,\alpha^L,\frac{2}{a},\beta_1^L,\eta)<2$. More formally, one can replace $\opt^{+}_{JMS}(q,T_{1})$ in the above analysis with the analytical upper bound $\opt'_{JMS}(T_1)$ on $\opt_{JMS}(T_{1})$ obtained in Section \ref{sec:LMP_old}. It is not hard to show analytically that $\opt'_{JMS}(T_1)<2$ for any constant $T_1$, hence the claim.

\section{A Refined Approximation Factor for $k$-Median}
\label{sec:refinedkMedian}

In this section we sketch how to refine the approximation factor for $k$-Median from Section \ref{sec:improvedkMedian}. Consider a bipoint solution $S_B=aS_1+(1-a)S_2$ computed as in Section \ref{sec:improvedkMedian}. With the usual notation, $d_i=d(S_i)$ and $\beta_i=d_i/\opt$. Recall that Li and Svensson \cite{LiS16} present a $2(1+2a)+\eps$ approximation for any constant $\eps > 0$ (with extra running time $n^{\poly(1/\eps)}$), hence we can also assume that $a$ is bounded away from $0$ since we aim at an approximation factor larger than $2$. Notice also that $a\beta_1+(1-a)\beta_2\leq 2$, hence $\beta_1\leq \frac{2}{a}$. This means that $S_1$, which is a feasible solution, provides a $2/a$ approximation. Hence we can also assume that $a$ is bounded away from $1$ for similar reasons. Recall that, for $a$ bounded away from $0$ and $1$, Li and Svensson \cite{LiS16} present a $\frac{2(1+2a)}{(1+2a^2)}+\eps$ approximation that we will next use.

We will consider the best of $3$ approximate solutions:
\begin{enumerate}
\item The solution $S_1$ alone, which provides a $\beta_1$ approximation.
\item The solution $S_{LS}$ computed with the mentioned $\frac{2(1+2a)}{(1+2a^2)}+\eps$ approximation in \cite{LiS16}.
\item The solution obtained  by rounding our bipoint solution $(S_1,S_2)$ with the $\rho_{BR}<1.3371$ rounding algorithm in \cite{ByrkaPRST17}.
\end{enumerate}
Let us focus on the latter solution. The discussion from Section \ref{sec:improvedFLuniform} directly implies that this solution has approximation ratio at most
$$
\rho_{BR}(a\beta_1+(1-a)\beta_2)\leq \rho_{BR}(a(2-\eta_1(a))+(1-a)(2-\eta_2))=\rho_{BR}(2-a\eta_1(a)-(1-a)\eta_2).
$$
This already gives an improvement on the approximation factor from Section \ref{sec:improvedkMedian} due to the extra term $-a\eta_1(a)$ which is strictly positive for $a>0$. We can further refine this factor as follows. In the computation of $\eta_1(a)$ we assumed pessimistically that $\beta_1=2/a$, and in the computation of $\eta_2$ that $\beta_2=2$. However these two conditions cannot simultaneously hold unless $a=1$. In the latter case however $S_1$ provides a $2$ approximation already. Under the assumption $a<1$ and fixing a value $\beta_1\in [2,2/a]$, we can rather use $\frac{2-a\beta_1}{1-a}$ as a pessimistic upper bound on $\beta_2$. Hence with essentially the same analysis as before we get two new functions $\eta_1(a,\beta_1)\geq \eta_1(a)$ and $\eta_2(a,\beta_1)\geq \eta_2$ for a fixed $\beta_1$, leading to the refined approximation factor
$$
\rho_{BR}(2-a\eta_1(a,\beta_1)-(1-a)\eta_2(a,\beta_1)).
$$
Taking the best of the $3$ mentioned approximate solutions, and choosing the worst possible parameter $a\in [0,1]$ and $\beta_1\in [2,2/a]$,
we obtain the approximation factor:
$$
\rho_{kMed}=\max_{a\in [0,1],\beta_1\in [2,\frac{2}{a}]}\min\{\beta_1,\frac{2(1+2a)}{1+2a^2}+\eps,\rho_{BR}\left(2-a\eta_1(a,\beta_1)-(1-a)\eta_2(a,\beta_1) \right)\}
$$

\newcommand{\extendJMS}{\texttt{Extend-JMS}}
\newcommand{\localsearchJMS}{\texttt{LocalSearch-JMS}\xspace}

\section{An Improved LMP Approximation for Facility Location:\\ General Opening Costs}
\label{sec:fl}
Consider a general facility location instance where each facility $f \in F$ has a cost $\open(f)$.
In this section, we prove the following theorem for Facility Location with general facility cost.
\begin{theorem}
There exists an absolute constant $\eta > 0$ such that 
there is a polynomial time LMP $(2-\eta)$ approximation algorithm for Facility Location with general facility costs. 
\label{thr:facility_location}
\end{theorem}
\noindent 
The current (unoptimized) lower bound for $\eta$ is $2.25 \cdot 10^{-7}$.

\subsection{New Local Search and Facility Deletion}
\label{subsec:newlocal}
Given an instance of Facility Location, our algorithm runs the JMS algorithm to get the initial solution $S'$ and improves it via local search until no local improvement is possible. 
In this subsection, we prove that when the final solution $S'$ has a relatively small connection cost, we already improve on an LMP $2$ approximation. 
The proof is similar to the one that we presented for uniform facility costs, but has additional ideas that may be of independent interest. 
In fact, the following example shows that the standard local search where one swaps a constant number of facilities will not give any LMP approximation for a large class of objective functions. 
\begin{claim}
For any constants $\Delta \in \N$ and $\alpha, \beta \in \R^+$, 
there is an instance of general Facility Location $(C, F, dist, o)$ that admits feasible solutions $\OPT, S \subseteq F$ such that 
(1) $S$ is locally optimal for the objective function $\alpha \cdot \open(S) + \beta \cdot d(S)$ and the width-$\Delta$ local search, but (2) $\open(\OPT) + t \cdot d(\OPT) < \open(S) + d(S)$ for any $t \in \R$. 
\label{claim:ls_bad}
\end{claim}
\begin{proof}
For some $n$ to be determined, consider the instance of UFL where $F = \{ f_0, \dots, f_n \}$, $C = \{ c_1, \dots, c_n \}$ where $dist(f_i, c_i) = 0$ and $dist(f_0, c_i) = 1$ for every $i \in [n]$.
We set $\open(f_0) = x$ and $\open(f_i) = y$ for $i \in [n]$ for some $x, y > 0$ such that $y > \beta/\alpha$, $\alpha (x - ry) < \beta (n - 2r)$ for every $r \in [\Delta]$, and $ny < x + n$. Such $x$ and $y$ exist for large enough $n$ depending on $\Delta, \alpha, \beta$ (e.g., take $y \approx \beta / \alpha$ for the first inequality, $x \approx n \cdot \min( \beta / \alpha - 1, 0)$ for the third, and $n$ large enough for the second). 
This finishes the description of the instance. 
Let $\OPT = \{ f_1, \dots, f_n \}$ and $S = \{  f_0 \}$. Note that $\OPT$ has $\open(OPT) = n \cdot y$ and $d(OPT) = 0$, and
$\open(S) = x$ and $d(S) = n$. 

From $S$, the width-$\Delta$ local search algorithm may close $f_0$ or not. 
Suppose that it considers a new solution $S'$ where $f_0$ is closed and $r \leq \Delta$ facilities from $\OPT$ are open. 
Then $\open(S') = r \cdot y$ and $d(S') = 2(n - r)$. 
Since $\alpha x + \beta n < \alpha r \cdot y + \beta \cdot 2(n - r) \Leftrightarrow 
\alpha(x - ry) < \beta(n - 2r)$, 
$S$ is better than $S'$ in terms of the objective function $\alpha \cdot \open(S) + \beta \cdot d(S)$. 
For the other case, when the local search considers a new solution $S'$ where $f_0$ is kept open and $r \leq \Delta$ facilities from $\OPT$ are open, 
$\open(S') = x + r\cdot y$ and $d(S') = (n - r)$.
Since $\alpha x + \beta n < \alpha (x + r \cdot y) + \beta \cdot (n - r) \Leftrightarrow \beta / \alpha < y$, $S$ is better than $S'$ in terms of the objective function $\alpha \cdot \open(S) + \beta \cdot d(S)$.  
Therefore, one can conclude that $S$ is locally optimal with respect to $\Delta$-local moves and the objective function $\alpha \cdot \open(S) + \beta \cdot d(S)$, but since $ny < x + n$, 
$\open(\OPT) + t \cdot d(\OPT) < \open(S) + d(S)$ for any finite $t$. 
\end{proof}
Therefore, our local search considers a new kind of local moves where we remove one facility from $S'$ and augment $S'$ by running the JMS algorithm on top of it.
More concretely, our local search algorithm \localsearchJMS goes as follows.
\begin{enumerate}
\item Start with an LMP $2$ approximate solution $S'$ (e.g., as returned by JMS).
\item While there exists a solution $S''$ such that either\\
  (1) $|S'' - S'| + |S' - S''| \le 1/\eps+1$ or\\
  (2) $S'' =$  \extendJMS($S' - \{f\} \cup \{f'\}$) for any $f \in S'$, $f' \notin S'$, \\
  and such that
  $\cost(S'') < \cost(S')$.
  \begin{enumerate}
  \item Do $S' \gets S''$.
  \end{enumerate}
\item Return $S'$.
\end{enumerate}
The procedure \extendJMS$(X)$ is simply
\begin{enumerate}
\item Modify the instance so that all facilities in $X$ have opening cost 0.
\item Run the standard JMS on the new instance.
\end{enumerate}

Let $S'$ be the solution at the end of \localsearchJMS (i.e., $S'$ is a local optimum).
Let $\OPT$ denote some given facility location solution.
Let $d'$ and $d^* = \opt$ be the total connection cost of $S'$ and $\OPT$ respectively.
The main result of this subsection is the following lemma. 

\begin{lemma}
There exists an absolute constant $\eta \geq 4.5 \cdot 10^{-7}$ such that 
if $S'$ is a local optimum of \localsearchJMS and $d' \leq 4\opt$, $\open(S') + d' \leq \open(OPT) + (2-\eta)d^*$.
\label{lem:facility_location_fixed_scale}
\end{lemma}

\paragraph{Matching and Local Search.} 
We classify the facilities in $S'$ and $\OPT$ as lonely and matched as follows. Let $0 < \delta'_1 < \delta_1 \leq 1/2$, $0 < \delta'_2 < \delta_2 \leq 1/2$ be parameters to be determined later. Say $f' \in S'$ and $f^* \in \OPT$ are {\em matched} if $f'$ $(1-\delta_1)$-captures $f^*$ and $f^*$ $(1-\delta_2)$-captures $f'$.
All the facilities which are not matched are \emph{lonely}. 
As usual, let $\OPT^M$, $\OPT^L$, $S^M$, and $S^L$ denote the set of matched facilities in $\OPT$, lonely facilities in $\OPT$, matched facilities in $S'$, and lonely facilities in $S'$ respectively. 
Let $\opt^M$ (resp. $\opt^L$) be the total connection cost of the clients served by a matched facility in $\OPT^M$ (resp. $\OPT^L$). 
Also for a client $c \in C$, let $d'(c)$ and $d^*(c)$ of the connection cost of $c$ in $S'$ and $\OPT$ respectively. 

The following theorem proved in Section~\ref{sec:facilityLS} shows that \localsearchJMS already gives the desired guarantee when $\opt^L$ is small and $d'$ is not too large compared to $\opt$. 
\begin{theorem}\label{thr:extendJMS}
Fix $\delta_1=\delta$ and $\delta_2=1/2$ in the above definition of matching. 
Let $S'$ be the solution at the end of \localsearchJMS.
Then one has
$
\open(S')+d' \leq
   \open(\OPT)  + \frac{\delta}{1-\delta} d' + \frac{1}{1 - \delta} \opt^{M} +   4 \opt^{L} + O(\eps(d' + opt)).
  $
\end{theorem}

\paragraph{Bounding $\open(S^L)$.} Using the local optimality of $S'$, one can prove the following lemma. 
\begin{lemma}\label{lem:boundLonelyFacilityCostFL}
There exists a constant $t = t(\delta_1, \delta'_2)$ such that 
\begin{equation}
\open(S^L) \leq \frac{1-\delta_2}{1 - \delta'_2} \cdot \open(OPT^L) + t\cdot (\opt + d'). 
\end{equation}
\end{lemma}
\begin{proof}
Consider a $f' \in S^L$. We will consider a (randomized) local move that closes $f'$ and opens at most one facility in $\OPT$.
Let $M(f') \subseteq \OPT$ be the facilities that are $(1-\delta_1)$-captured by $f'$. 
We consider the following cases. 
For a client $c$, let $f^*(c)$ (resp. $f'(c)$) be the facility serving $c$ in $\OPT$ (resp. $S'$). 

\begin{enumerate}
\item $M(f')$ does not $(1-\delta'_2)$ capture $f'$: Here we close $f'$ without opening any facility, leading to a decrease  of the opening cost by $\open(f')$. Consider the following randomized rerouting of the clients in $C(f')$ to $S' \setminus \{ f' \}$.

\begin{itemize}
\item If $c \in C(f') \setminus C(M(f'))$: Choose a random client $rand(c)$ from $C(f^*(c)) \setminus C(f')$. Reroute $c$ to $f'(rand(c))$. The expected increase in the connection cost for $c$ is at most 
\begin{equation}\label{eq:fl-deletion-1}
d^*(c) + \E_{rand(c)}[d^*(rand(c)) + d'(rand(c))] - d'(c).
\end{equation}
\item If $c \in C(f') \cap C(M(f'))$: Choose a random client $next(c)$ from $C(f') \setminus C(M(f'))$. As in the previous case, choose a random $rand(next(c))$ from $C(f^*(next(c))) \setminus C(f')$ and reroute $c$ to $f'(rand(next(c)))$. The increase in the connection cost for $c$ is at most 
\begin{equation}\label{eq:fl-deletion-2}
\E_{next(c)}[d'(next(c)) + d^*(next(c))] + \E_{rand(next(c))}[d^*(rand(next(c))) + d'(rand(next(c)))].
\end{equation}

\end{itemize}

\item If $M(f')$ does $(1-\delta_2')$ captures $f'$: Here we open one facility $f^* \in M(f')$ with probability $|C(f^*) \cap C(f')| / |C(M(f')) \cap C(f')|$ (note that the sum of the probabilities is $1$). Each $c \in C(f')$ is rerouted to the open facility $f^*$. 
Note that for each $c \in C(f')$, this is equivalent to choosing a random client $rand'(c)$ uniformly from $C(M(f')) \cap C(f')$ and rerouting $c$ to $f^*(rand'(c))$. 
Therefore, in expectation, the connection cost for $c$ is increased by
\begin{equation}\label{eq:fl-deletion-3}
\E_{rand'(c) \in C(M(f')) \cap C(f')} [d'(rand'(c)) + d^*(rand'(c))].
\end{equation}
\end{enumerate}

Since $S'$ is local optimal, all the (randomized) local moves considered above do not improve the cost of $S'$. Consider the sum of all the (expected) increased costs, considering both opening and connection cost. This sum has the following terms. 
\begin{enumerate}
\item $-\open(f')$ for every $f' \in S^L$. 
\item At most $\frac{1 - \delta_2}{1 - \delta'_2} \open(f^*)$ for every $f^* \in \OPT^L$. Note that $M(f') \subseteq \OPT^L$ for every $f' \in S^L$, and $M(f') \cap M(f'') = \emptyset$ for $f' \neq f'' \in S^L$. 
When $M(f')$ $(1-\delta_2')$-captures $f'$ and $f^*\in M(f')$ does not $(1 - \delta_2)$-capture $f^*$, $f^*$ is open with probability at most $\frac{1 - \delta_2}{1 - \delta'_2}$. (If $f^*$ $(1-\delta_2)$-captured $f'$, then they should be matched.) 
\item At most $d^*(c)$ for every $c \in C$. (From the sum of $d^*(c)$ in~\eqref{eq:fl-deletion-1}).
\item For every $c'$ such that $f'(c') \in S^L$ and $f^*(c') \notin M(f'(c'))$: 
\begin{enumerate}
\item \label{item:rand} 
In the sum of $\E_{next(c)}[d'(next(c)) + d^*(next(c))]$ in~\eqref{eq:fl-deletion-2}, 
$(d^*(c') + d'(c'))$ will appear with coefficient at most $\frac{1 - \delta_2'}{\delta_2'}$; $c'$ can be sampled by $c \in C(f') \cap c(M(f'))$,
and each such $c$ samples $next(c)$ from $|C(f') \setminus c(M(f'))|$ clients. In~\eqref{eq:fl-deletion-2}, we have $|C(f') \setminus c(M(f'))| \geq \frac{\delta'_2}{1-\delta'_2} |C(f') \cap C(M(f'))|$. 
\end{enumerate}
\item For every $c' \in C$: 
\begin{enumerate}
\item \label{item:next} 
In the sum of $\E_{rand(c)}[d^*(rand(c)) + d'(rand(c))]$ in~\eqref{eq:fl-deletion-1}, $(d^*(c') + d'(c'))$ will appear with coefficient at most $\frac{1}{\delta_1}$, because $c'$ can be possibly sampled by $c \in C(f^*(c)) \cap C(f'')$ where $f''$ does not capture $f^*(c)$ but each such $c$ samples $c'$ from at least $|C(f^*(c)) \setminus C(f'')| \geq \delta_1|C(f^*(c))|$ clients. 

\item In the sum of $\E_{rand(next(c))}[d^*(rand(next(c))) + d'(rand(next(c)))]$ in~\eqref{eq:fl-deletion-2}, $(d^*(c') + d'(c'))$ will appear with coefficient at most $\frac{1 - \delta_2'}{\delta_1 \delta_2'}$; each $c'' \in C$ will be $next(c)$ at most $\frac{1 - \delta_2'}{\delta'_2}$ times in expectation by the same argument as in item 4.(a), and using the same argument as in item 5.(a), the expected number of times $c'$ is chosen as $rand(next(c))$ is at most $\frac{1}{\delta_1}$ times the maximum expected number of times any $c'' \in C$ is chosen as $next(c)$. 
\end{enumerate}
\item For every $c'$ such that $f' \in S^L$ and $f^*(c') \in M(f'(c'))$:
\begin{enumerate}
\item  In the sum of $\E_{rand'(c) \in C(M(f')) \cap C(f')} [d'(rand'(c)) + d^*(rand'(c))]$ 
in~\eqref{eq:fl-deletion-3}, $(d^*(c) + d'(c'))$ will appear with coefficient at most $\frac{1}{1 - \delta'_2}$; $c'$ can be possibly sampled by $c \in C(f')$ and each such $c$ samples from at least $|C(M(f')) \cap C(f')| \geq (1 - \delta_2') |C(f')|$ clients. 
\end{enumerate}
\end{enumerate}

Therefore, if we add 3, 4, 5, 6 above over each client, we have 
\begin{equation}\label{eq:fl-deletion-4}
\sum_{f' \in S^L} \open(f') \leq \frac{1-\delta_2}{1 - \delta'_2} \cdot \sum_{f^* \in \OPT^L}  \open(f^*) + t(d^* + d'),
\end{equation}
where $t = 1 + \frac{1}{\delta_1\delta_2'} + \max(\frac{1 - \delta'_2}{\delta'_2}, \frac{1}{1 - \delta'_2})$ (note that 4 and 6 are mutually exclusive so we have the maximum instead of the sum).
\end{proof}

\paragraph{Bounding $\open(\OPT^L)$.}
Similarly to Section~\ref{sec:boundsS1}, 
partition $\OPT^L$ into $D_A$ and $D_B$ such that for every $f \in D_A$ (resp. $f\in D_B$), one of the facilities closest to $f$ in $\OPT$ is in $\OPT \setminus D_A$ (resp. $\OPT \setminus D_B$). 

\begin{lemma}\label{lem:boundDeletionOPTFL}
Let $t = t(\delta_1, \delta'_2)$ be the constant determined in Lemma~\ref{lem:boundLonelyFacilityCostFL}.
There exists a constant $t' = t'(\delta_2, \delta_1')$ and a randomized solution $\OPT^{\dagger}$ such that 
\begin{itemize}
\item The expected facility cost of $\OPT^{\dagger}$ is at most $\open(\OPT) - \big(\frac{1}{2} - \frac{(1-\delta_1)(1-\delta_2)}{2(1-\delta'_1)(1-\delta'_2)}\big) \open(\OPT^L) + t\cdot (\opt + d')$. 
\item The expected connection cost of $\OPT^{\dagger}$ is at most $t'\cdot (\opt + d')$. 
\end{itemize}
\end{lemma}
\begin{proof} 
For each $X \in \{ A, B \}$, let $\OPT_X$ be the randomized solution achieved by applying the proof of the above Lemma~\ref{lem:boundLonelyFacilityCostFL} to delete each of $D_X$ and possibly reopen some facilities in $S^L$. $\OPT^{\dagger}$ samples $X \in \{A , B\}$ with probability $1/2$ each and samples $\OPT_X$. 

In $\OPT^{\dagger}$, each $f^* \in \OPT^L$ is deleted exactly with probability $1/2$, 
and each $f' \in S^L$ is reopen with probability at most $\frac{1 - \delta_1}{2(1 - \delta'_1)}$.
(Note that if $f' \in S^L$ reopens only when some $f^* \in \OPT^L$ capturing $f'$ is deleted, so it can possibly reopen in at most one of $\OPT_A$ and $\OPT_B$.)  
Therefore, the expected facility cost is at most 
\begin{align*}
& \, \, \open(\OPT) - \frac{\open(\OPT^L)}{2} + \frac{1 - \delta_1}{2(1 - \delta'_1)} \open(S^L) \\ 
\leq & \, \, \open(\OPT) - \frac{\open(\OPT^L)}{2} + \frac{1 - \delta_1}{2(1 - \delta'_1)} \bigg(\frac{1-\delta_2}{1 - \delta'_2} \cdot \open(OPT^L) + t(\opt + d') \bigg) \qquad \mbox{(Lemma~\ref{lem:boundLonelyFacilityCostFL})}\\
\leq & \, \, \open(\OPT) - \big(\frac{1}{2} - \frac{(1-\delta_1)(1-\delta_2)}{2(1-\delta'_1)(1-\delta'_2)}\big) \open(\OPT^L) + t(\opt + d').
\end{align*}
For the connection cost, since each $f' \in D_X$ has one of its closest facilities in $\OPT$ in $\OPT \setminus D_X$, 
the proof of Lemma~\ref{lem:boundLonelyFacilityCostFL} can be applied verbatim (switching the roles of $\delta_1 \leftrightarrow \delta_2$ and $\delta'_1 \leftrightarrow \delta'_2$) to show that the expected connection cost is $t'(\opt + d')$ for 
$t' = \frac{1}{2} ( 1 + \frac{1}{\delta_2\delta_1'} + \max(\frac{1 - \delta'_1}{\delta'_1}, \frac{1}{1 - \delta'_1}))$. 
(The additional factor $\frac{1}{2}$ comes from the fact that each $f^* \in \OPT^L$ is deleted with probability exactly $1/2$.) 
\end{proof}

Set $\delta_1 = \delta, \delta_2 = 1/2$, and 
assume that $d' \leq 4\opt$. 
By Theorem~\ref{thr:extendJMS},
\[
\open(S')+d' \leq
   \open(\OPT)  + \frac{\delta_1}{1-\delta_1} d' + \frac{1}{1 - \delta_1} \opt^{M} +   4 \opt^{L}.
\]
Let $\alpha^L = \opt^L / \opt$. Since 
\[
\frac{\delta_1}{1-\delta_1} d' + \frac{1}{1 - \delta_1} \opt^{M} + 4 \opt^{L} \leq \opt ( \frac{1 + 4\delta_1}{1 - \delta_1} + 4 \alpha^L), 
\]
if $\alpha^L \leq (1.9 - \frac{1 + 4\delta_1}{1 - \delta_1}) / 4$, $S'$ is already an LMP $1.9$ approximate solution. 
Therefore, we can assume that $\alpha^L \geq (1.9 - \frac{1 + 4\delta_1}{1 - \delta_1}) / 4$. (We will choose $\delta_1 > 0$ such that the RHS is strictly positive.) 

Also by Lemma~\ref{lem:boundDeletionOPTFL}, if we let $d^{\dagger}$ be the connection cost of $\OPT^{\dagger}$, 
\[
\open(S')+d' \leq
   \E[\open(\OPT^{\dagger})] + 2\E[d^{\dagger}] 
\leq \open(\OPT) - \zeta \open(\OPT^L) + t'' (\opt + d'),
\]
where $\zeta = \zeta(\delta_1, \delta_2, \delta'_1, \delta'_2) = \big(\frac{1}{2} - \frac{(1-\delta_1)(1-\delta_2)}{2(1-\delta'_1)(1-\delta'_2)}\big) > 0$ and 
$t'' = t''(\delta_1, \delta_2, \delta'_1, \delta'_2) = t + 2t' < \infty$. 

When $\zeta \cdot \open(\OPT^L) \geq t''(\opt + d')$ then $S'$ is already an LMP 1 approximate solution. 
Therefore, one can assume $\zeta \cdot \open(\OPT^L) < t''(\opt + d') \leq 3t'' \opt$. Together with $\opt^L \geq \opt \cdot (1.9 - \frac{1 + 4\delta_1}{1 - \delta_1}) / 4$, it upper bounds the ratio between the facility cost and the connection cost of $\OPT^L$ by some constant depending only on $\delta_1, \delta_2, \delta'_1, \delta'_2$. 
Therefore, we can apply an analogue of Corollary~\ref{cor:modifiedJMSclaim} to conclude that $S'$ is an LMP $(2 - \eta)$ approximation for some $\eta = \eta(\delta_1, \delta_2, \delta'_1, \delta'_2) > 0$. More specifically, a lower bound of $\eta$ can be obtained as follows. 

\begin{itemize}
\item Let $\delta$ be a free parameter to be determined and $\delta_1 = \delta, \delta'_1 = \delta, \delta_2 = 1/2, \delta'_2 = 1/4$. 
\item $t = 1 + \frac{1}{\delta_1\delta_2'} + \max(\frac{1 - \delta'_2}{\delta'_2}, \frac{1}{1 - \delta'_2}) = 4 + \frac{4}{\delta} \leq \frac{5}{\delta}$ for $\delta \leq \frac{1}{4}$.
\item $t' = \frac{1}{2} ( 1 + \frac{1}{\delta_2\delta_1'} + \max(\frac{1 - \delta'_1}{\delta'_1}, \frac{1}{1 - \delta'_1})) \leq \frac{1}{2} + \frac{3}{\delta} \leq \frac{4}{\delta}$. 
\item $\zeta = \big(\frac{1}{2} - \frac{(1-\delta_1)(1-\delta_2)}{2(1-\delta'_1)(1-\delta'_2)}\big) \geq \frac{1}{2} - \frac{1 - \delta_2}{2(1 - \delta'_2)} = \frac{1}{2} - \frac{1}{3} = \frac{1}{6}$.
\item $\opt^L \geq \opt \cdot (1.9 - \frac{1 + 4\delta}{1 - \delta}) / 4$.
\item $\open(OPT^L) \leq \frac{3(t + 2t')}{\zeta} \opt \leq 3 \cdot 6 \cdot (\frac{5}{\delta} + \frac{8}{\delta}) = \frac{234}{\delta}$. 
\item By Corollary~\ref{cor:modifiedJMSclaim} and the analytic bound on $opt_{JMS}(T)$ proved in Corollary~\ref{cor:jms-general} with $OPT' \leftarrow OPT^L$ and $T \leftarrow \frac{234/\delta}{(1.9 - \frac{1 + 4\delta}{1 - \delta}) / 4}$, we have 
$\eta \geq \frac{1.9 - \frac{1 + 4\delta}{1 - \delta} }{4} \cdot \frac{1}{4(7 + 3\frac{234/\delta}{(1.9 - \frac{1 + 4\delta}{1 - \delta}) / 4})}$. 
\item For $\delta = 0.05$, the lower bound becomes at least $4.5 \cdot 10^{-7}$. 
\end{itemize}

\subsubsection{Extend-JMS Local Search -- Proof of Theorem~\ref{thr:extendJMS}}
\label{sec:facilityLS}
We let $S_1$ denote the solution output by Extend-JMS local search and $\OPT$ be the optimum solution.
We will need the following lemma to analyse the \extendJMS\ swap.
\begin{lemma}
  \label{lem:extendJMS}
  Consider a set of facilities $S_0, S^*$, such that $S_0 \subseteq S^*$.
  Consider a not-necessarily-optimal assignment $\mu$ of clients to facilities in $S^*$, let $C_{\mu}(f)$ be the set
  of clients assigned to facility $f \in S^*$ in assignment $\mu$.
  The cost of the solution $S''$ (with optimal assignment) produced by \extendJMS($S_0$) is at most 
  \[\open(S^*) + \sum_{f \in S_0} \sum_{c \in C_{\mu}(f) } \dist(c,f) + 2\sum_{f \in S^* - S_0} \sum_{c \in C_{\mu}(f)} \dist(c,f).\]
\end{lemma}
\begin{proof}
We follow the standard JMS analysis. The only difference between \extendJMS($S_0$) and the standard JMS algorithm is that 
we pretend that the opening cost of each $f \in S_0$ is zero, while paying the original cost once $f$ is open. 
We conservatively assume that each $f \in S_0$ is open, paying $\open(S_0)$ from the beginning.
Then the cost of $S''$ is at most $\open(S_0)$ plus the cost of $S''$ in the facility location instance when each $f \in S_0$ has the opening cost zero. 

Let $\alpha : C \to \R$ be the dual variables during the JMS algorithm. 
For each $f \in S^* - S_0$, the standard LMP $2$ approximation guarantee ensures that $\sum_{c \in C_{\mu}^f} \alpha_c \leq \open(f) + 2 d(c, f)$. 
For each $f \in S_0$, we claim that $\sum_{c \in C_{\mu}^f} \alpha_c \leq d(c, f)$. This follows from the fact that $\alpha_c$ is exactly the time at which $c$ is connected to some facility for the first time; since $f$ has the opening cost $0$ and is at distance $d(c, f)$ from $c$, $c$ is surely connected to some facility no later than $d(c, f)$. 

Combining the two bounds, we conclude that the cost of $S''$ is at most 
\[
\open(S_0) + 
\bigg( \sum_{f \in S_0} \sum_{c \in C_{\mu}(f) } \dist(c,f) \bigg) + 
\bigg(\open(S^* - S_0) + 2\sum_{f \in S^* - S_0} \sum_{c \in C_{\mu}(f)} \dist(c,f) \bigg),
\]
which proves the lemma. 
\end{proof}

\begin{proof}[Proof of Theorem~\ref{thr:extendJMS}] 
Let $S'$ be the solution at the end of \localsearchJMS (i.e., a local optimum).
Let $d' = d(S')$ be the total connection cost of $S'$, and $d'(c) := \dist(c, S')$ be the connection cost of $c$. 
We analyse the algorithm by forming \emph{swap} pairs $(A_i,B_i)$, $A_i\subseteq S'$ and $B_i\subseteq \OPT$.
At the intuition level, 
  each swap leads to a solution  $S'-A_i \cup B_i$ whose cost is thus at least as high as $S'$.
Start by creating a pair $(\{f'\}, \emptyset)$ for each facility of $S'$. 
Then, each matched facility $f^*\in \OPT^M$ is added to the group of the corresponding (unique!) facility $f'\in \OPT^M$
it is matched to. Finally, each $f^*\in \OPT^L$ is added to the group of the closest $f'\in S'$.
Notice that each group now consists of exactly one facility $f'\in S'$ and that each facility $f' \in S'$ appears
in exactly one group. Notice also that each $f^*\in \OPT$ belongs to exactly one group.

We say that a group containing more than $1/\eps$ facilities of $\OPT$ is a \emph{heavy} group. A group that is not heavy is
\emph{light}.
For each heavy group $g = (\{f'\}, B)$, designate a special facility $g(B) \in B$ that is the facility that is matched
to $f'$ if $f' \in S^M$ and an arbitrary facility otherwise.
For each group $g = (\{f'\}, B)$, we will consider a specific solution $S''_g$. 
Let $s''_g = d(S''_g)$ and $s''_g(c) := \dist(c, S''_g)$. 
\begin{enumerate}
\item If $(\{f'\}, B)$ is light, we consider the solution $S''_g = S' - \{f'\} \cup B$. Since $g$ is light, we have
  by local optimality that $\cost(S') \le \cost(S''_g)$.
\item If $(\{f'\}, B)$ is heavy, we proceed as follows. If $\{f'\}$ is matched, then the corresponding facility $f^*$ is
  part of $B$ and we consider the solution $S''_g$ produced by \extendJMS($S' - \{f'\} \cup \{g(B)\}$).
  We have by local optimality that $\cost(S') \le \cost(S''_g)$.
\end{enumerate}

We now aim at bounding the cost of each solution $S''_g$ defined by the above swaps. Given a facility $f^* \in \OPT$,
recall that $C_{\OPT}(f^*)$ is the set of clients served by $f^*$ in solution $\OPT$.
We first consider the case where $g = (\{f'\}, B)$ is light.
We then partition the clients into four categories.
\begin{itemize}
\item Let $C_1$ be the set of clients that are either
  (1) served by a lonely facility in $B$ in $\OPT$, or (2) served by $f'$ in $S'$ and the facility of $\OPT$ that is matched to $f'$.
\item Let $C_2$ be the set of clients not in $C_1$ that are served by $\{f'\}$ in $S'$ and by a matched
  facility in $\OPT$.
\item Let $C_3$ be the set of clients not in $C_1$ that are served by $\{f'\}$ in $S'$ and by a lonely
  facility  in $\OPT$.
\item Finally let $C_4$
  be the remaining clients.
\end{itemize}
We let $s''_g(c)$ denote the cost of client $c$ in solution $S''_g$.
  We have
\begin{itemize}
\item If $c \in C_1$, $s''_g(c) \le \opt(c)$.
\item If $c \in C_2$, we consider the facility $f^*$ serving
  $c$ in $\OPT$. Let $(\{f_2\}, B')$ be the group such that $f^* \in B'$. Since $B \neq B'$ we have
  that $f_2 \in S''_g$ and so $s''_g(c) \le \dist(c, f_2)$, which by triangle inequality
  is at most $\opt(c) + \dist(f_2,f^*)$.
  We have that \[\dist(f_2,f^*) \le \frac{1}{(1-\delta)|C_{\OPT}(f^*)|} \sum_{c' \in C_{\OPT}(f^*)} (\opt(c') + d'(c')).\]
\item If $c \in C_3$, we consider the facility $f^*$ serving
  $c$ in $\OPT$. Let $(\{f_2\}, B')$ be the group such that $f^* \in B'$. Since $B \neq B'$ we have
  that $f_2 \in S''_g$ and so $s''_g(c) \le \dist(c, f_2)$, which by triangle inequality
  is at most $\opt(c) + \dist(f_2,f^*)$. We have that 
  \[\dist(f^*, f_2) \le  \dist(f^*, f') \le  \opt(c) + d'(c).\]
\item If $c \in C_4$, we have $s''_g(c) \le d'(c)$.
\end{itemize}
Recall that for each client $c$, let $\OPT(c)$ be the facility of $\OPT$ that serves it in $\OPT$.
Since $\cost(S') \le \cost(S''_g)$, we have $  \cost(S') - \cost(S''_g) \le 0$,
and so
\begin{align}
&  \open(f') - \open(B) + \sum_{c \in C_1} \big( \opt(c) - d'(c) \big) \nonumber \\
+ & \sum_{c \in C_2} \bigg(\opt(c) - d'(c) + \frac{\sum_{c' \in C_{\OPT}(\OPT(C))} \opt(c') + d'(c')}{(1-\delta)|C_{\OPT}(\OPT(c))|} \bigg)
+ \sum_{c \in C_3} 2 \opt(c) & \le 0 \label{eq:light}
\end{align}

We now turn to the case where $(\{f'\}, B)$ is heavy. Recall $S''_g= $ \extendJMS($S' - \{f'\} \cup \{g(B)\}$) and
let $S'_g = S' - \{f'\} \cup B$.
We consider the following assignment of points
to centers.
We then partition the clients into four categories and define a not-necessarily-optimal assignment $\mu$ of
points to center in $S'_g$.
\begin{itemize}
\item If $g(B)$ and $f'$ are matched, let $C^a_1$ be the set of clients that are served by $g(B)$ in $\OPT$ and
  by $f'$ in $S'$. Otherwise let $C^a_1$ be the whole set of clients served by $g(B)$ in $\OPT$.
  Let the clients in $C^a_1$ be assigned to $g(B)$ in $\mu$.
\item Let $C^b_1$ be the set of clients that are served by a in $B - g(B)$ in $\OPT$. This facility is necessarily
  a lonely facility.
  These clients are assigned to the closest facility in $B-g(B)$ in $\mu$.
\item Let $C_2$ be the set of clients not in $C^a_1 \cup C^b_1$ that are served by $\{f'\}$ in $S'$ and by a matched
  facility $f^*$ in $\OPT$.
  These clients are assigned to the facility of $S'$ that is matched to $f^*$ in $\mu$.
\item Let $C_3$ be the set of clients not in $C^a_1 \cup C^b_1$ that are served by $\{f'\}$ in $S'$ and by a lonely
  facility $f^*$ in $\OPT$.
  These clients are assigned to the facility of $S'$ that is the closest to $f^*$ in $\mu$.
\item Finally let $C_4$
  be the remaining clients.
  These clients are assigned to the closest facility of $S' - \{f'\}$ in $\mu$.
\end{itemize}

We now bound the cost of solution $S'_g$ with the assignment $\mu$, which will allow us to bound the cost of $S''_g$ via Lemma~\ref{lem:extendJMS}.
Let $s'_g$ and $s'_g(c)$ be the total connection cost and the connection cost of $c$ in $S'_g$, {\em according to the assignment $\mu$}. 

\begin{itemize}
\item If $c \in C^a_1$, $s'_g(c) \le \opt(c)$.
\item If $c \in C^b_1$, $s'_g(c) \le \opt(c)$.
\item If $c \in C_2$, we consider the facility $f^*$ serving
  $c$ in $\OPT$. Let $(\{f_2\}, B')$ be the group such that $f^* \in B'$. Since $B \neq B'$ we have
  that $f_2 \in S''_g$ and so $s'_g(c) \le \dist(c, f_2)$, which by triangle inequality
  is at most $\opt(c) + \dist(f_2,f^*)$.
  We have that \[\dist(f_2,f^*) \le \frac{1}{(1-\delta)|C_{\OPT}(f^*)|} \sum_{c' \in C_{\OPT}(f^*)} (\opt(c') + d'(c')).\]
\item If $c \in C_3$, we consider the facility $f^*$ serving
  $c$ in $\OPT$. Let $(\{f_2\}, B')$ be the group such that $f^* \in B'$. Since $B \neq B'$ we have
  that $f_2 \in S''_g$ and so $s'_g(c) \le \dist(c, f_2)$, which by triangle inequality
  is at most $\opt(c) + \dist(f_2,f^*)$. We have that 
  \[\dist(f^*, f_2) \le  \dist(f^*, f') \le  \opt(c) + d'(c).\]
\item If $c \in C_4$, we have $s'_g(c) \le d'(c)$.
\end{itemize}

We then apply Lemma~\ref{lem:extendJMS} ($S_0 \leftarrow S' - \{f'\} \cup \{g(B)\}$, $S^* \leftarrow S'_g$ and $\mu \leftarrow \mu$)
to get a guarantee on $S''_g = $\extendJMS($S' - \{f'\} \cup \{g(B)\}$).
Since clients of type $C^a_1$, $C_2$, $C_3$ and $C_4$ are served by
a facility in $S' - \{f'\} \cup \{g(B)\}$ in assignment $\mu$, Lemma~\ref{lem:extendJMS} implies
\begin{align*}
\cost(S''_g) \, \le \, \, &  \open(S' - \{f'\} \cup B) + \sum_{c \in C^a_1} \opt(c) + \sum_{c \in C_2}\bigg(\opt(c) +
\frac{\sum_{c' \in C_{\OPT}(\OPT(C))} \opt(c') + d'(c')}{(1-\delta)|C_{\OPT}(\OPT(c))|} \bigg) \\
& + \sum_{c \in C_3} \big( 2\opt(c) + d'(c) \big) + \sum_{c \in C_4} d'(c)+ \sum_{c \in C^b_1} 2 \opt(c).
\end{align*}
It follows that, since   $\cost(S') - \cost(S''_g) \le 0$, we have
\begin{align}
&  \open(f') - \open(B)  + 
    \sum_{c \in C^a_1} \big( \opt(c)  - d'(c)  \big) \nonumber \\
  & + \sum_{c \in C_2}\bigg(\opt(c) - d'(c) +
   \frac{\sum_{c' \in C_{\OPT}(\OPT(C))} \opt(c') + d'(c')}{(1-\delta)|C_{\OPT}(\OPT(c))|} \bigg)  
   + \sum_{c \in C_3} 2\opt(c) + \sum_{c \in C^b_1} \big( 2 \opt(c) - d'(c) \big) \nonumber \\
   &\le 0 \label{eq:heavy}
\end{align}

We now conclude the proof of the theorem.
Consider summing up the inequalities~\eqref{eq:light} and \eqref{eq:heavy} over all groups $(\{f'\}, B)$. 
First, observe that each facility of $\OPT$ and $S'$ appears in exactly one group.
Hence, the sum of the quantities $\open(f') - \open(B)$ is exactly $\open(S') - \open(\OPT)$.

Moreover, each client $c$ served by a facility $f^* \in \OPT^L$ (hence not matched) appears as a
client of type $C_1, C_1^a$ or $C_1^b$ in exactly one swap (i.e.: the swap $(\{f'\}, B)$ where
$f^* \in B$), as a client of type $C_3$ in at most one swap (i.e.: the swap $(\{f'\}, B)$ where
$f'$ is the facility serving it in $S'$) and as a type $C_4$ client in all remaining swaps.
Therefore, its total contribution to LHS of the sum is at most
$4\opt(c) - d'(c)$.

In addition, each client $c$ served by a facility $f^* \in \OPT^M$ and served by a facility
$f' \in S'^M$ such that $f^*$ and $f'$ are matched appears as a client of type $C_1$ or $C_1^a$
in exactly one swap and as a client of type $C_4$ otherwise. Its total contribution is thus
at most $\opt(c) - d'(c)$.

Each client $c$ served by a facility $f^* \in \OPT^M$ and served by a facility
$f' \in S'$ such that $f^*$ and $f'$ are not matched appears as a client of type $C_2$
in exactly one swap and as a client of type $C_4$ otherwise. Its total contribution is thus
at most $\opt(c) - d'(c) +  \frac{\sum_{c' \in C_{\OPT}(\OPT(c))} \opt(c') + d'(c')}{(1-\delta)|C_{\OPT}(\OPT(c))|}$.
Observe that for each matched facility $f^*$ of $\OPT^M$ the number of such clients is at most
$\delta |C_{\OPT}(\OPT(c))|$ by the definition of matched. Hence, summing up over all clients served
by $f^*$ in $\OPT$, the above bound is at most
\[\sum_{c \in C_{\OPT}(f^*)} \big( \opt(c) - d'(c) \big) + \frac{\delta}{1-\delta} \sum_{c \in C_{\OPT}(f^*)} \big(\opt(c) + d'(c) \big).\]
Therefore, the total cost of $S'$ is as claimed at most
\begin{equation*}\open(\OPT) + \frac{\delta}{1-\delta} d' + \frac{1}{1-\delta} \opt^M + 4 \opt^L. \qedhere \end{equation*}
\end{proof}

\subsection{Final Result by Cost Scaling} 
The result in the above subsection shows that for any Facility Location instance, if $\OPT$ is the solution minimizing $\open(\OPT) + (2 - \eta)opt$, $S'$ is a local optimum of the extended local search starting from a JMS solution, and $d' \leq 4 \opt$, then it is already an LMP $(2 - \eta)$ approximate solution for some $\eta > 0$ depending only on our choices of $\delta_1, \delta_2, \delta'_1, \delta'_2$. 
Though ensuring $d' \leq 4\opt$ in general is a nontrivial task, we show how to get around this bottleneck and get an LMP $(2 - \eta / 2)$ approximation. 
Let $\OPT$ be the feasible solution that minimizes $\open(\OPT) + (2 - \eta/2)opt$.

The idea is to use the concept of bipoint solutions again. 
In particular, we multiply each facility cost by a factor $\lambda$, and apply the JMS algorithm and the extended local search to the resulting Facility Location instance. Let $S(\lambda)$ be the corresponding solution. 

If $\lambda = 0$, $S(0)$ connects all clients to their closest facilities. If the original facility cost of $S(0)$ is at most that of $\OPT$, then $S(0)$ is already an LMP $1$ approximation, so we can assume otherwise. 
On the other hand, let $\gamma =\min( \min_{f \in F} \open(f), \min_{f, f' \in F:\open(f) \neq \open(f')} |\open(f) - \open(f')| )$, 
and $M = \sum_{c \in C, f \in F} d(c, f)$.
Suppose we scale facility costs by $\lambda = 3M / \gamma$, and consider an arbitrary solution $S'$ that consists of one facility with the lowest opening cost. 
By the definition of $\lambda$, for any solution $S''$ that strictly spends more on the facility cost than $S'$, whether it opens two cheapest facilities or opens a facility which is not the cheapest, the difference between the opening costs between $S'$ and $S''$ is at least $3M$, which is at least 3 times larger than the connection cost of any solution. 
The fact that $S(\lambda)$ is also LMP $2$ approximate w.r.t. $S'$ (with respect to scaled facility costs) implies that $S(\lambda)$ also opens exactly one cheapest facility, so its facility cost must be at most that of $\OPT$ (both scaled and original). 

We perform a binary search over $\lambda$ until we find a value $\lambda^*$ such that $S_1=S(\lambda^*+\delta)$ has facility cost (with respect to original facility costs) $\open(S_1) \leq \open(\OPT)$ and $S_2=S(\lambda^*)$ has facility cost (with respect to original facility costs) $\open(S_2) \geq \open(\OPT)$.  Here $\delta>0$ can be chosen to be exponentially small in the input size and we will next neglect it at the cost of an extra factor $(1+\eps)$ in the approximation for arbitrarily small constant $\eps > 0$. 
Also note that the value of $\open(\OPT)$ can be guessed at the cost of an extra $(1+\eps)$ in the approximation. 

Define $a\in [0,1]$ such that $a\open(S_1)+(1-a)\open(S_2)=\open(\OPT)$. Let $d_i=d(S_i)$. Since $S_1$ and $S_2$ are LMP $2$ approximate solutions with respect to scaled facility costs, we have 
$$
\lambda^* \open(S_1) + d_1  \leq \lambda^* \open(\OPT) + 2\opt,
\quad \lambda^* \open(S_2) + d_2  \leq \lambda^* \open(\OPT) + 2\opt,
$$
which implies 
$$
ad_1+(1-a)d_2\leq 2\opt.
$$
Suppose that $a \geq 1/2$. (The other case is symmetric.) 
This implies that $d_1 \leq 4\opt$, so we can Lemma~\ref{lem:facility_location_fixed_scale} so that $S_1$ is indeed LMP $(2 - \eta)$ approximate with respect to scaled facility costs for some $\eta \geq 4.5 \cdot 10^{-7}$; $\lambda^* \open(S_1) + d_1  \leq \lambda^* \open(\OPT) + (2 - \eta)\opt$. 
This implies that a random solution choosing $S_1$ with probability $a$ and $S_2$ with probability $(1 - a)$ has the expected (original) facility cost
$\open(\OPT)$ and the connection cost at most $a(2-\eta)\opt + (1-a)2\opt \leq (2 - \eta/2)\opt$. 
Therefore, one of $S_1$ and $S_2$ is an LMP $(2-\eta/2)$ approximate solution. This proves Theorem~\ref{thr:facility_location}.

%

\end{document}